\documentclass{article}%
\usepackage{amsfonts}
\usepackage{amsmath}
\usepackage{color}
\usepackage{fancybox}
\usepackage{appendix}%
\usepackage{amssymb}%
\usepackage[dvips]{graphicx}
\usepackage{graphicx}
\usepackage{psfrag}
\usepackage{mathtools}
\usepackage{tikzpagenodes}
\usetikzlibrary{tikzmark}
\usepackage{tikz}
\usetikzlibrary{shapes,snakes}
\usepackage{amsmath,amssymb}
\usepackage{verbatim}
\usepackage{gensymb}

\providecommand{\U}[1]{\protect\rule{.1in}{.1in}}
\newtheorem{theorem}{Theorem}
\newtheorem{acknowledgement}[theorem]{Acknowledgement}

\newtheorem{lemma}[theorem]{Lemma}

\newtheorem{proposition}[theorem]{Proposition}

\newenvironment{proof}[1][Proof]{\noindent\textbf{#1.} }{\ \rule{0.5em}{0.5em}}
\begin{document}

\title{\textbf{Gravitational redshift/blueshift of light emitted by geodesic test particles,  frame-dragging and pericentre-shift effects, in the Kerr-Newman-de Sitter and Kerr-Newman black hole geometries}}
\author{G. V. Kraniotis \footnote{email: gkraniotis@uoi.gr}\\
University of Ioannina, Physics Department \\ Section of
Theoretical Physics, GR- 451 10, Greece \\
}

 \maketitle

\begin{abstract}
We investigate the redshift and blueshift of light emitted by timelike geodesic particles in orbits around a Kerr-Newman-(anti) de Sitter (KN(a)dS) black hole.
Specifically we compute the redshift and blueshift of photons that are emitted by geodesic massive particles and travel along null geodesics towards a distant observer-located at a finite distance from the KN(a)dS black hole. For this purpose we use the Killing-vector formalism and the associated first integrals-constants of motion. We consider in detail stable timelike equatorial circular orbits of stars and express their corresponding redshift/blueshift in terms of the metric physical black hole parameters (angular momentum per unit mass, mass, electric charge and the cosmological constant) and the orbital radii of both the emitter star and the distant observer. These radii are linked through the constants of motion along the null geodesics followed by the photons since their emission until their detection and as a result we get closed form analytic expressions for the orbital radius of the observer in terms of the emitter radius, and the black hole parameters.
In addition, we compute exact analytic expressions for the frame dragging of timelike spherical orbits in the KN(a)dS spacetime in terms of multivariable generalised hypergeometric functions of Lauricella and Appell.  We apply our exact solutions of timelike non-spherical polar KN geodesics for the computation of frame-dragging, pericentre-shift, orbital period for the orbits of S2 and S14 stars within the  $1^{\prime\prime}$ of SgrA*.
We solve the conditions for timelike spherical orbits in KN(a)dS and KN spacetimes. We present new, elegant compact forms for the parameters of these orbits.
Last but not least we derive a very elegant and novel exact formula for the periapsis advance for a test particle in a non-spherical polar orbit in KNdS black hole spacetime in terms of Jacobi's elliptic function sn and Lauricella's hypergeometric function $F_D$.

\end{abstract}

\section{Introduction}

General
relativity (GR) \cite{Einstein} has triumphed all experimental tests so far which cover a wide range of field strengths and physical scales that include: those in large scale cosmology \cite{Supern},\cite{GVKSWB}, the prediction of
solar system effects like the perihelion precession of Mercury with a very high
precision \cite{Einstein},\cite{GKSWMERCURY}, the recent discovery of gravitational waves in Nature \cite{GW150914},\cite{GW151226}, as well as the observation of the shadow of the M87 black hole \cite{EHT}, see also  \cite{ClifWill}.

The orbits of short period stars in the central arcsecond  (S-stars)
of the Milky Way Galaxy provide the best current evidence
for the existence of supermassive black holes, in particular for the galactic centre SgrA* supermassive black hole \cite{GhezA},\cite{Eisenhauer},\cite{GenzelETAL}.

One of the most fundamental exact non-vacuum solutions of the
gravitational field equations of general relativity is the
Kerr-Newman black hole \cite{Newman}. The Kerr-Newman (KN) exact
solution describes the curved spacetime geometry surrounding a
charged, rotating black hole and it solves the coupled system of
differential equations for the gravitational and electromagnetic
fields \cite{Newman} (see also \cite{HansOhanian}).

 The KN exact solution generalised the Kerr
solution \cite{KerrR}, which describes the curved spacetime
geometry around a rotating black hole, to include a net electric
charge carried by the black hole.

A wide variety of astronomical and cosmological observations in the last two decades, including high-redshift type Ia supernovae, cosmic microwave background radiation and large scale structure indicate convincingly an accelerating expansion of the Universe \cite{Supern},\cite{Jones},\cite{Aubourg}. Such observational data  can be explained by a positive cosmological constant $\Lambda$ ($\Lambda>0$) with a magnitude $\Lambda\sim 10^{-56}{\rm cm}^{-2}$ \cite{GVKSWB}.
On the other hand, we note the significance of a negative cosmological constant in the anti-de Sitter/conformal field theories correspondence \cite{Maldacena}-\cite{Horowitz}.
Thus, solutions of the Einstein-Maxwell equations with a cosmological constant deserve attention.

A more realistic description of the spacetime geometry surrounding a black hole should include the cosmological
constant \cite{GKSWMERCURY},\cite{GVKraniotisMPI},
\cite{GRGKRANIOTIS},\cite{CQGKraniotis}
\cite{Claus},\cite{Zdenek},\cite{Sultana},\cite{Saheb}, \cite{Xu}.

Taking into account the contribution from the cosmological
constant $\Lambda,$ the generalisation of the Kerr-Newman solution
is described by the Kerr-Newman de Sitter $($KNdS$)$ metric
element which in Boyer-Lindquist (BL) coordinates is given by
\cite{BCAR},\cite{Stuchlik1},\cite{GrifPod},\cite{ZdeStu} (in units where $G=1$ and $c=1$):
\begin{align}
\mathrm{d}s^{2}  & =\frac{\Delta_{r}^{KN}}{\Xi^{2}\rho^{2}}(\mathrm{d}%
t-a\sin^{2}\theta\mathrm{d}\phi)^{2}-\frac{\rho^{2}}{\Delta_{r}^{KN}%
}\mathrm{d}r^{2}-\frac{\rho^{2}}{\Delta_{\theta}}\mathrm{d}\theta
^{2}\nonumber \\ &-\frac{\Delta_{\theta}\sin^{2}\theta}{\Xi^{2}\rho^{2}}(a\mathrm{d}%
t-(r^{2}+a^{2})\mathrm{d}\phi)^{2}%
\label{KNADSelement}
\end{align}%
\begin{equation}
\Delta_{\theta}:=1+\frac{a^{2}\Lambda}{3}\cos^{2}\theta,
\;\Xi:=1+\frac {a^{2}\Lambda}{3},
\end{equation}

\begin{equation}
\Delta_{r}^{KN}:=\left(  1-\frac{\Lambda}{3}r^{2}\right)  \left(  r^{2}
+a^{2}\right)  -2Mr+e^{2},
\label{DiscrimiL}
\end{equation}

\begin{equation}
\rho^{2}=r^{2}+a^{2}\cos^{2}\theta,
\end{equation}
where $a,M,e,$ denote the Kerr parameter, mass and electric charge
of the black hole, respectively.
The KN(a)dS metric is the most general exact stationary black hole solution of the Einstein-Maxwell system of differential equations.
This
is accompanied by a non-zero electromagnetic field
$F=\mathrm{d}A,$ where the vector potential is
\cite{GrifPod},\cite{ZST}:
\begin{equation}
A=-\frac{er}{\Xi(r^{2}+a^{2}\cos^{2}\theta)}(\mathrm{d}t-a\sin^{2}\theta
\mathrm{d}\phi).
\end{equation}

Properties of the Kerr-Newman spacetimes with a non-zero cosmological constant are appropriately described by their geodesic structure which determines the motion of test particles and photons.
The KN(a)dS dynamical system of geodesics is a
completely integrable system \footnote{This is proven by solving the
relativistic Hamilton-Jacobi equation by the method of separation of variables.}
as was shown in
\cite{BCAR},\cite{Stuchlik1},\cite{ZST},\cite{GVKraniotisMPI} \
and the geodesic differential equations take the form:
\begin{align}
\int\frac{\mathrm{d}r}{\sqrt{R^{\prime}}}  &
=\int\frac{\mathrm{d}\theta
}{\sqrt{\Theta^{\prime}}},\label{lensKN1}\\
\rho^{2}\frac{\mathrm{d}\phi}{\mathrm{d}\lambda}  & =-\frac{\Xi^{2}}%
{\Delta_{\theta}\sin^{2}\theta}\left(  aE\sin^{2}\theta-L\right)  +\frac
{a\Xi^{2}}{\Delta_{r}^{KN}}\left[  (r^{2}+a^{2})E-aL\right]  ,\label{azimueq}%
\\
c\rho^{2}\frac{\mathrm{d}t}{\mathrm{d}\lambda}  & =\frac{\Xi^{2}(r^{2}%
+a^{2})[(r^{2}+a^{2})E-aL]}{\Delta_{r}^{KN}}-\frac{a\Xi^{2}(aE\sin^{2}%
\theta-L)}{\Delta_{\theta}},\\
\rho^{2}\frac{\mathrm{d}r}{\mathrm{d}\lambda}  & =\pm\sqrt{R^{\prime}},\\
\rho^{2}\frac{\mathrm{d}\theta}{\mathrm{d}\lambda}  & =\pm\sqrt%
{\Theta^{\prime}},\label{polareq}%
\end{align}
where%
\begin{align}
R^{\prime}  & :=\Xi^{2}\left[  (r^{2}+a^{2})E-aL\right]  ^{2}-\Delta_{r}%
^{KN}(\mu^{2}r^{2}+Q+\Xi^{2}(L-aE)^{2})\label{sextic},\\
\Theta^{\prime}  & :=\left[  Q+(L-aE)^{2}\Xi^{2}-\mu^{2}a^{2}\cos^{2}%
\theta\right]  \Delta_{\theta}-\Xi^{2}\frac{\left(  aE\sin^{2}\theta-L\right)
^{2}}{\sin^{2}\theta}.
\end{align}
Null geodesics are derived by setting $\mu=0$. The proper time $\tau$ and the affine parameter $\lambda$ are connected by the relation $\tau=\mu\lambda$.
In the following we use geometrised units, $G=c=1$, unless it is stipulated otherwise. The first integrals of motion $E$ and $L$ are related to the isometries of the KNdS metric while $Q$ (Carter's constant) is the hidden integral of motion that results from the separation of variables of the Hamilton-Jacobi equation.

Although it is not the purpose of this paper to discuss how a net electric
charge is accumulated inside the horizon of the black hole, we briefly mention
recent attempts which address the issue of the formation of charged black
holes. Indeed, we note at this point, that the authors in \cite{Ray}, have
studied the effect of electric charge in compact stars assuming that the
charge distribution is proportional to the mass density. They found solutions
with a large positive net electric charge. From the local effect of the forces
experienced on a single charged particle, they concluded that each individual
charged particle is quickly ejected from the star. This is in turn produces
 a huge force imbalance,  in which the gravitational force overwhelms the
repulsive Coulomb and fluid pressure forces. In such a scenario the star
collapses to form a charged black hole before all the charge leaves the system
\cite{Ray}.  A mechanism for generating charge asymmetry that may be linked
to the formation of a charged black hole has been suggested in \cite{Cuesta}.
Besides these theoretical considerations for the formation of charged black holes, recent observations of structures near SgrA*  by the GRAVITY experiment, indicate possible presence of a small electric charge of central supermassive black hole \cite{Zajacek},\cite{Britzen}. Accretion disk physics around magnetised Kerr black holes under the influence of cosmic repulsion  is extensively discussed in the review \cite{universemdpi} \footnote{We also mention that supermassive black holes as possible sources of ultahigh-energy cosmic rays have been suggested in \cite{waldcharge}, where it has been shown that large values of the Lorentz $\gamma$ factor of an escaping ultrahigh-energy particle from the inner regions of the black hole accretion disk may occur only in the presence of the induced charge of the black hole.}.
Therefore, it is quite interesting to study the combined effect of the cosmological constant and electromagnetic fields on the black hole astrophysics.

Most of previous studies of the black hole backgrounds with $\Lambda>0$ were mainly concentrated on uncharged Schwarzschild-de Sitter and Kerr-de Sitter spacetimes  \cite{GKSWMERCURY},\cite{GVKraniotisMPI},\cite{GVKperiapsisadvLTprecession},\cite{KraniotisShadow},\cite{CQGKraniotis},\cite{Claus},\cite{Opava1},\cite{opav2}.
Particle motion and shadow of rotating spacetimes have been studied for the Kerr case in \cite{JMBardeenNY} for the braneworld Kerr-Newman case with a tidal charge in \cite{BraneworldKN}, \cite{RotatingBHbraneworld}, and for the Kerr-de Sitter case in \cite{LightescapeConeskds}.
In a series of papers we solved exactly timelike and null geodesics in Kerr and Kerr-(anti) de Sitter black hole  spacetimes \cite{GVKperiapsisadvLTprecession}, \cite{CQGKraniotis},\cite{GVKraniotisMPI}, and null geodesics and the gravitational lens equations in electrically charged rotating black holes in \cite{GRGKRANIOTIS}. We also computed in \cite{GVKperiapsisadvLTprecession} elegant closed form analytic solutions for  the general relativistic effects of periapsis advance, Lense-Thirring precession, orbital and Lense-Thirring periods and applied our solutions for calculating these GR-effects for the observed orbits of S-stars.
The shadow of the Kerr and charged Kerr black holes were computed in \cite{KraniotisShadow},\cite{CQGKraniotis} and \cite{GRGKRANIOTIS} respectively.

It is the main purpose of this work to derive new exact analytic solutions for timelike geodesic equations in the Kerr-Newman and Kerr-Newman-de Sitter spacetimes as well as to develop a theory of gravitational redshift/blueshift of photons emitted by geodesic particles in timelike orbits around a Kerr-Newman-(anti) de Sitter black hole. We derive  closed form analytic solutions for both spherical  orbits as well as for non-spherical orbits. A special class of orbits around a charged rotating black hole that we investigate in this paper are spherical orbits i.e. orbits with constant radii that are not necessarily confined to the equatorial plane.
We solve for the \textit{first time} the conditions for spherical timelike orbits in Kerr-Newman and Kerr-Newman-(anti) de Sitter spacetimes. As a result of this procedure we derive
elegant, compact forms for the parameters of these orbits  which are associated with the Killing vectors of these spacetimes.
For such spherical orbits we derive closed form analytic expressions in terms of special functions (multivariable hypergeometric functions of Appell-Lauricella) for frame dragging precession. We perform an analysis of the parameter space for such spherical orbits around a Kerr-Newman black hole and present many examples of such orbits.
We also derive the exact orbital solution for non-spherical timelike geodesics around a Kerr-Newman black hole, that cross the symmetry axis of the black hole, in terms of elliptic functions. Moreover, for such orbits we derive closed form analytic expressions for relativistic observables that include the pericentre shift, frame-dragging precession and the orbital period. We extend our analytic computations by deriving a very elegant analytic solution for the periapsis advance for a test particle in a non-spherical polar orbit around a Kerr-Newman-de Sitter black hole.

As we mentioned a very important motivation for our present work is to investigate frequency redshift/blueshift of light emitted by timelike geodesic particles  in orbits around a Kerr-Newman-(anti) de Sitter black hole using the Killing vector formalism.

Indeed, one of the targets of observational astronomers of the galactic centre is to measure the gravitational redshift predicted by the theory of general relativity \cite{zucker}.
In the Schwarzschild spacetime geometry the ratio of the frequencies measured by two stationary clocks at the radial positions $r_1$ and $r_2$ is given by \cite{HansOhanian}:
\begin{equation}
\frac{\nu_1}{\nu_2}=\frac{\sqrt{1-2GM/r_2}}{\sqrt{1-2GM/r_1}},
\end{equation}
where $G$ is the gravitational constant and $M$ is the mass of the black hole.
Recently, the redshift/blueshift of photons emitted by test particles in timelike circular equatorial orbits in Kerr spacetime were investigated in \cite{HERRERA}.

The analytic computation we perform in this work for the \textit{first time}, for the redshift and blueshift of light emitted by timelike geodesic particles in orbits around the Kerr-Newman-(anti) de Sitter (KN(a)dS) black hole,  extend our previous results on relativistic observables. The theory we develop enriches our arsenal for  studying the combined effect of the cosmological constant and electromagnetic fields on the black hole astrophysics and can serve as a method to infer important information about black hole observables from photon frequency shifts. In addition, we derive
new exact analytic expressions for the pericentre-shift and frame-dragging for non-spherical non-equatorial (polar) timelike KNdS and KN black hole orbits.
Moreover, we derive novel exact expressions for the frame dragging effect for particles in  spherical, non-equatorial orbits in KNdS and KN black hole geometries.
These results could be of interest to the observational astronomers of the Galactic centre \cite{GhezA},\cite{GenzelETAL} whose aim is to measure experimentally, the relativistic effects predicted by the theory of General Relativity \cite{GeorgeVKraniotis} for the  observed orbits of
short-period stars-the so called S-stars in our Galactic centre.
During 2018, the close proximity of the star S2 (S02) to the supermassive Galactic centre black hole allowed the first measurements of the relativistic redshift observable by the GRAVITY collaboration \cite{gravityCol} and the UCLA Galactic centre group whose astrometric measurements were obtained at the W.M. Keck Observatory \cite{Do}\footnote{Observational work is ongoing towards the detection of the periastron shift of the star S2 and the discovery of putative closer stars-in the central milliarcsecond of SgrA* supermassive black hole \cite{Waisberg}, which could allow an astrometric measurement of the black hole spin as envisaged e.g. in \cite{GVKperiapsisadvLTprecession}.}.

The material of the paper is organised as follows: In Sec. \ref{KillingFormKNdS} we consider the Killing vector formalism and the corresponding conserved quantities in Kerr-Newman-(anti) de Sitter spacetime. In Sec.\ref{FiIntegralsKNdSequatocirc} we consider  equatorial circular geodesics in KN(a)dS spacetime and derive novel expressions for the specific energy and specific angular momentum for test particles moving in such orbits, see equations (\ref{marvelOne}) and (\ref{marvelTwo}). Typical behaviour of these functions is displayed in Fig.\ref{EPlusLPlusCosmoCo}-Fig.\ref{Lneg3}.  Furthermore, we investigate the stability of such timelike circular equatorial geodesics in Kerr-Newman-de Sitter spacetime and derive a new condition that restricts their radii, namely the inequality (\ref{equilibriumKNdS}).  In Sec. \ref{grredblue}  we provide general expressions for the redshift/blueshift that emitted photons by massive particles experience while travelling along null geodesics towards an observer located far away from their source by making use of the Killing vector formalism. In Sec.\ref{erithrimplemetatopisi} we derive novel exact analytic expressions for the redshift/blueshift of photons for circular and equatorial emitter/detector orbits around the Kerr-Newman-(anti) de Sitter black hole-see equations (\ref{shiftred}) and (\ref{shiftblue}) respectively. In the procedure we take into account the bending of light due to the field of the Kerr-Newman-(anti)de Sitter black hole at the moment of detection by the observer.  We derive the corresponding frequency shifts for circular equatorial orbits in Kerr-de Sitter spacetime  in  Sec. \ref{kdsredblue}. We also examine the particular case when the detector is located far away from the source.
In Sec.\ref{nonequatorial} we study non-equatorial orbits in rotating charged black hole spacetimes. Specifically, we compute in closed analytic form the frame-dragging for test particles in  timelike spherical orbits in the  Kerr-Newman and Kerr-Newman-de Sitter black hole spacetimes-equations (\ref{GVKframePrecessionKN}),theorem \ref{LenseThirFDRAG} and (\ref{LTlambdaKNdS}) respectively. The former equation (KN case) involves the ordinary Gau$\ss$  hypergeometric function and Appell's $F_1$ two-variable hypergeometric function, while the latter (KNdS case) is expressed in terms of Lauricella's $F_D$ and Appell's $F_1$ generalised multivariate hypergeometric functions \cite{Appell}. We solve the conditions for spherical orbits in Kerr Newman spacetime and derive \textit{new} elegant compact forms for the particle's energy and angular momentum about the $\phi$-axis, Theorem \ref{EndiamesoORO}, and relations (\ref{NESPenergyp}),(\ref{AngulMomenNESPher}). We investigated in some detail the ranges of $r$ and Carter's constant $Q$ for which the solutions (\ref{NESPenergyp}),(\ref{AngulMomenNESPher}) are valid. Having solved analytically the geodesic equations and armed with Theorems \ref{LenseThirFDRAG} and \ref{EndiamesoORO},  we present several examples of timelike spherical orbits around a Kerr-Newman black hole in Tables \ref{ProgradeKNnp} and \ref{retrogradeKNnpOr}.
In subsection \ref{asteraspole} we investigate an important subset of spherical orbits: polar spherical orbits i.e. timelike geodesics with constant coordinate radii crossing the symmetry axis of the Kerr-Newman spacetime. We perform an effective potential analysis and determine simplified forms for the physical parameters of such orbits: relations (\ref{energiapolars}) and (\ref{consCartpols}). We solve analytically the geodesic equations and derive closed analytic form expression for the Lense-Thirring precession, eqn.(\ref{ltfdpolarsphekn}). Moreover we present several examples of such orbits in Table \ref{TomiAxonaSymmetrias}. Remarkably, In Theorem \ref{OroTheoLambda} we have solved the conditions for timelike spherical orbits for the constants of motion $E,L$ that are associated with the two Killing vectors of the Kerr-Newman-(anti) de Sitter black hole spacetime. The original elegant relations we derive, equations (\ref{GENERALESPHE}),(\ref{LGENERALSPHERE}), represent the most general forms for these parameters of timelike spherical geodesics around the fundamental Kerr-Newman(anti)de Sitter black hole: they have as limits eqns (\ref{marvelOne}),(\ref{marvelTwo}) and eqns (\ref{NESPenergyp}),((\ref{AngulMomenNESPher}). In Theorem \ref{LambdaUniverseDomin} of section \ref{morecosmos} we derive a closed form analytic solution for frame dragging for a timelike polar spherical orbit in Kerr-Newman de Sitter spacetime, equation (\ref{CosmoLenseThirringMeta}). The solution is written in terms of Lauricella's hypergeometric function $F_D$ and Appell's $F_1$.
In Sec. \ref{nonspherepLTprece} and \ref{troxiakiPeriodPol} we derive new exact solutions for the Lense-Thirring precession and the orbital period for a test particle in a polar non-spherical orbit around a Kerr-Newman black hole, Eqns(\ref{LTprecessionKNpolarNonsphere}) and (\ref{PolarPeriodOrbitKN}) respectively. In Sec.\ref{periapsisKN} $\&$ \ref{periapsisLambdaKNpolar} we derive new closed form analytic expressions for the periapsis advance for test particles in non-spherical polar orbits in KN and KNdS spacetimes respectively. In the latter case we derive a novel,  very elegant, exact formula in terms of Jacobi's elliptic function $\rm {sn}$ and Lauricella's hypergeometric function $F_D$ of three variables-see equation (\ref{remarkablePeriapsisKNdS}). The exact results of sections \ref{nonspherepLTprece},\ref{troxiakiPeriodPol},\ref{periapsisKN}, are applied for the computation of the relativistic effects for frame-dragging, pericentre shift and the orbital periods for the observed orbits of stars S2 and S14 in the central arcsecond of Milky way, assuming that the galactic centre supermassive black hole is a Kerr-Newman black hole.

\section{Particle orbits and Killing vectors formalism in Kerr-Newman-(anti)de Sitter spacetime}\label{KillingFormKNdS}

From the condition for the invariance of the metric tensor:
\begin{equation}
0=\xi^{\alpha}\frac{\partial g_{\mu\nu}}{\partial x^{\alpha}}+
\frac{\partial \xi^{\alpha}}{\partial x^{\mu}}g_{\alpha\nu}+
\frac{\partial \xi^{\beta}}{\partial x^{\nu}}g_{\mu\beta},
\end{equation}
under the infinitesimal transformation:
\begin{equation}
x^{\prime\alpha}=x^{\alpha}+\epsilon\xi^{\alpha}(x),\;\;\text{with}\;\;\epsilon\rightarrow 0
\end{equation}
it follows that whenever the metric is independent of some coordinate a constant vector in the direction of that coordinate is a Killing vector . Thus the generic metric:
\begin{equation}
{\rm d}s^2=g_{tt}{\rm d}t^2+2g_{t\phi}{\rm d}t{\rm d}\phi+
g_{\phi\phi}{\rm d}\phi^2+g_{rr}{\rm d}r^2+g_{\theta\theta}{\rm d}\theta^2,
\end{equation}
possesses two commuting Killing vector fields:
\begin{align}
\xi^{\mu}&=(1,0,0,0) \;\;\text{timelike Killing vector},\label{killing1}\\
\psi^{\mu}&=(0,0,0,1)\label{KILLING2}\;\;\text{rotational Killing vector.}
\end{align}

According to Noether's theorem to every continuous symmetry of a physical system corresponds a conservation law. In a general curved spacetime, we can formulate the conservation laws for the motion of a particle on the basis of Killing vectors. We can prove that if $\xi_{\nu}$ is a Killing vector, then for a particle moving along a geodesic, the scalar product of this Killing vector and the momentum $P^{\nu}=\mu\frac{{\rm d}x^{\nu}}{{\rm d}\tau}$ of the particle is a constant \cite{HansOhanian}:
\begin{equation}
\xi_{\nu}P^{\nu}=\text{\rm constant}
\label{ConservationLaw}
\end{equation}

Due to the existence of these Killing vector fields (\ref{killing1}),(\ref{KILLING2}) there are two conserved quantities the total energy and the angular momentum per unit mass at rest of the test particle \footnote{As we mentioned already in the introduction, the charged Kerr solution possesses another hidden constant, Carter's constant $Q$. Alternatively, the complete integrability of the geodesic equations in KN(a)dS spacetime can be understood as follows: The Kerr-Newman family of spacetimes possesses in addition to the two Killing vectors a Killing tensor field $K_{\alpha\beta}$. This tensor can be expressed in terms of null tetrads (e.g. see eqn (7) in \cite{KraniotisDirac}) which implies the existence of a constant of motion $K=K_{\mu\nu}U^{\mu}U^{\nu}$. This is related to Carter's constant by: $K\equiv Q+(L-aE)^2$. See also \cite{WalkPen},\cite{Chandrasekhar}.}:
\begin{align}
E&=\frac{\tilde{E}}{\mu}=g_{\mu\nu}\xi^{\mu}U^{\nu}
=g_{tt}U^t+g_{t\phi}U^{\phi},\\
L&=\frac{\tilde{L}}{\mu}=-g_{\mu\nu}\psi^{\mu}U^{\nu}
=-g_{\phi t}U^t-g_{\phi\phi}U^{\phi}.
\end{align}

Thus, the photon's emitter is a probe massive test particle which geodesically moves around a rotating electrically charged cosmological black hole in the spacetime with a four-velocity:

\begin{equation}
U^{\mu}_e=(U^t,U^r,U^{\theta},U^{\phi})_e.
\end{equation}

The conservation law (\ref{ConservationLaw}) also applies to photon moving in the curved spacetime. Thus, if the spacetime geometry is time independent, the photon energy $P_0$ is constant. In section \ref{erithrimplemetatopisi} we will extract the redshift/blueshift of photons from this conservation law.

\section{Equatorial circular orbits in Kerr-Newman spacetime with a cosmological constant}\label{FiIntegralsKNdSequatocirc}
It is convenient to introduce a dimensionless cosmological parameter:
\begin{equation}
\Lambda^{\prime}=\frac{1}{3}\Lambda M^2,
\end{equation}
and set $M=1$. For equatorial orbits  Carter's constant $Q$ vanishes. For the following discussion, it is useful to introduce new constants of motion, the specific energy and specific angular momentum:
\begin{align}
\hat{E}&\equiv \frac{\Xi E}{\mu},\\
\hat{L}&\equiv\frac{\Xi L}{\mu}.
\end{align}
This is equivalent to setting $\mu=1$. Thus for reasons of notational simplicity we omit the caret for the specific energy and specific angular momentum in what follows.

Equatorial circular orbits correspond to local extrema of the effective potential. Equivalently, these orbits are given by the conditions $R^{\prime}(r)=0,{\rm d} R^{\prime}/{\rm d} r=0$, which have to be solved simultaneously. Following this procedure, we obtain the following novel equations for the specific energy and the specific angular momentum of test particles moving along equatorial circular orbits in KN(a)dS spacetime:
\begin{equation}
E_{\pm}(r;\Lambda^{\prime},a,e)=
\frac{e^2+r(r-2)-r^2(r^2+a^2)\Lambda^{\prime}\pm a
\sqrt{r^4\left(\frac{1}{r^3}-\Lambda^{\prime}\right)-e^2}}{r
\sqrt{2e^2+r(r-3)-a^2 r^2 \Lambda^{\prime}\pm 2 a\sqrt{r^4\left(\frac{1}{r^3}-\Lambda^{\prime}\right)-e^2}}},
\label{marvelOne}
\end{equation}
\begin{equation}
L_{\pm}(r;\Lambda^{\prime},a,e)=
\frac{\pm (r^2+a^2)\sqrt{r^4\left(\frac{1}{r^3}-\Lambda^{\prime}\right)-e^2}-2ar-
ar^2\Lambda^{\prime}(r^2+a^2)+ae^2}{r
\sqrt{2e^2+r(r-3)-a^2 r^2 \Lambda^{\prime}\pm 2 a\sqrt{r^4\left(\frac{1}{r^3}-\Lambda^{\prime}\right)-e^2}}}
\label{marvelTwo}
\end{equation}
The upper sign in (\ref{marvelTwo})-(\ref{marvelOne}) corresponds to the parallel orientation of particle's angular momentum $L$ and black hole spin $a$ (corotation-prograde motion), the lower to the antiparallel one (counter-rotation or retrograde motion).
The reality conditions connected with equations (\ref{marvelOne}) and (\ref{marvelTwo}) are given by the inequalities:
\begin{align}
&2e^2+r(r-3)-a^2r^2\Lambda^{\prime}\pm 2a\sqrt{r^4\left(\frac{1}{r^3}-\Lambda^{\prime}\right)-e^2}\geq0\label{reality1}\\
&\Leftrightarrow \frac{2e^2}{r^2}+\frac{r-3}{r}-a^2\Lambda^{\prime}
\pm 2a\sqrt{\frac{1}{r^3}-\Lambda^{\prime}-\frac{e^2}{r^4}}\geq0\label{reality2},
\end{align}
and \footnote{We note that for zero electric charge and $\Lambda>0$, inequality (\ref{staticradius}) defines the concept of static radius $r_{s}$ at which, the gravitational attraction caused by the central mass is just compensated by the cosmic repulsion due to cosmological constant. It has been shown that the static radius gives an upper limit on an extension of disk-like structures around black holes, coinciding with dimensions of large galaxies with central supermassive black holes \cite{STATICRADIUSzs},\cite{RADIUSBEFOREL}. The static radius is relevant also for the polytropic structures that could model dark matter halos, because it represents the upper limit on their extension, as shown in \cite{polytropicZS}. Moreover, it is relevant also for motion around galaxies, as discussed e.g. in \cite{MagellanosZSSCN},\cite{Boonserm},\cite{Faraoni}. On the other hand, see \cite{CarreraG} for arguments about the more likely detection of $\Lambda$ in a galaxy which is in the Hubble flow. }
\begin{equation}
1-\Lambda^{\prime}r^3\geq e^2/r.
\label{staticradius}
\end{equation}

For zero electric charge $e=0$ equations (\ref{marvelOne}),(\ref{marvelTwo}) reduce correctly to those in Kerr-anti de Sitter (KadS) spacetimes \cite{Opava1}. For zero electric charge and zero cosmological constant ($\Lambda=e=0$) equations (\ref{marvelOne}),(\ref{marvelTwo}) reduce correctly to the corresponding ones in Kerr spacetime \cite{Bardeen}. Relations (\ref{marvelOne}),(\ref{marvelTwo}) have also been discovered independently in \cite{plaglogo}.

In the figures \ref{EPlusLPlusCosmoCo}-\ref{Lneg3}, for concrete values of the electric charge and the cosmological parameter we present the radial dependence of the specific energy and specific angular momentum for various values of the black hole's spin. For the cosmological parameter $\Lambda^{\prime}$ we choose the values $\Lambda^{\prime}=10^{-5},10^{-4},10^{-3}$ as well as their negative counterparts. For stellar mass black holes, and positive cosmological constant this corresponds to $\Lambda\sim 10^{-15}{\rm cm}^{-2}-10^{-13}{\rm cm}^{-2}$.
For supermassive black holes such as at the centre of Galaxy M87 with mass $M^{M87}_{\rm BH}=6.7\times 10^9$ solar masses \cite{EHT} the value of $\Lambda^{\prime}=10^{-5}$ corresponds to the value for the cosmological constant: $\Lambda=3.06\times 10^{-35}{\rm cm}^{-2}$. The physical relevance of the graphs in figures \ref{EPlusLPlusCosmoCo}-\ref{Lneg3} lies in the following fact: the local extrema of the radial profiles of the specific energy and specific angular momentum of particles on equatorial circular orbits given by the relations (\ref{marvelOne})and (\ref{marvelTwo}) correspond to the marginally stable orbits that we shall discuss later.
The apparent singularity that appears in some graphs for negative cosmological constant e.g.
in figures \ref{PlusEneg},\ref{PlusEnegAM} is due to fact that the quantity inside the larger root in the denominator of the radial profiles in relations (\ref{marvelOne}),(\ref{marvelTwo}), and for a range of radii values, becomes negative and therefore the constants of motion become complex numbers.
In figures \ref{DEPlusKNadScompKERRex},\ref{DEPlusKNadScompKERR},\ref{ProgradeLKerr},
\ref{DLPlusKNadScompKERR} we plot the radial profiles for the direct constants of motion for the KN(a)dS black hole together with the profiles that correspond to the Kerr case ($e=\Lambda=0$) in order appreciate the electric charge and cosmological constant contributions.

Horizons of the KN(a)dS geometries are given by the condition $\Delta_r^{KN}=0$, which determines pseudosingularities of the spacetime interval (\ref{KNADSelement}). This condition yields the relations:
\begin{equation}
\Lambda^{\prime}=\Lambda^{\prime}_h(r,a,e)\equiv\frac{r^2+a^2+e^2-2r}{r^2(r^2+a^2)}
\end{equation}
Figures \ref{HorizonsposLambda},\ref{HorizonsnegLambda} exhibit typical behaviour of these functions for positive and negative cosmological constant respectively. For the KNdS black hole spacetime there are three horizons (event, apparent or Cauchy, cosmological) while for the KNadS spacetime there are two horizons (event,apparent).

\begin{figure}
[ptbh]
\psfrag{kinderleoforouausgung13}{$e=0.11,\;\Lambda^{\prime}=0.001$}
\psfrag{ra}{$r$} \psfrag{EpLp1}{$E_{+}$}
\psfrag{a9939}{$a=0.9939$}
\psfrag{a79}{$a=0.7939$}
\psfrag{a6}{$a=0.6$}
\psfrag{a52}{$a=0.52$}
\psfrag{a2}{$a=0.2$}
\begin{center}
\includegraphics[height=2.4526in, width=3.3797in ]{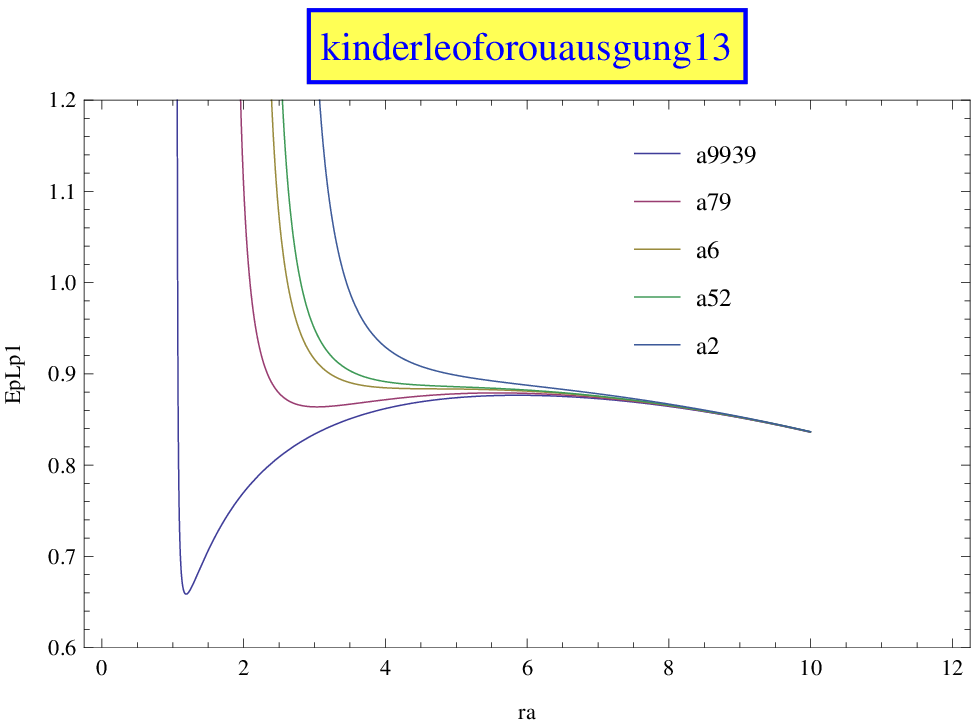}
 \caption{Specific energy $E_{+}$ for $e=0.11, \Lambda^{\prime}=0.001$ for different values for the Kerr parameter.}%
\label{EPlusLPlusCosmoCo}%
\end{center}
\end{figure}

\begin{figure}
[ptbh]
\psfrag{Enverdeleoforosausgung13}{$e=0.11,\;\Lambda^{\prime}=10^{-4}$}
\psfrag{ra}{$r$} \psfrag{EpLp}{$E_{+}$}
\psfrag{a9939}{$a=0.9939$}
\psfrag{a79}{$a=0.7939$}
\psfrag{a6}{$a=0.6$}
\psfrag{a52}{$a=0.52$}
\psfrag{a2}{$a=0.2$}
\begin{center}
\includegraphics[height=2.4526in, width=3.3797in ]{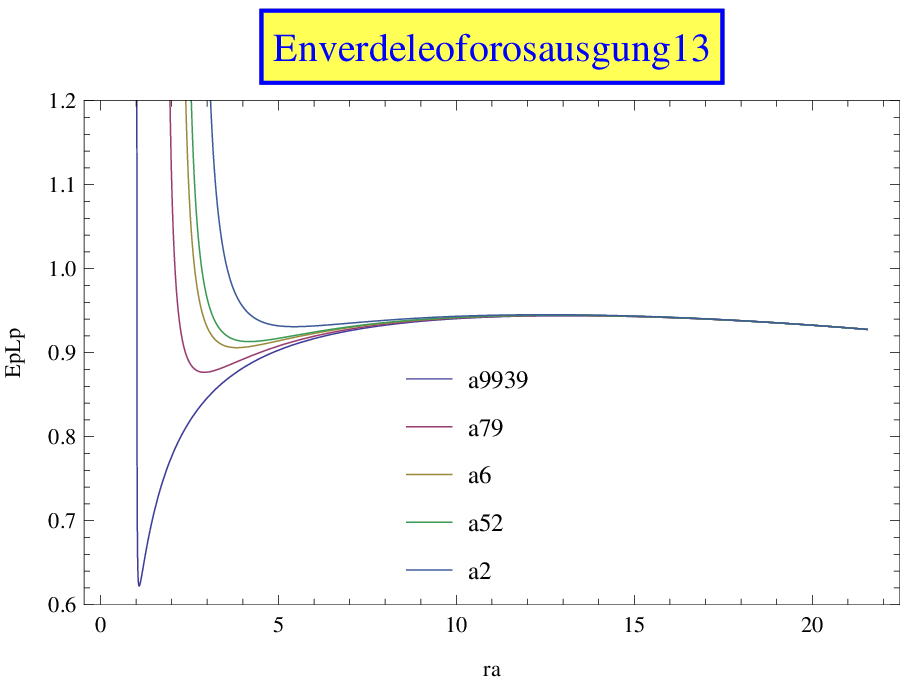}
 \caption{Specific energy $E_{+}$ for $e=0.11, \Lambda^{\prime}=0.0001$ for different values for the Kerr parameter.}%
\label{EPlusLPlusLambda}%
\end{center}
\end{figure}

\begin{figure}
[ptbh]
\psfrag{trifilileoforouausgung13}{$e=0.11,\;\;\Lambda^{\prime}=10^{-5}$}
\psfrag{ra}{$r$} \psfrag{EpLp3}{$E_{+}$}
\psfrag{a9939}{$a=0.9939$}
\psfrag{a79}{$a=0.7939$}
\psfrag{a6}{$a=0.6$}
\psfrag{a52}{$a=0.52$}
\psfrag{a2}{$a=0.2$}
\begin{center}
\includegraphics[height=2.4526in, width=3.3797in ]{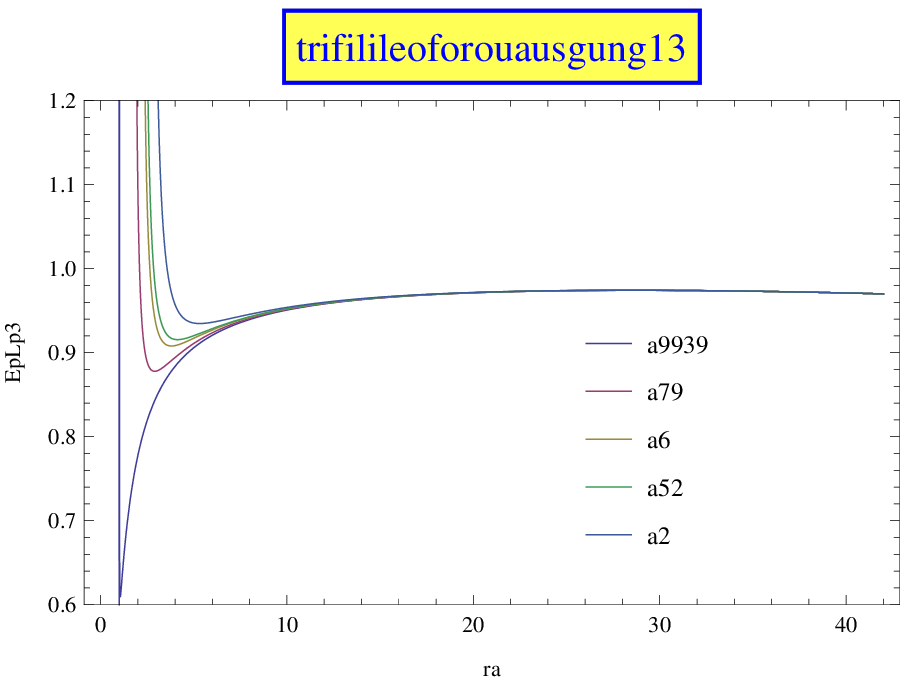}
 \caption{Specific energy $E_{+}$ for $e=0.11, \Lambda^{\prime}=0.00001$ for different values for the Kerr parameter.}%
\label{EPlusLPlusLambda}%
\end{center}
\end{figure}

\begin{figure}
[ptbh]

\psfrag{r}{$r$} \psfrag{Eplus}{$E_{+}$}
\psfrag{a052Lpm3e0}{$a=0.52,\Lambda^{\prime}=10^{-3},e=0$}
\psfrag{a052kerronly}{$a=0.52,\Lambda^{\prime}=0,e=0$}
\psfrag{a052L0e085}{$a=0.52,\Lambda^{\prime}=0,e=0.85$}
\psfrag{a052Lnm3e0}{$a=0.52,\Lambda^{\prime}=-10^{-3},e=0$}

\begin{center}
\includegraphics[height=2.4526in, width=3.3797in ]{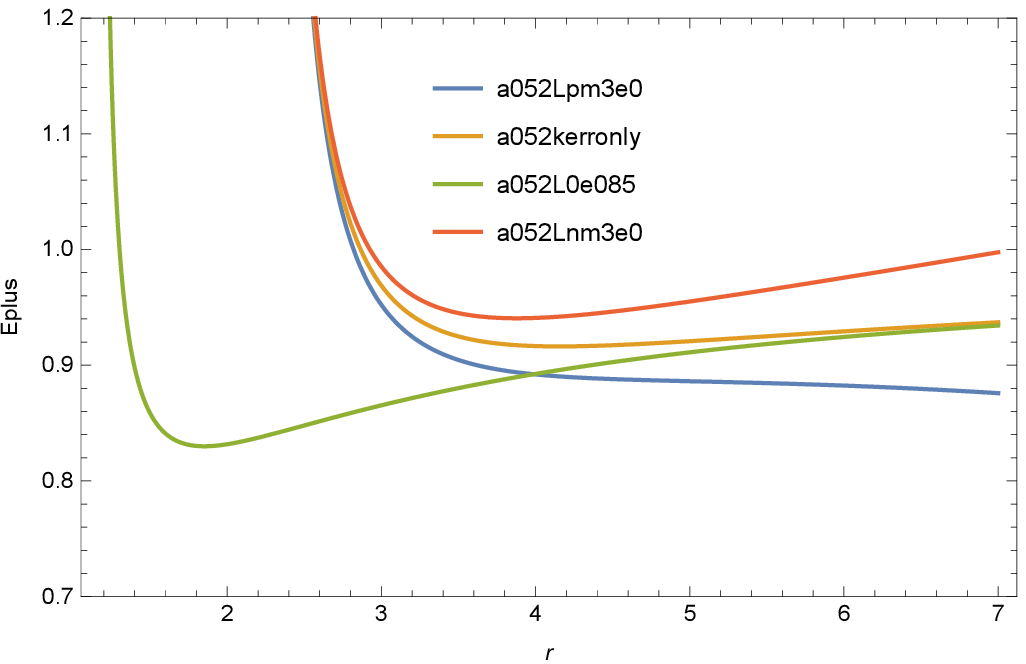}
 \caption{Radial profile for specific energy $E_{+}$ of test particles moving on equatorial circular orbits in KN(a)dS black hole with Kerr parameter $a=0.52$, and various values for the cosmological constant and the black hole's electric charge. The case of the Kerr black hole $e=\Lambda^{\prime}=0$ is displayed.}%
\label{DEPlusKNadScompKERR}%
\end{center}
\end{figure}

\begin{figure}
[ptbh]

\psfrag{r}{$r$} \psfrag{Eplus}{$E_{+}$}
\psfrag{a09939Lpm3e0}{$a=0.9939,\Lambda^{\prime}=10^{-3},e=0$}
\psfrag{a09939kerronly}{$a=0.9939,\Lambda^{\prime}=0,e=0$}
\psfrag{a09930L0e011}{$a=0.9939,\Lambda^{\prime}=0,e=0.11$}
\psfrag{a09939Lnm3e0}{$a=0.9939,\Lambda^{\prime}=-10^{-3},e=0$}

\begin{center}
\includegraphics[height=2.4526in, width=3.3797in ]{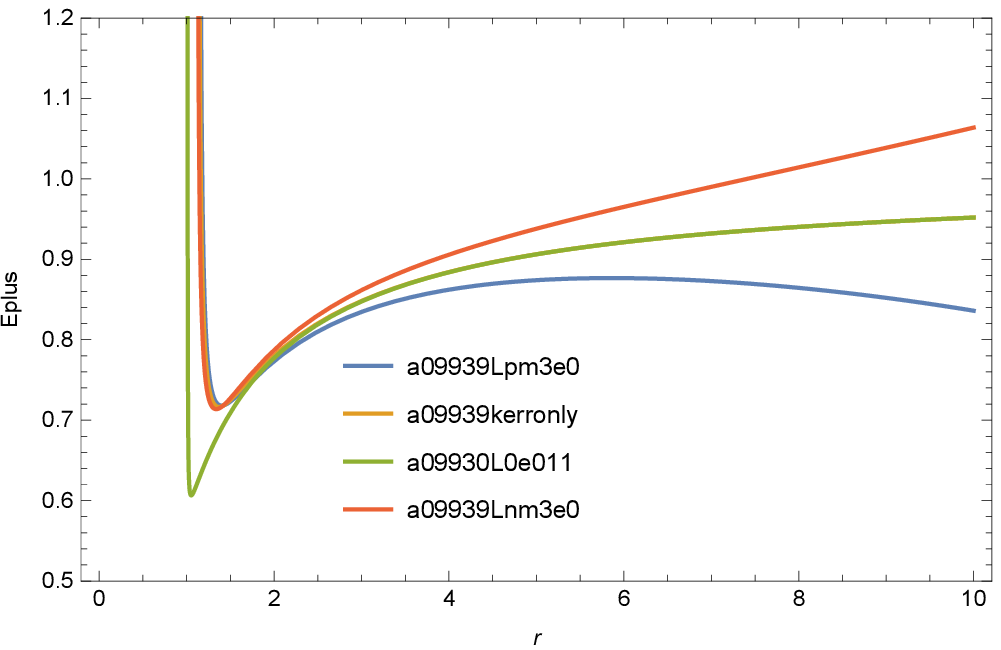}
 \caption{Radial profile for specific energy $E_{+}$ of test particles moving on equatorial circular orbits in KN(a)dS black hole with Kerr parameter $a=0.9939$, and various values for the cosmological constant and the black hole's electric charge. The case of the Kerr black hole $e=\Lambda^{\prime}=0$ is displayed.}%
\label{DEPlusKNadScompKERRex}%
\end{center}
\end{figure}

\begin{figure}
[ptbh]
\psfrag{kinderlalexandrasausgung13}{$e=0.11,\;\;\Lambda^{\prime}=10^{-3}$}
\psfrag{ra}{$r$} \psfrag{EnLp1}{$E_{-}$}
\psfrag{a9939}{$a=0.9939$}
\psfrag{a79}{$a=0.7939$}
\psfrag{a6}{$a=0.6$}
\psfrag{a52}{$a=0.52$}
\psfrag{a2}{$a=0.2$}
\begin{center}
\includegraphics[height=2.4526in, width=3.3797in ]{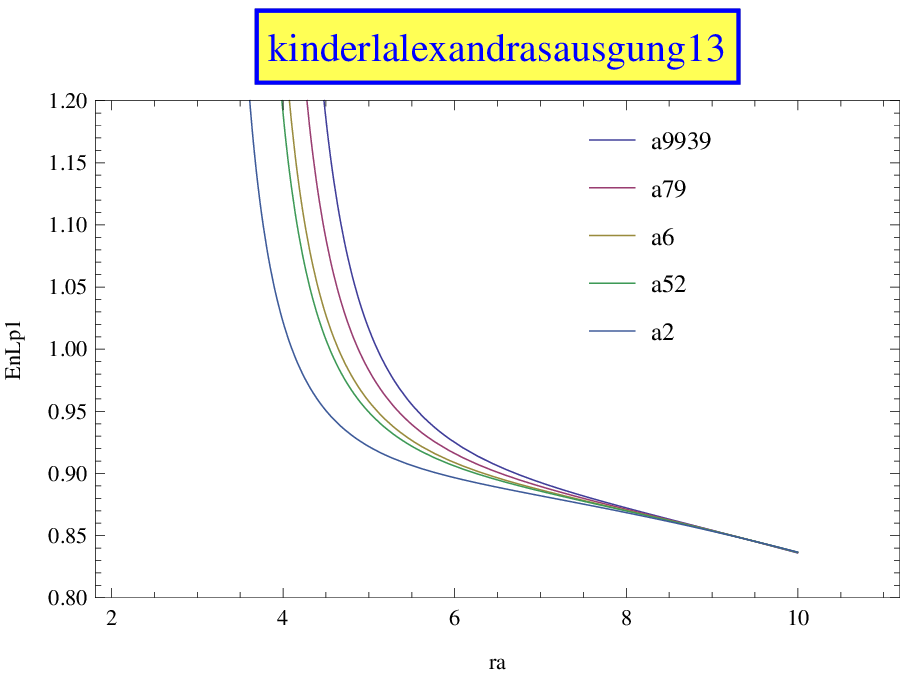}
 \caption{Specific energy $E_{-}$ for $e=0.11, \Lambda^{\prime}=0.001$ for different values for the Kerr parameter.}%
\label{EnegLPlusLambda}%
\end{center}
\end{figure}

\begin{figure}
[ptbh]
\psfrag{leofalexandrasausgung13}{$e=0.11,\;\;\Lambda^{\prime}=10^{-4}$}
\psfrag{ra}{$r$} \psfrag{EnLp2}{$E_{-}$}
\psfrag{a9939}{$a=0.9939$}
\psfrag{a79}{$a=0.7939$}
\psfrag{a6}{$a=0.6$}
\psfrag{a52}{$a=0.52$}
\psfrag{a2}{$a=0.2$}
\begin{center}
\includegraphics[height=2.4526in, width=3.3797in ]{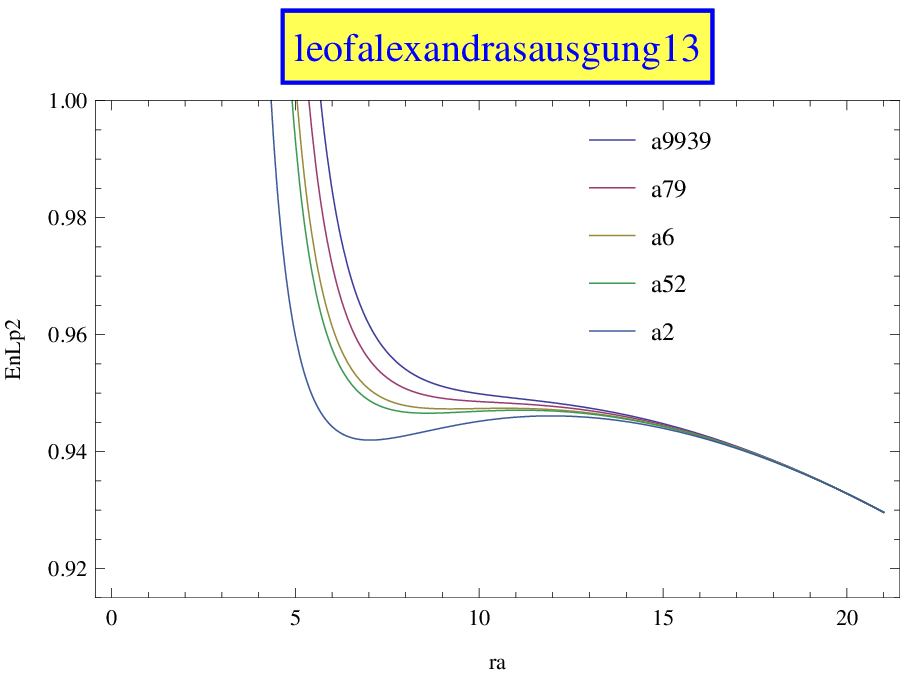}
 \caption{Specific energy $E_{-}$ for $e=0.11, \Lambda^{\prime}=0.0001$ for different values for the Kerr parameter.}%
\label{EnegPlusLambda}%
\end{center}
\end{figure}

\begin{figure}
[ptbh]
\psfrag{leofalexandrasgate13}{$e=0.11,\;\;\Lambda^{\prime}=10^{-5}$}
\psfrag{ra}{$r$} \psfrag{EnLp3}{$E_{-}$}
\psfrag{a9939}{$a=0.9939$}
\psfrag{a79}{$a=0.7939$}
\psfrag{a6}{$a=0.6$}
\psfrag{a52}{$a=0.52$}
\psfrag{a2}{$a=0.2$}
\begin{center}
\includegraphics[height=2.4526in, width=3.3797in ]{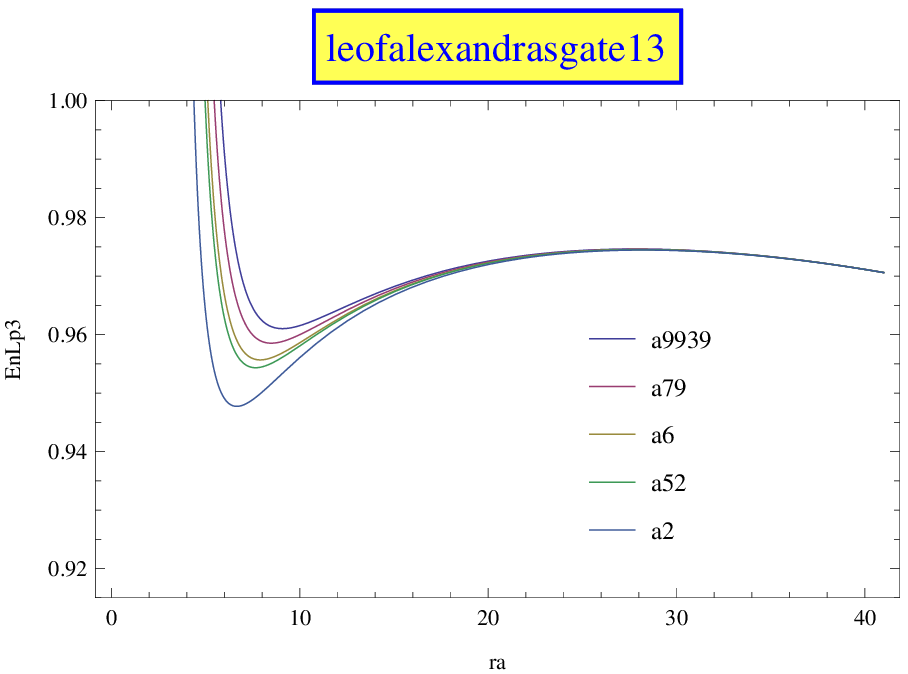}
 \caption{Specific energy $E_{-}$ for $e=0.11, \Lambda^{\prime}=0.00001$ for different values for the Kerr parameter.}%
\label{EnegNEGpLambda3}%
\end{center}
\end{figure}

\begin{figure}
[ptbh]
\psfrag{verdefunspaoausgung13}{$e=0.11,\Lambda^{\prime}=0.001$}
\psfrag{ra}{$r$} \psfrag{Lp}{$L_{+}$}
\psfrag{a9939}{$a=0.9939$}
\psfrag{a79}{$a=0.7939$}
\psfrag{a6}{$a=0.6$}
\psfrag{a52}{$a=0.52$}
\psfrag{a26}{$a=0.2$}

\begin{center}
\includegraphics[height=2.4526in, width=3.3797in ]{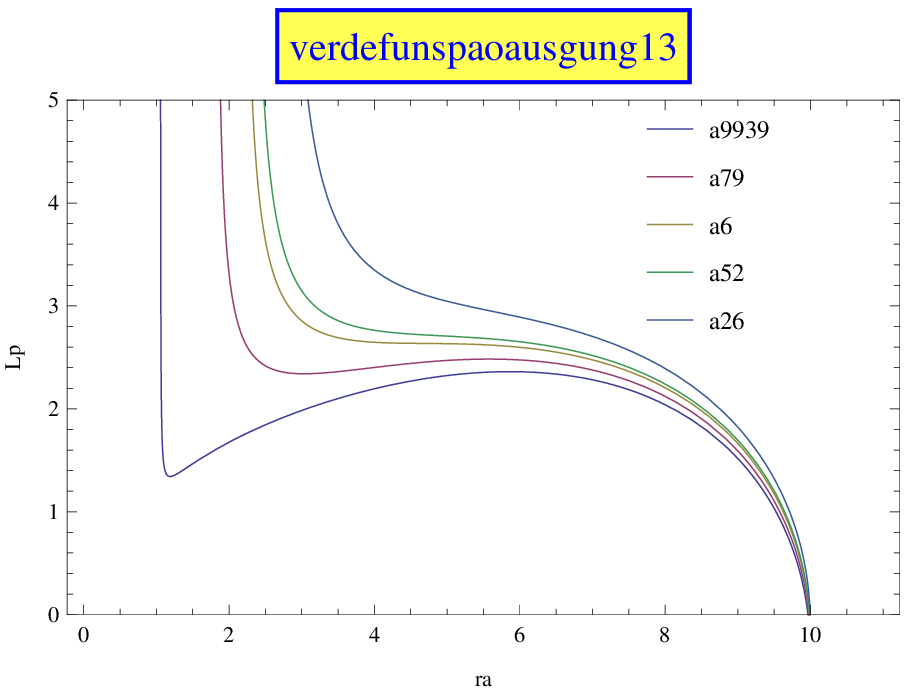}
 \caption{Specific angular momentum $L_{+}$ for $e=0.11, \Lambda^{\prime}=0.001$ for different values for the Kerr parameter.}%
\label{PlusLPlusLammbda}%
\end{center}
\end{figure}

\begin{figure}
[ptbh]
\psfrag{verdeleoforosausgung13}{$e=0.11,\;\Lambda^{\prime}=0.0001$}
\psfrag{ra}{$r$} \psfrag{Lp}{$L_{+}$}
\psfrag{a9939}{$a=0.9939$}
\psfrag{a79}{$a=0.7939$}
\psfrag{a6}{$a=0.6$}
\psfrag{a52}{$a=0.52$}
\psfrag{a2}{$a=0.2$}

\begin{center}
\includegraphics[height=2.4526in, width=3.3797in ]{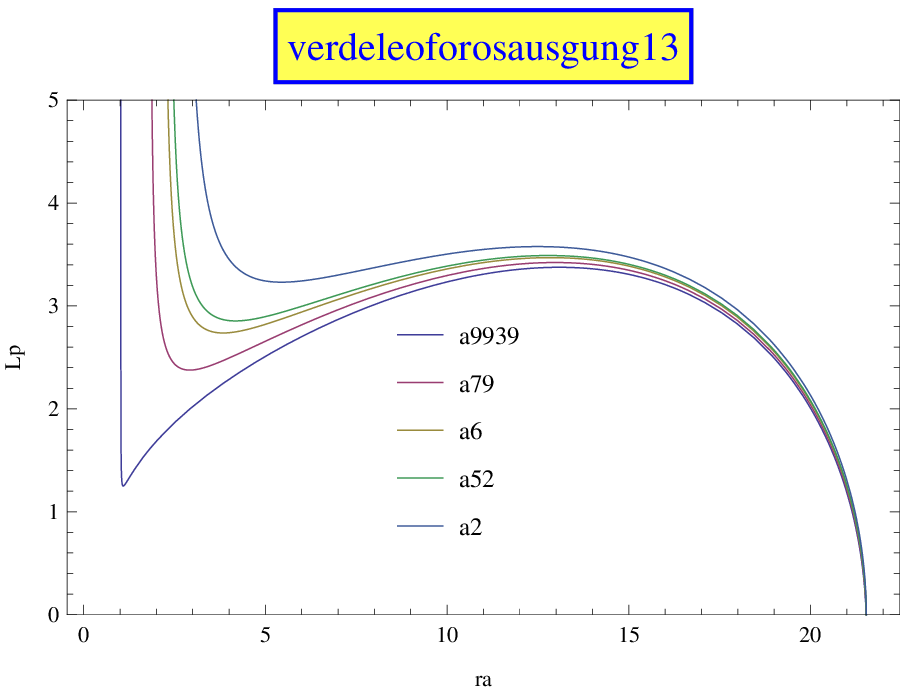}
 \caption{Specific angular momentum $L_{+}$ for $e=0.11, \Lambda^{\prime}=0.0001$ for different values for the Kerr parameter.}%
\label{PlusLPlusLambda}%
\end{center}
\end{figure}

\begin{figure}
[ptbh]
\psfrag{grueneleoforosausgung13}{$e=0.11,\;\;\Lambda^{\prime}=10^{-5}.$}
\psfrag{ra}{$r$} \psfrag{Lp}{$L_{+}$}
\psfrag{a9939}{$a=0.9939$}
\psfrag{a79}{$a=0.7939$}
\psfrag{a6}{$a=0.6$}
\psfrag{a52}{$a=0.52$}
\psfrag{a2}{$a=0.2$}

\begin{center}
\includegraphics[height=2.4526in, width=3.3797in ]{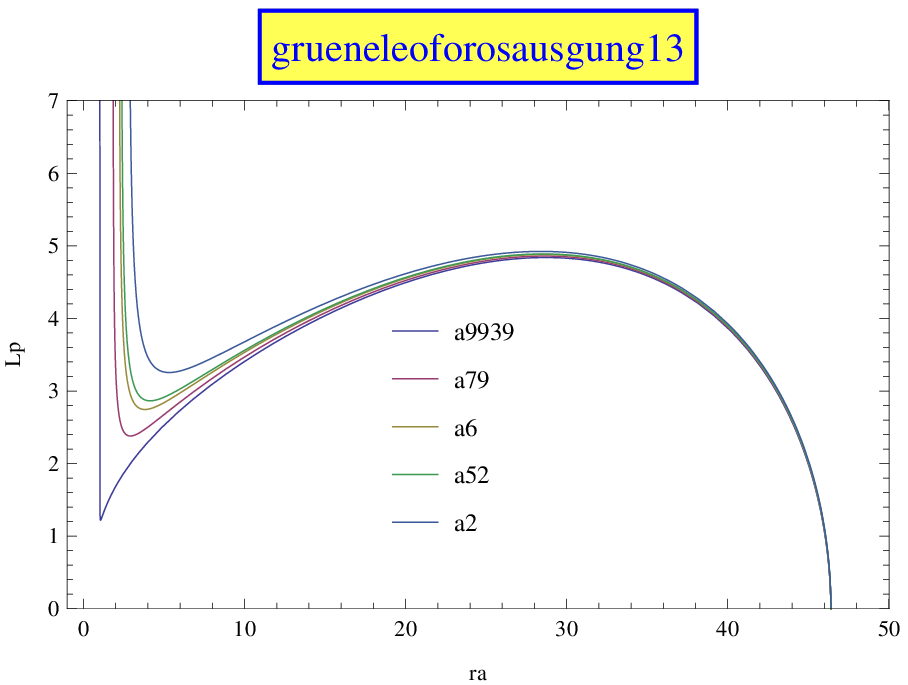}
 \caption{Specific angular momentum $L_{+}$ for $e=0.11, \Lambda^{\prime}=10^{-5}$ for different values for the Kerr parameter.}%
\label{PLPlusLambda}%
\end{center}
\end{figure}

\begin{figure}
[ptbh]

\psfrag{r}{$r$} \psfrag{Lplus}{$L_{+}$}
\psfrag{a52Lpm4e0}{$a=0.52,\Lambda^{\prime}=10^{-4},e=0$}
\psfrag{a52kerronly}{$a=0.52,\Lambda^{\prime}=0,e=0$}
\psfrag{a52L0e085}{$a=0.52,\Lambda^{\prime}=0,e=0.85$}
\psfrag{a52Lnm4e0}{$a=0.52,\Lambda^{\prime}=-10^{-4},e=0$}

\begin{center}
\includegraphics[height=2.4526in, width=3.3797in ]{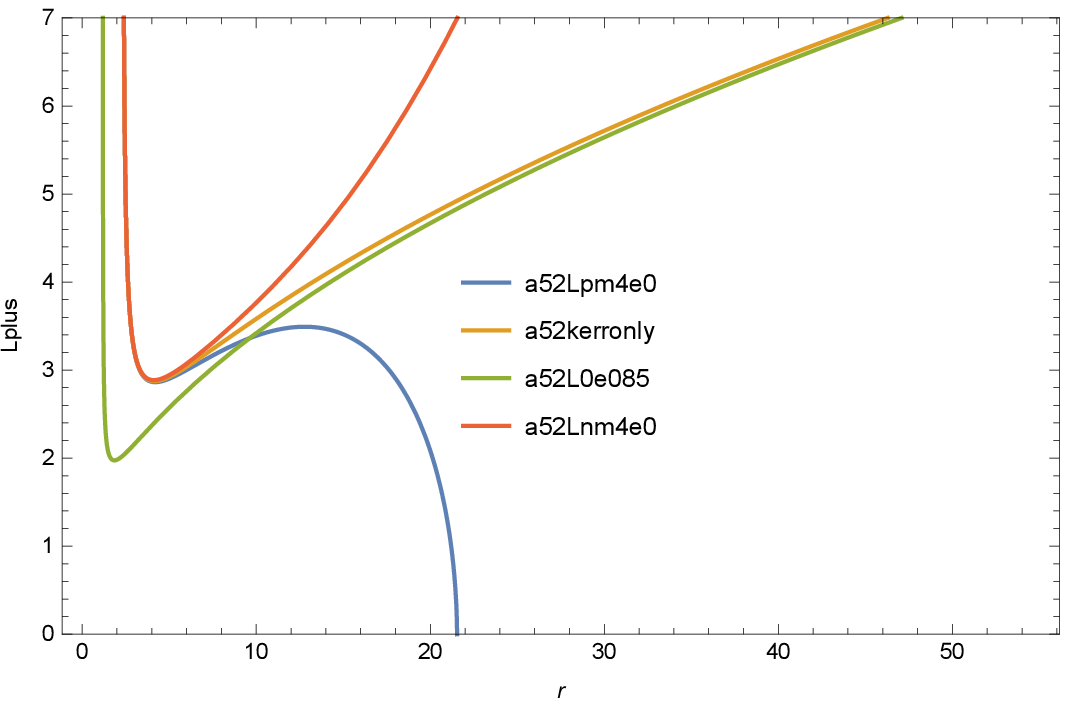}
 \caption{Radial profile for specific angular momentum $L_{+}$ of test particles moving on equatorial circular orbits in KN(a)dS black hole with Kerr parameter $a=0.52$, and various values for the cosmological constant and the black hole's electric charge. The case of the Kerr black hole $e=\Lambda^{\prime}=0$ is displayed.}%
\label{DLPlusKNadScompKERR}%
\end{center}
\end{figure}

\begin{figure}
[ptbh]

\psfrag{r}{$r$} \psfrag{Lplus}{$L_{+}$}
\psfrag{a9939Lpm4e0}{$a=0.9939,\Lambda^{\prime}=10^{-4},e=0$}
\psfrag{a9939kerronly}{$a=0.9939,\Lambda^{\prime}=0,e=0$}
\psfrag{a9939L0e085}{$a=0.9939,\Lambda^{\prime}=0,e=0.11$}
\psfrag{a9939Lnm4e0}{$a=0.9939,\Lambda^{\prime}=-10^{-4},e=0$}

\begin{center}
\includegraphics[height=2.4526in, width=3.3797in ]{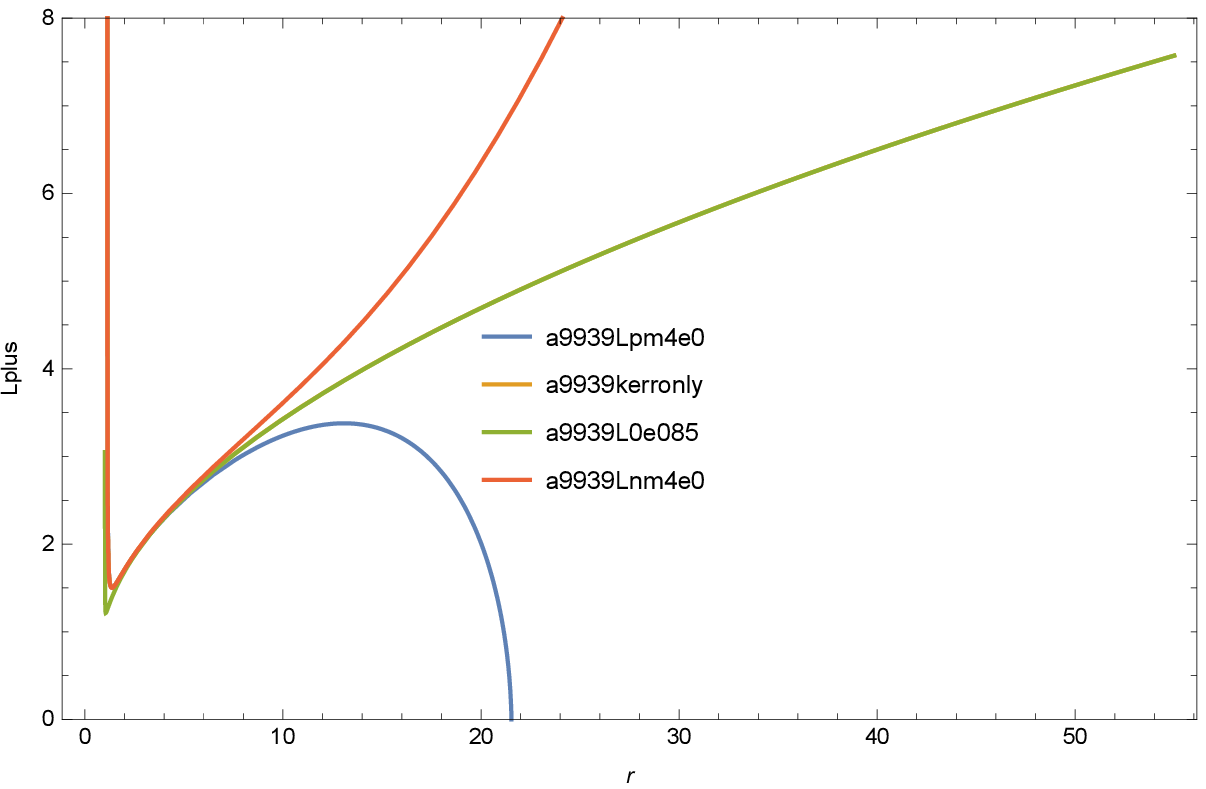}
 \caption{Radial profile for specific angular momentum $L_{+}$ of test particles moving on equatorial circular orbits in KN(a)dS black hole with Kerr parameter $a=0.9939$, and various values for the cosmological constant and the black hole's electric charge. The case of the Kerr black hole $e=\Lambda^{\prime}=0$ is displayed.}%
\label{ProgradeLKerr}%
\end{center}
\end{figure}

\begin{figure}
[ptbh]
\psfrag{Bulevardpaoausgung13a}{$e=0.11,\;\;\Lambda^{\prime}=0.001.$}
\psfrag{ra}{$r$} \psfrag{LnpC}{$L_{-}$}
\psfrag{a9939}{$a=0.9939$}
\psfrag{a79}{$a=0.7939$}
\psfrag{a6}{$a=0.6$}
\psfrag{a52}{$a=0.52$}
\psfrag{a2}{$a=0.2$}

\begin{center}
\includegraphics[height=2.4526in, width=3.3797in ]{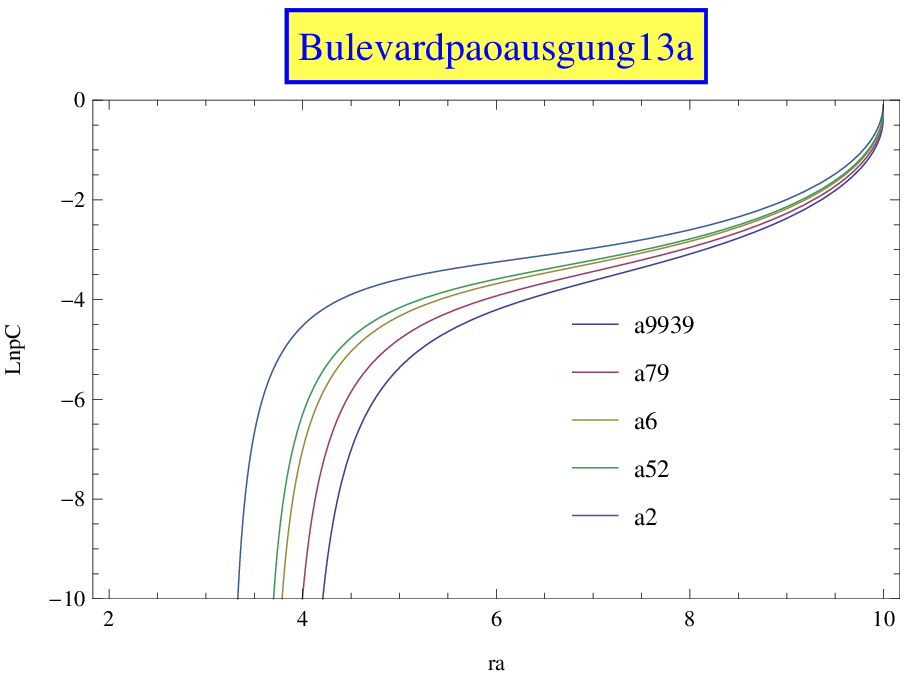}
 \caption{Specific angular momentum $L_{-}$ for $e=0.11, \Lambda^{\prime}=10^{-3}$ for different values for the Kerr parameter.}%
\label{negLposL1}%
\end{center}
\end{figure}

\begin{figure}
[ptbh]
\psfrag{Alleeleoforosausgung13}{$e=0.11,\;\;\Lambda^{\prime}=0.0001.$}
\psfrag{ra}{$r$} \psfrag{LnpC1}{$L_{-}$}
\psfrag{a9939}{$a=0.9939$}
\psfrag{a79}{$a=0.7939$}
\psfrag{a6}{$a=0.6$}
\psfrag{a52}{$a=0.52$}
\psfrag{a2}{$a=0.2$}

\begin{center}
\includegraphics[height=2.4526in, width=3.3797in ]{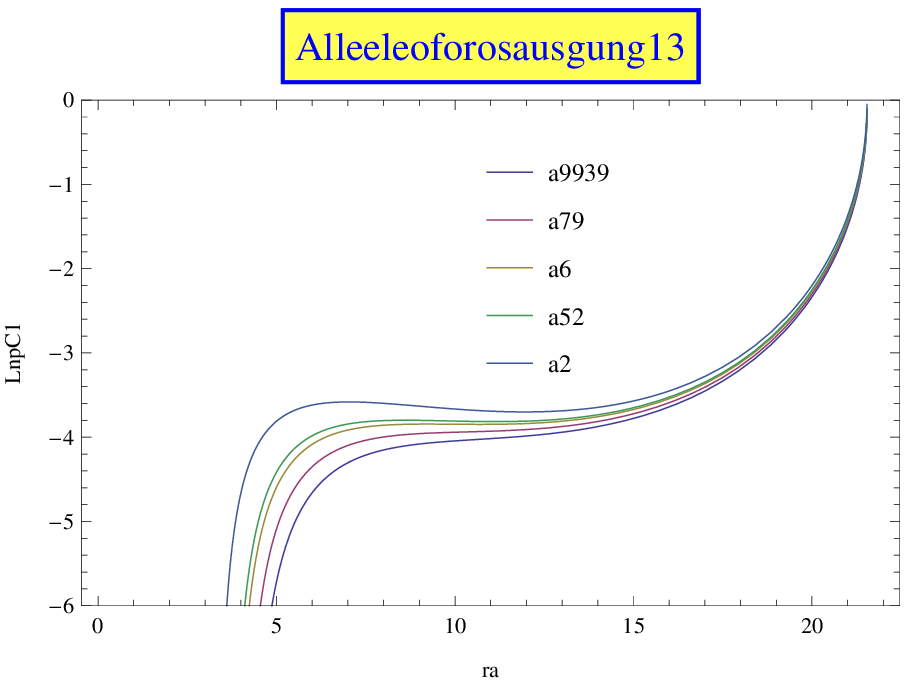}
 \caption{Specific angular momentum $L_{-}$ for $e=0.11, \Lambda^{\prime}=10^{-4}$ for different values for the Kerr parameter.}%
\label{negLposL2}%
\end{center}
\end{figure}

\begin{figure}
[ptbh]
\psfrag{Alleealexandrasausgung13}{$e=0.11,\;\;\Lambda^{\prime}=10^{-5}.$}
\psfrag{ra}{$r$} \psfrag{LnpC2}{$L_{-}$}
\psfrag{a9939}{$a=0.9939$}
\psfrag{a79}{$a=0.7939$}
\psfrag{a6}{$a=0.6$}
\psfrag{a52}{$a=0.52$}
\psfrag{a2}{$a=0.2$}

\begin{center}
\includegraphics[height=2.4526in, width=3.3797in ]{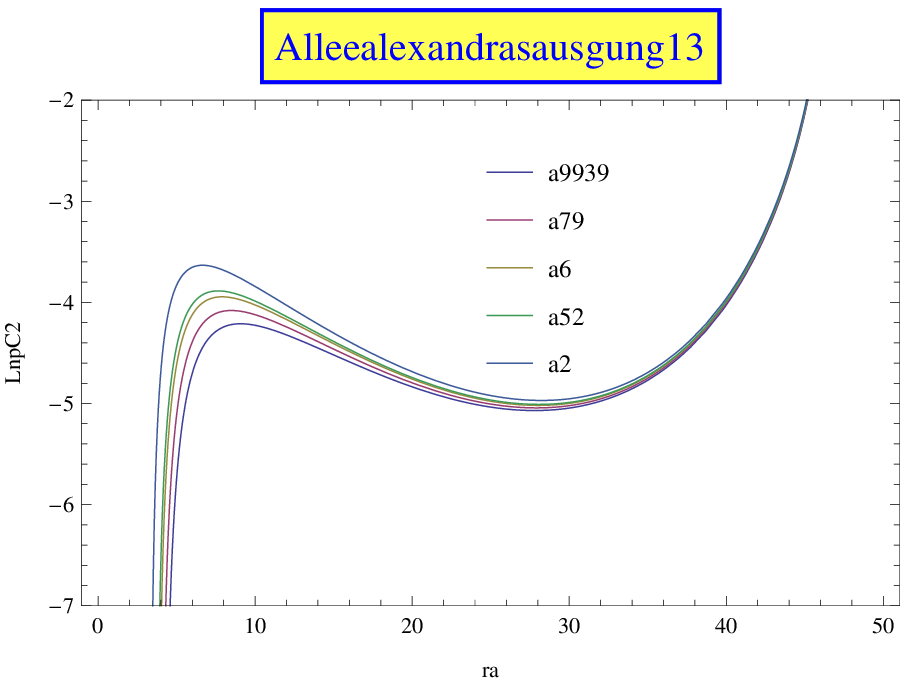}
 \caption{Specific angular momentum $L_{-}$ for $e=0.11, \Lambda^{\prime}=10^{-5}$ for different values for the Kerr parameter.}%
\label{negLposL3}%
\end{center}
\end{figure}

\begin{figure}
[ptbh]
\psfrag{locoparapanathinaikosg13}{$e=0.11,\Lambda^{\prime}=-0.001$}
\psfrag{ra}{$r$} \psfrag{Em}{$E_{+}$}
\psfrag{a79}{$a=0.7939$}
\psfrag{a6}{$a=0.6$}
\psfrag{a52}{$a=0.52$}
\psfrag{a26}{$a=0.26$}
\psfrag{a2}{$a=0.2$}
\begin{center}
\includegraphics[height=2.4526in, width=3.3797in ]{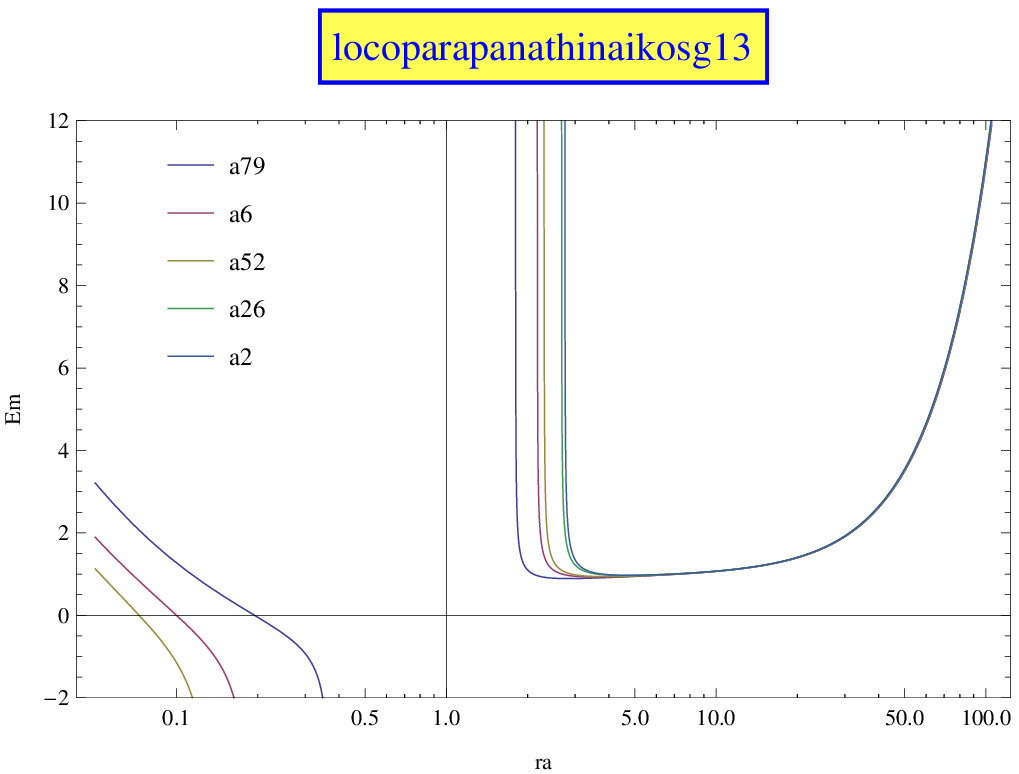}
 \caption{Specific energy $E_{+}$ for $e=0.11, \Lambda^{\prime}=-0.001$ for different values for the Kerr parameter. The radii of the event horizons for the specific spins of the black hole $(r_+,a)$ are:$(1.59128,0.7939),(1.78524,0.6),(1.83972,0.52),(1.9516,0.26),(1.96582,0.2)$. }%
\label{PlusEneg}%
\end{center}
\end{figure}

\begin{figure}
[ptbh]
\psfrag{greendvazelverdeausgung13}{$e=0.11,\Lambda^{\prime}=-0.001$}
\psfrag{ra}{$r$} \psfrag{Lp}{$L_{+}$}
\psfrag{a79}{$a=0.7939$}
\psfrag{a6}{$a=0.6$}
\psfrag{a52}{$a=0.52$}
\psfrag{a26}{$a=0.26$}
\psfrag{a2}{$a=0.2$}
\begin{center}
\includegraphics[height=2.4526in, width=3.3797in ]{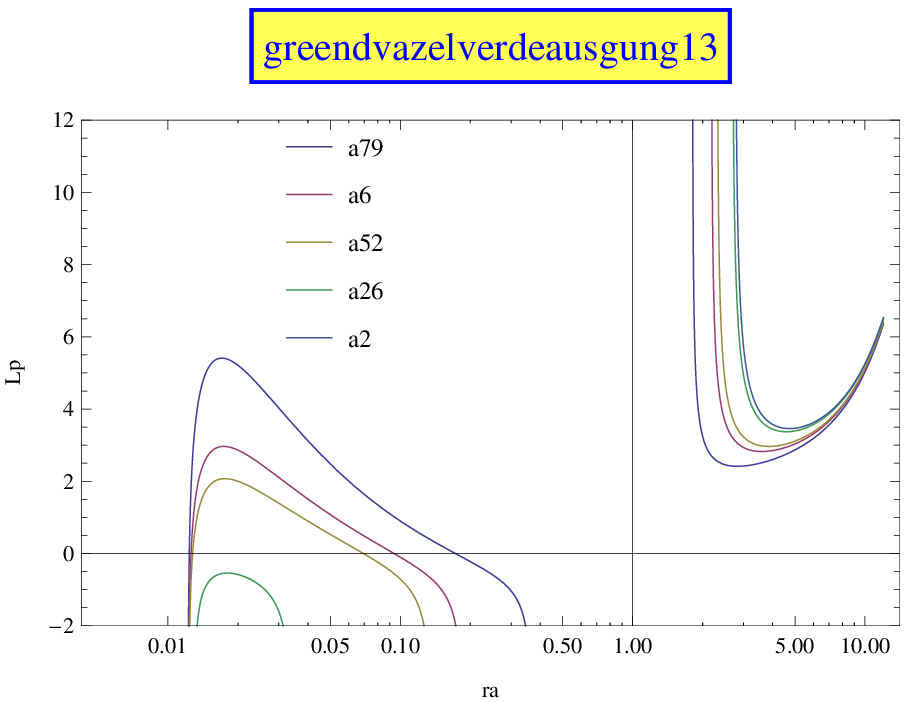}
 \caption{Specific momentum $L_{+}$ for $e=0.11, \Lambda^{\prime}=-0.001$ for different values for the Kerr parameter. The radii of the event horizons for the specific spins of the black hole $(r_+,a)$ are:$(1.59128,0.7939),(1.78524,0.6),(1.83972,0.52),(1.9516,0.26),(1.96582,0.2)$.}%
\label{PlusEnegAM}%
\end{center}
\end{figure}

\begin{figure}
[ptbh]
\psfrag{fitagate13paoreligion}{$e=0.6,\Lambda^{\prime}=-0.01$}
\psfrag{ra}{$r$} \psfrag{Lm}{$L_{-}$}
\psfrag{a52}{$a=0.52$}
\psfrag{a42}{$a=0.42$}
\psfrag{a026}{$a=0.26$}
\begin{center}
\includegraphics[height=2.4526in, width=3.3797in ]{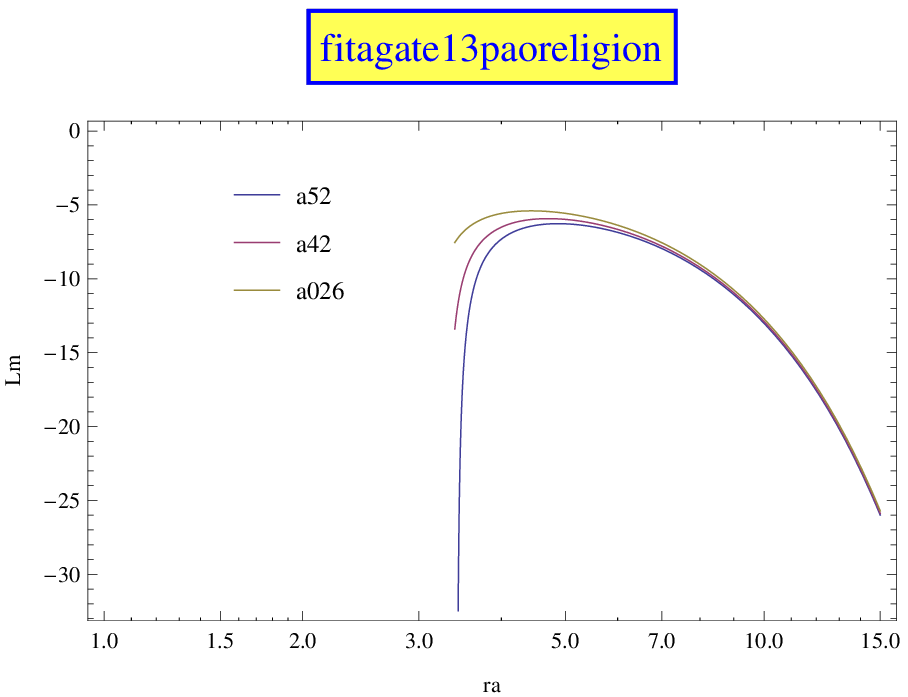}
 \caption{Specific angular momentum $L_{-}$ for $e=0.6, \Lambda^{\prime}=-0.01$ for different values for the Kerr parameter.}%
\label{Lneg}%
\end{center}
\end{figure}

\begin{figure}
[ptbh]
\psfrag{parasiempregatedreizehn}{$e=0.11,\Lambda^{\prime}=-0.001$}
\psfrag{ra}{$r$} \psfrag{Em}{$E_{-}$}
\psfrag{a6}{$a=0.6$}
\psfrag{a52}{$a=0.52$}
\psfrag{a26}{$a=0.26$}
\psfrag{a1}{$a=0.1$}
\begin{center}
\includegraphics[height=2.4526in, width=3.3797in ]{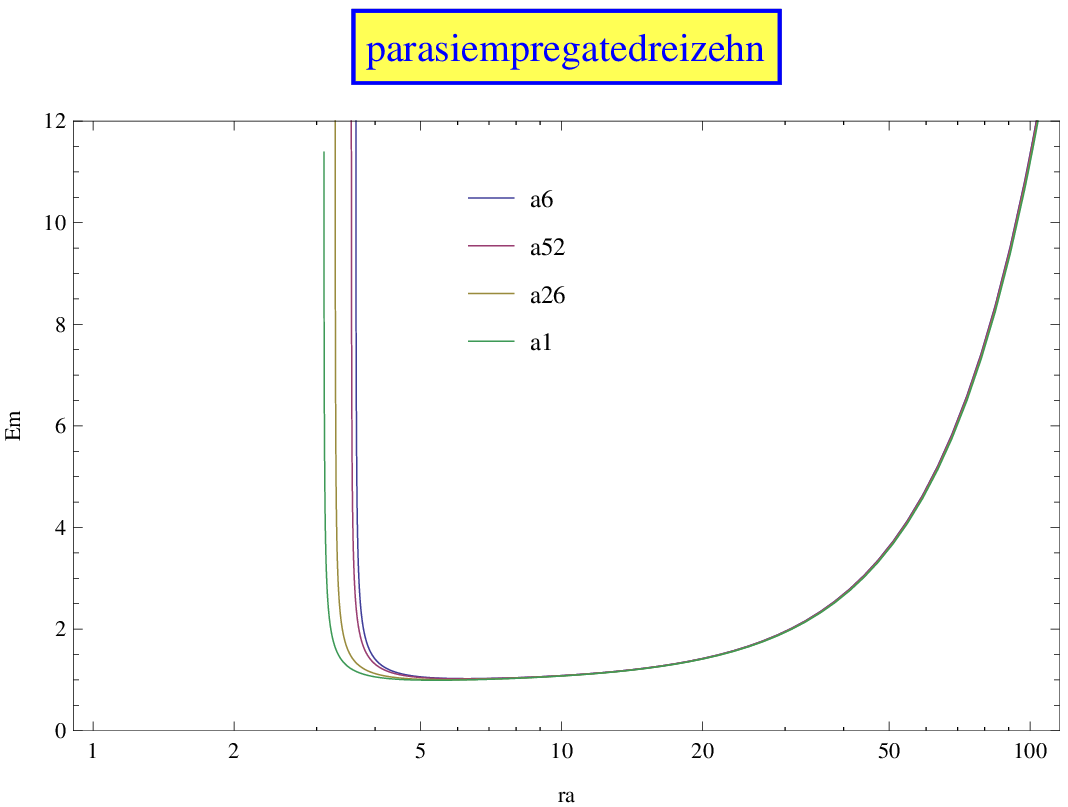}
 \caption{Specific energy $E_{-}$ for $e=0.11, \Lambda^{\prime}=-0.001$ for different values for the Kerr parameter.}%
\label{Eneg}%
\end{center}
\end{figure}

\begin{figure}
[ptbh]
\psfrag{fitathira13paoreligionps}{$e=0.11,\Lambda^{\prime}=-0.001$}
\psfrag{ra}{$r$} \psfrag{Lm}{$L_{-}$}
\psfrag{a9939}{$a=0.9939$}
\psfrag{a52}{$a=0.52$}
\psfrag{a26}{$a=0.26$}
\begin{center}
\includegraphics[height=2.4526in, width=3.3797in ]{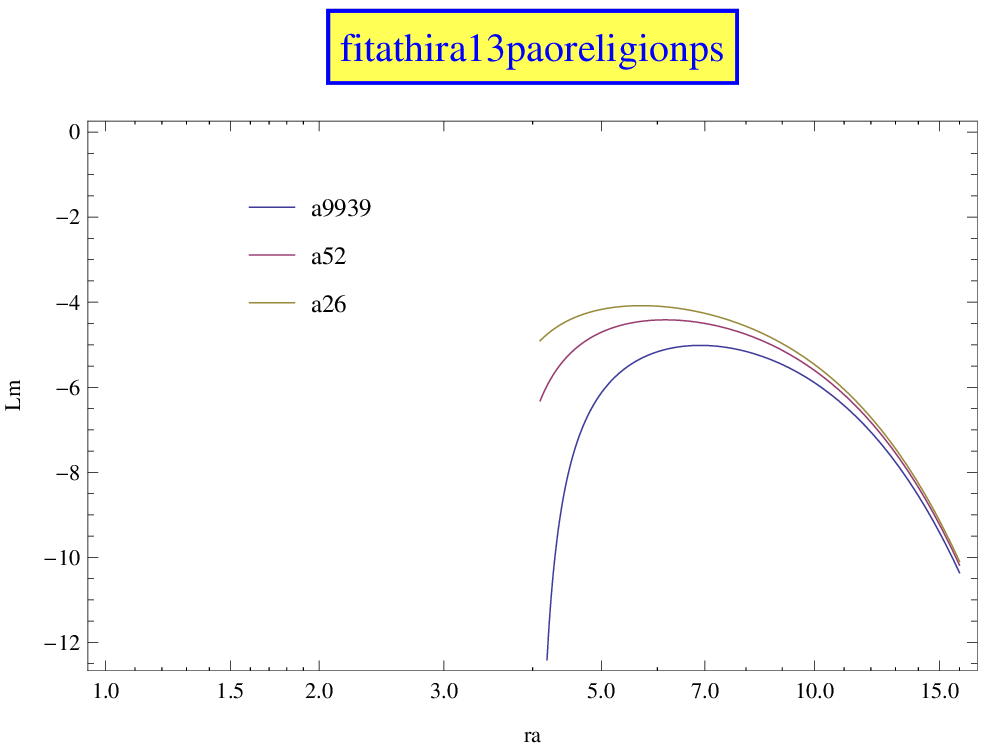}
 \caption{Specific angular momentum $L_{-}$ for $e=0.11, \Lambda^{\prime}=-0.001$ for different values for the Kerr parameter.}%
\label{Lneg3}%
\end{center}
\end{figure}

\begin{figure}
[ptbh]

\psfrag{r}{$r$} \psfrag{LHP}{$\Lambda^{\prime}_h$}
\psfrag{a0.3,e0.26}{$a=0.3,e=0.26$}
\psfrag{a0.7,e0.26}{$a=0.7,e=0.26$}
\psfrag{a0.9939,e0.05}{$a=0.9939,e=0.05$}

\begin{center}
\includegraphics[height=2.4526in, width=3.3797in ]{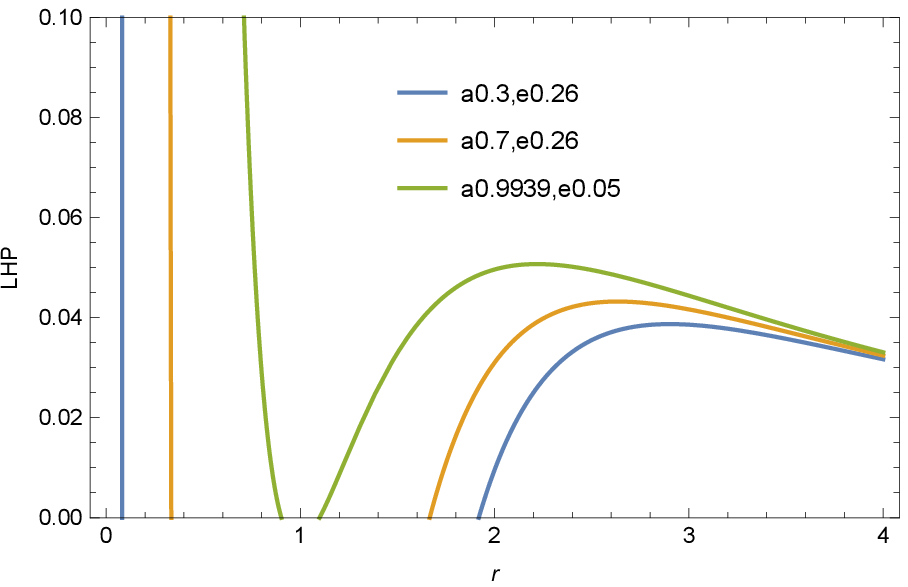}
 \caption{Horizons of the Kerr-Newman-de Sitter spacetimes defined by the function $\Lambda_h^{\prime}(r;a,e)$ illustrated for different values of the rotational (Kerr) parameter  $a$ and the electric charge of the black hole. }%
\label{HorizonsposLambda}%
\end{center}
\end{figure}

\begin{figure}
[ptbh]

\psfrag{r}{$r$} \psfrag{LHP}{$\Lambda^{\prime}_h$}
\psfrag{a0.3,e0.26}{$a=0.3,e=0.26$}
\psfrag{a0.7,e0.26}{$a=0.7,e=0.26$}
\psfrag{a0.9939,e0.05}{$a=0.9939,e=0.05$}
\psfrag{a052e06}{$a=0.52,e=0.6$}
\psfrag{a052e0}{$a=0.52,e=0$}

\begin{center}
\includegraphics[height=2.4526in, width=3.3797in ]{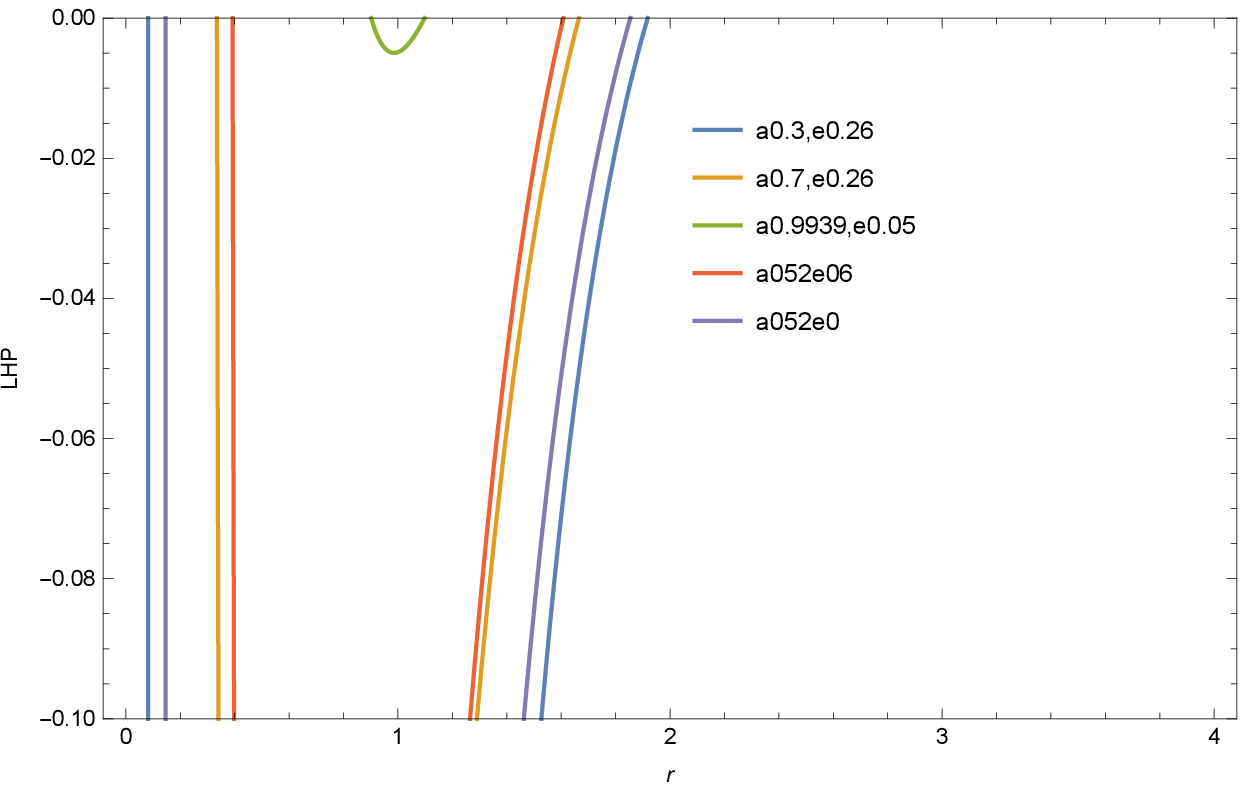}
 \caption{Horizons of the Kerr-Newman-anti de Sitter spacetimes defined by the function $\Lambda_h^{\prime}(r;a,e)$ illustrated for different values of the rotational (Kerr) parameter  $a$ and the electric charge of the black hole. }%
\label{HorizonsnegLambda}%
\end{center}
\end{figure}

\subsection{Stability of circular equatorial orbits in Kerr-Newman-de Sitter spacetime}
The loci of the stable equatorial circular orbits are determined by the inequality condition:
\begin{equation}
\frac{{\rm d}^2R^{\prime}}{{\rm d}r^2}\geq 0,
\end{equation}
which has to be satisfied simultaneously with the conditions $R^{\prime}(r)=0$
and ${\rm d}R^{\prime}/{\rm d}r=0$ determining the specific energy and specific angular momentum of the constant radius orbits. Using relations
(\ref{marvelOne})-(\ref{marvelTwo}) we find that:
\begin{align}
&\frac{{\rm d}^2R^{\prime}}{{\rm d}r^2}\geq 0\Leftrightarrow \nonumber\\
\Biggl\{\Biggl[&\mp 8ar^6(\frac{1}{r^3}-\Lambda^{\prime}-\frac{e^2}{r^4})^{3/2}+r^2(6-r+r^3(-15+4r)\Lambda^{\prime})+
4e^4+3e^2r(-3+4r^3\Lambda^{\prime})\nonumber\\
&+a^2(-4e^2+r[3+r^2\Lambda^{\prime}(1-4r^3\Lambda^{\prime})])\Biggr]2
(2e^2+r(-3+r\pm 2ar\sqrt{\frac{1}{r^3}-\Lambda^{\prime}-\frac{e^2}{r^4}}-
a^2r\Lambda^{\prime}))\Biggr\}/\nonumber\\
&[2e^2+r(-3+r\pm 2ar\sqrt{\frac{1}{r^3}-\Lambda^{\prime}-\frac{e^2}{r^4}}-
a^2r\Lambda^{\prime})]^2\geq 0.
\end{align}
Due to reality conditions (\ref{reality1}) we find that the radii of the stable orbits in Kerr-Newman-de Sitter spacetime are restricted by the inequality:
\begin{align}
&\mp 8ar^6(\frac{1}{r^3}-\Lambda^{\prime}-\frac{e^2}{r^4})^{3/2}+r^2(6-r+r^3(-15+4r)\Lambda^{\prime})+
4e^4+3e^2r(-3+4r^3\Lambda^{\prime})\nonumber\\
&+a^2(-4e^2+r[3+r^2\Lambda^{\prime}(1-4r^3\Lambda^{\prime})])\geq0.
\label{equilibriumKNdS}
\end{align}

The marginally stable orbits can be obtained by solving the quadratic equation (embedded in (\ref{equilibriumKNdS}))  for the rotational parameter a procedure that yields the relation:
\begin{align}
a^2&=a^2_{{\rm ms}(1,2)}\nonumber \\
&\equiv\Biggl\{ 128 r(1-r^3\Lambda^{\prime}-e^2/r)^3r^2\nonumber \\
&-4[-4e^2+r(3+r^2\Lambda^{\prime}(1-4r^3\Lambda^{\prime}))]
\{r^2(6-r+r^3(-15+4r)\Lambda^{\prime})+4e^4+3e^2r(-3+4r^3\Lambda^{\prime})\}\nonumber\\
&\pm 32[r(1-r^3\Lambda^{\prime}-e^2/r)^3]^{1/2}r\sqrt{r^3}\nonumber\\
&\times\sqrt{(1-4\Lambda^{\prime}r^3)}
\sqrt{-2+3r-r^2\Lambda^{\prime}(6+10r-15\Lambda^{\prime}r^3)-4e^4\Lambda^{\prime}
+\frac{3e^2}{r}(1+4r^3\Lambda^{\prime})-4e^2+9\Lambda^{\prime}e^2r-12r^4\Lambda^{\prime}e^2}
\Biggr\}
\nonumber\\
&4^{-1}[-4e^2+r(3+r^2\Lambda^{\prime}(1-4r^3\Lambda^{\prime}))]^{-2}.
\label{marginal}
\end{align}
From (\ref{marginal}) we deduce the following  reality conditions :
\begin{equation}
\Lambda^{\prime}\leq \Lambda^{\prime}_{{\rm ms}}(r)\equiv \frac{1}{4r^3},
\end{equation}
and
\begin{align}
&-2+3r-r^2\Lambda^{\prime}(6+10r-15\Lambda^{\prime}r^3)-4e^4\Lambda^{\prime}
+\frac{3e^2}{r}(1+4r^3\Lambda^{\prime})-4e^2+9\Lambda^{\prime}e^2r-12r^4\Lambda^{\prime}e^2\geq0
\\&\Leftrightarrow
\Lambda^{\prime 2}15r^5+\Lambda^{\prime}[-r^2(6+10r)-4e^4+12e^2r^2+9e^2r-12r^4e^2]-2+3r+\frac{3e^2}{r}-4e^2\geq0.
\label{trionimo}
\end{align}
From (\ref{trionimo}) and the fact that $15r^5>0$ we derive two more conditions:
\begin{equation}
\Lambda^{\prime}\leq \Lambda_{{\rm ms}-}^{\prime}\;\; \text{or}\;\; \Lambda^{\prime}\geq \Lambda_{{\rm ms}+}^{\prime},
\end{equation}
where $\Lambda_{{\rm ms}\pm}^{\prime}$ are the two roots of the quadratic equation (\ref{trionimo}).
A detailed discussion of stability of the geodesic circular equatorial motion in the KNdS spacetimes is presented in \cite{plaglogo}.

\subsubsection{Stability of equatorial circular geodesics in Kerr-de Sitter spacetime}

Our results in inequality (\ref{equilibriumKNdS}), for zero electric charge, $e=0$, give a condition \footnote{We observe that for $e=\Lambda=0$, (\ref{equilibriumKNdS}) reduces correctly to the equation that determines the radii of marginally stable orbits in Kerr spacetime: $r^2-6Mr-3a^2\mp8a\sqrt{Mr}=0$ for $M=1$, \cite{Bardeen}.} that agrees with the results in \cite{opav2} and restricts
 the radii of the stable orbits in   Kerr-de Sitter spacetime:

\begin{equation}
\pm 8ar^2\sqrt{\frac{1}{r^3}-\Lambda^{\prime}}\;(-1+r^3\Lambda^{\prime})
+r(6-r+r^3(-15+4r)\Lambda^{\prime})+a^2(3+r^2\Lambda^{\prime}(1-4r^3\Lambda^{\prime}))\geq0.
\end{equation}
\section{Gravitational redshift-blueshift of emitted photons}\label{grredblue}

In this section we will provide general expressions for the redshift/blueshift that emitted photons by massive particles experience while travelling along null geodesics towards an observed located far away from their source.

In general, the frequency of a photon measured by an observer with proper velocity $U^{\mu}_A$ at the spacetime point $P_A$ reads\cite{HansOhanian},\cite{HERRERA}:
\begin{equation}
\omega_A=k_{\mu}U^{\mu}_A|_{P_A},
\end{equation}
where the index $A$ refers to the emission $(e)$ and/or detection $(d)$ at the corresponding point $P_A$.

The emission frequency is defined as follows:
\begin{align}
\omega_e&=k_{\mu}U^{\mu}\nonumber \\
&=k_tU^t+k_rU^r+k_{\theta}U^{\theta}+k_{\phi}U^{\phi}\nonumber \\
&=(k^tE-k^{\phi}L+g_{rr}k^rU^r+g_{\theta\theta}k^{\theta}U^{\theta})|_e.
\label{emitter}
\end{align}
Likewise the detected frequency is given by the expression:
\begin{align}
\omega_d&=+k_{\mu}U^{\mu}\nonumber \\
&=(Ek^t-Lk^{\phi}+g_{rr}k^rU^r+g_{\theta\theta}k^{\theta}U^{\theta})|_d.
\label{detection}
\end{align}
In producing (\ref{emitter}),(\ref{detection}) we used the expressions for $U^t$ and $U^{\phi}$ in terms of the metric components and the conserved quantities $E,L$:
\begin{align}
U^t&=\frac{-E g_{\phi\phi}-Lg_{t\phi}}{g^2_{t \phi}-g_{tt}g_{\phi\phi}},\\
U^{\phi}&=\frac{g_{tt}L+g_{\phi t}E}{g^2_{t \phi}-g_{tt}g_{\phi\phi}}.
\end{align}

Thus, the frequency shift associated to the emission and detection of photons is given by either of the following relations:
\begin{align}
1+z&=\frac{\omega_e}{\omega_d}\nonumber\\
&=\frac{(k^tE-k^{\phi}L+g_{rr}k^rU^r+g_{\theta\theta}k^{\theta}U^{\theta})|_e}{(Ek^t-Lk^{\phi}+g_{rr}k^rU^r+g_{\theta\theta}k^{\theta}U^{\theta})|_d}\nonumber \\
&=\frac{(E_{\gamma}U^t-L_{\gamma}U^{\phi}+g_{rr}K^rU^r+g_{\theta\theta}K^{\theta}U^{\theta})|_e}{(
E_{\gamma}U^t-L_{\gamma}U^{\phi}+g_{rr}K^rU^r+g_{\theta\theta}K^{\theta}U^{\theta})|_d}
\label{gravrbshiftdoppler}
\end{align}

This is the most general expression for the redshift/blueshift that light signals emitted by massive test particles experience in their path along null geodesics towards a distant observer (ideally located near the cosmological horizon in particular or at spatial infinity assuming a  zero cosmological constant).

\subsection{The redshift/blueshift of photons for circular and equatorial emitter/detector orbits around the Kerr-Newman-(anti) de Sitter black hole}\label{erithrimplemetatopisi}
For equatorial circular orbits $U^r=U^{\theta}=0$ thus
\begin{equation}
1+z=\frac{(E_{\gamma}U^t-L_{\gamma}U^{\phi})|_e}{(E_{\gamma}U^t-L_{\gamma}U^{\phi})|_d}
=\frac{U^t-\Phi U^{\phi}|_e}{U^t-\Phi U^{\phi}|_d}=\frac{U^t_e-\Phi_e U^{\phi}_e}{U^t_d-\Phi_d U^{\phi}_d}=\frac{U^t_e-\Phi U^{\phi}_e}{U^t_d-\Phi U^{\phi}_d},
\label{initialredshift}
\end{equation}
where $\Phi=L_{\gamma}/E_{\gamma}$ \footnote{Since the constants of motion $E_{\gamma}$ and $L_{\gamma}$ are preserved along the null geodesics followed by the photons from emission till detection, we have $\Phi_e=\Phi_d=\Phi$, i.e. this quantity is also constant along the whole photons trajectory.}. For $\Phi=0,1+z_c=\frac{U_e^t}{U_d^t}$. Following the procedure for the Kerr black hole in \cite{HERRERA}, we consider the kinematic redshift of photons either side of the line of sight that links the Kerr-Newman-de Sitter black hole and the observer, and subtract from Eq.(\ref{initialredshift}) the central value $z_c$.  We note that $z_c$ corresponds to a gravitational frequency shift of a photon emitted by a static particle located in a radius equal to the circular orbit radius and on the signal line going from the centre of coordinates to the far detector \cite{HERRERA},\cite{becerril}. Then we obtain:
\begin{align}
z_{{\rm kin}}&\equiv z-z_c=\frac{U^t_e-\Phi_e U^{\phi}_e}{U^t_d-\Phi_d U^{\phi}_d}-\frac{U_e^t}{U_d^t}\nonumber\\
&=\frac{\Phi U_d^{\phi}U_e^t-\Phi U_e^{\phi}U_d^t}{U_d^tU_d^t-\Phi U_d^{\phi}U_d^t}.
\label{kineticredshift}
\end{align}
We further comment on the different gravitational frequency shifts of photons included in (\ref{initialredshift}) and (\ref{kineticredshift}). In Eq.(\ref{initialredshift}), the redshift/blueshift is indeed gravitational, but it includes an equivalent Doppler effect (redshift/blueshift) as the emitter moves towards/away from the observer along the timelike circular orbit and, additionally, two gravitational effects: a gravitational redshift for the photon emitted by a static particle and a redshift/blueshift due to the rotation of the space time (as is the case for the Kerr-Newman-de Sitter black hole spacetime).

Let us now consider photons with 4-momentum vector $k^{\mu}=(k^t,k^r,k^{\theta},k^{\phi})$ which move along null geodesics $k^{\mu}k_{\mu}=0$ outside the event horizon of the Kerr-Newman-de Sitter black hole, which explicitly can be expressed as
\begin{align}
0=&g_{tt}(k^t)^2+2g_{t\phi}(k^{t}k^{\phi})+g_{\phi\phi}(k^{\phi})^2+g_{rr}(k^r)^2\nonumber \\
&+g_{\theta\theta}(k^{\theta})^2.
\end{align}

\begin{align}
k^t&=\frac{\Xi^2\Delta_{\theta}(r^2+a^2)[(r^2+a^2)E_{\gamma}-aL_{\gamma}]-a\Xi^2\Delta_r^{KN}(aE_{\gamma}\sin^2\theta-L_{\gamma})}{\Delta_r^{KN}\Delta_{\theta}\rho^2}\nonumber \\
&=\frac{E_{\gamma}[\Xi^2\Delta_{\theta}(r^2+a^2)^2-a^2\sin^2\theta\Xi^2\Delta_r^{KN}]+L_{\gamma}[-a\Xi^2\Delta_{\theta}(r^2+a^2)+a\Xi^2\Delta_r^{KN}]}{\Delta_r^{KN}\Delta_{\theta}\rho^2},
\\
k^{\phi}&=\frac{-\Xi^2\Delta_r^{KN}(aE_{\gamma}\sin^2\theta-L_{\gamma})+
a\Xi^2\Delta_{\theta}\sin^2\theta[(r^2+a^2)E_{\gamma}-a L_{\gamma}
]}{\Delta_r^{KN}\Delta_{\theta}\rho^2\sin^2\theta}\nonumber \\
&=\frac{[-\Xi^2\Delta_r^{KN}a\sin^2\theta+a\Xi^2\Delta_{\theta}\sin^2\theta(r^2+a^2)]E_{\gamma}
+L_{\gamma}[\Xi^2\Delta_r^{KN}-a^2\Xi^2\Delta_{\theta}\sin^2\theta]}{\Delta_r^{KN}\Delta_{\theta}\rho^2\sin^2\theta},\\
(k^{\theta})^2&=\frac{\mathcal{Q}_{\gamma}\Delta_{\theta}+(L_{\gamma}-aE_{\gamma})\Xi^2\Delta_{\theta}-\frac{\Xi^2(aE_{\gamma}\sin^2\theta-L_{\gamma})^2}{\sin^2\theta}}{\rho^4},\\
(k^{r})^2&=\frac{\Xi^2[(r^2+a^2)E_{\gamma}-aL_{\gamma}]^2-\Delta_r^{KN}(\mathcal{Q}_{\gamma}+\Xi^2(L_{\gamma}-aE_{\gamma})^2)}{\rho^4}.
\end{align}

We must take into account the  bending of light from the rotating and charged  Kerr-Newman-(anti) de Sitter black hole. In other words we need to find the apparent impact parameter $\Phi$ for every orbit, i.e. to find the map, $\Phi=\Phi(r)$ as a function of the radius $r$ of the circular orbit associated with the emitter.  The criteria employed in \cite{HERRERA} to construct this mapping is to choose the maximum value of $z$ at a fixed distance from the observed centre of the source at a fixed $\Phi$ and that the apparent impact parameter must also be maximised.
The apparent impact factor $\Phi_{\gamma}\equiv L_{\gamma}/E_{\gamma}$ can be obtained from the expression $k^{\mu}k_{\mu}=0$ \footnote{Taking into account that $k^r=0$ and $k^{\theta}=0$.} as follows:
\begin{align}
&k^{\mu}k_{\mu}=0\Leftrightarrow k^tk_t+k^{\phi}k_{\phi}=0\nonumber \\
&\Leftrightarrow \left[\frac{E_{\gamma}g_{\phi\phi}+
g_{t\phi}L_{\gamma}}{g^2_{t\phi}-g_{\phi\phi}g_{tt}}\right](-E_{\gamma})+
\left[\frac{-L_{\gamma}g_{tt}-E_{\gamma}g_{\phi t}}{g^2_{t\phi}-g_{\phi\phi}g_{tt}}\right]L_{\gamma}=0\nonumber\\
&\Leftrightarrow g_{\phi\phi}+2g_{t\phi}\Phi_{\gamma}+\Phi^2_{\gamma}g_{tt}=0
\end{align}
Solving the quadratic equation we obtain:
\begin{equation}
\Phi_{\gamma}^{\pm}=\frac{-g_{\phi t}\pm \sqrt{g^2_{t\phi}-g_{\phi\phi}g_{tt}}}{g_{tt}}=\frac{a(\Delta_r^{KN}-(r^2+a^2))\pm r^2\sqrt{\Delta_r^{KN}}}{\Delta_r^{KN}-a^2},
\label{phiparagontas}
\end{equation}
where we got two values, $\Phi_{\gamma}^{+}$ and $\Phi_{\gamma}^{-}$ (either evaluated at the emitter or detector position, since this quantity is preserved along the null geodesic photon orbits, i.e., $\Phi_e=\Phi_d$) that give rise to two different shifts respectively, $z_1$ and $z_2$ of the emitted photons corresponding to a receding and to an approaching object with respect to a far away positioned observer:
\begin{align}
z_1&=\frac{\Phi_d^-U_d^{\phi}U_e^t-\Phi_e^-U_e^{\phi}U_d^t}{U_d^t(U_d^t-
\Phi_d^-U_d^{\phi})}=\frac{\Phi_e^-U_d^{\phi}U_e^t-\Phi_e^-U_e^{\phi}U_d^t}{U_d^t(U_d^t-
\Phi_e^-U_d^{\phi})}\label{red1ein},\\
z_2&=\frac{\Phi_d^+U_d^{\phi}U_e^t-\Phi_e^+U_e^{\phi}U_d^t}{U_d^t(U_d^t-
\Phi_d^+U_d^{\phi})}=\frac{\Phi_e^+U_d^{\phi}U_e^t-\Phi_e^+U_e^{\phi}U_d^t}{U_d^t(U_d^t-
\Phi_e^+U_d^{\phi})}\label{red2zwei}.
\end{align}
In general the two values $z_1$ and $z_2$ differ from each other due to light bending experienced by the emitted photons and the differential rotation experienced by the detector as encoded in $U_d^{\phi}$ and $U_d^t$ components of the four-velocity \footnote{The second term in the denominator in equations (\ref{red1ein})-(\ref{red2zwei}) encodes the contribution of the movement of the detector's frame \cite{HERRERA}. If this quantity is negligible in comparison to the contribution steming from the $U^t_d$ component
($U_d^{\phi}\ll U_d^t$) then the detector can be considered static at spatial infinity for the case of Kerr black hole. There are no static observers in infinity of the Kerr-de Sitter (KdS) spacetimes \cite{LightescapeConeskds}. Indeed, in the KdS black hole spacetimes the free static observers are expected near the static radius representing the outermost region of gravitationally bound systems in Universe with accelerated cosmic expansion-see discussion in \cite{LightescapeConeskds}.}.

In order to get a closed analytic expression for the gravitational redshift/blueshift experienced by the emitted photons we shall express the required quantities in terms of the Kerr-Newman-(anti) de Sitter metric. Thus, the $U^{\phi}$ and $U^t$ components of the four-velocity for circular equatorial orbits read:
\begin{align}
U^t(r,\pi/2)&=\frac{-(\Delta_r^{\rm KN}a^2-(r^2+a^2)^2)E-L(-a (\Delta_r^{\rm KN}-(r^2+a^2)))}{\frac{\Xi^2 r^2 \Delta_r^{\rm KN}}{\Xi^4}},\\
U^{\phi}(r,\theta=\pi/2)&=\frac{\Xi^2(\Delta_r^{KN}-a^2)L+E\Xi^2(-a(\Delta_r^{KN}-(r^2+a^2)))}{r^2\Delta_r^{KN}}.
\end{align}

Substituting the expressions (\ref{marvelOne})-(\ref{marvelTwo}) for $E_{\pm}$ and $L_{\pm}$ into $U^t(r,\pi/2),U^{\phi}(r,\pi/2)$
we finally obtain remarkable novel expressions for these four-velocity components in Kerr-Newman-(anti) de Sitter spacetime:
\begin{align}
U^{t}(r,\pi/2)&=\frac{(r^2\pm a\sqrt{-e^2+r^4\left(\frac{1}{r^3}-\Lambda^{\prime}\right)})\;\Xi^2}{r
\sqrt{2e^2+r(r-3)-r^2a^2\Lambda^{\prime}\pm 2a\sqrt{-e^2+r^4\left(\frac{1}{r^3}-\Lambda^{\prime}\right)}}},\\
U^{\phi}(r,\pi/2)&=\frac{\pm \sqrt{-e^2+r^{4}\left(\frac{1}{r^3}-\Lambda^{\prime}\right)}\;\Xi^2}{
r\sqrt{2e^2+r(r-3)-r^2a^2\Lambda^{\prime}\pm 2 a\sqrt{-e^2+r^4
\left(\frac{1}{r^3}-\Lambda^{\prime}\right)}}}.
\end{align}

We now compute the angular velocity $\Omega$:
\begin{align}
\Omega\equiv\frac{{\rm d}\phi}{{\rm d}t}=\frac{1}{a\pm\frac{r^{3/2}}{\sqrt{1-\Lambda^{\prime}r^3-\frac{e^2}{r}}}}.
\end{align}

In terms of the angular velocities the quantities $z_1,z_2$ read as follows:
\begin{align}
z_1&=\frac{\Phi_d^- \Omega_d U_e^t-\Phi_e^-U_e^{\phi}}{U_d^t-\Phi_d^-U_d^{\phi}}=\frac{U_e^t[\Phi_d^-\Omega_d-\Phi_e^-\Omega_e]}{U_d^t(1-\Phi_d^-\Omega_d)}\nonumber\\
&=\frac{U_e^t\Phi_e^-[\Omega_d-\Omega_e]}{U_d^t(1-\Phi_e^-\Omega_d)},\\
z_2&=\frac{\Phi_d^+\Omega_dU_e^t-\Phi_e^+U_e^{\phi}}{U_d^t-\Phi_d^+U_d^{\phi}}=\frac{\Phi_d^+\Omega_dU_e^t-\Phi_e^+U_e^{\phi}}{U_d^t(1-\Phi_d^+\Omega_d)}\nonumber\\
&=\frac{U_e^t[\Phi_d^+\Omega_d-\Phi_e^+\Omega_e]}{U_d^t(1-\Phi_d^+\Omega_d)}\nonumber \\
&=\frac{U_e^t\Phi_e^+[\Omega_d-\Omega_e]}{U_d^t(1-\Phi_e^+\Omega_d)}.
\end{align}

Thus for the Kerr-Newman-(anti) de Sitter black hole we can write for the redshift and blueshift, respectively:

\begin{align}
z_{\rm red}&=\frac{\Omega_d\Phi_d^{-}-\Phi_e^-\Omega_e}{1-\Phi_d^- \Omega_d}\frac{[r_e^{3/2}\pm a \sqrt{\mathfrak{G}_e}]}{
r_e^{3/4}\sqrt{\frac{2e^2}{r_e^{1/2}}+r_e^{3/2}-3\sqrt{r_e}-r_e^{3/2}a^2\Lambda^{\prime}\pm
 2 a\sqrt{\mathfrak{G}_e}}}\nonumber \\
 &\times \frac{r_d^{3/4}\sqrt{\frac{2e^2}{r_d^{1/2}}+
 r_d^{3/2}-3\sqrt{r_d}-r_d^{3/2}a^2\Lambda^{\prime}\pm 2 a
 \sqrt{\mathfrak{G}_d}}}{
 r_d^{3/2}\pm a \sqrt{\mathfrak{G}_d}},
\end{align}

\begin{align}
z_{\rm blue}&=\frac{\Omega_d\Phi_d^{+}-\Phi_e^+\Omega_e}{1-\Phi_d^+ \Omega_d}\frac{[r_e^{3/2}\pm a \sqrt{\mathfrak{G}_e}]}{
r_e^{3/4}\sqrt{\frac{2e^2}{r_e^{1/2}}+r_e^{3/2}-3\sqrt{r_e}-r_e^{3/2}a^2\Lambda^{\prime}\pm
 2 a\sqrt{\mathfrak{G}_e}}}\nonumber \\
 &\times \frac{r_d^{3/4}\sqrt{\frac{2e^2}{r_d^{1/2}}+
 r_d^{3/2}-3\sqrt{r_d}-r_d^{3/2}a^2\Lambda^{\prime}\pm 2 a
 \sqrt{\mathfrak{G}_d}}}{
 r_d^{3/2}\pm a \sqrt{\mathfrak{G}_d}},
\end{align}
where now $r_e$ and $r_d$ stand for the radius of the emitter's and detector's orbits, respectively.  We also define
\begin{align}
\mathfrak{G}_d &\equiv -e^2/r_d+r_d^3\left(\frac{1}{r_d^3}-\Lambda^{\prime}\right),\\
\mathfrak{G}_e &\equiv  -e^2/r_e+r_e^3\left(\frac{1}{r_e^3}-\Lambda^{\prime}\right)
\end{align}

These elegant and novel expressions can be written in terms of the physical parameters of the Kerr-Newman-(anti) de Sitter black hole and the detector radius, $r_d$, as follows:

\begin{align}
&z_{\rm red}=\frac{r_d^{3/4}\sqrt{\frac{2e^2}{r_d^{1/2}}+
r_d^{3/2}-3\sqrt{r_d}-r_d^{3/2}a^2\Lambda^{\prime}\pm
2a\sqrt{\mathfrak{G}_d}}}{
r_e^{3/4}\sqrt{\frac{2e^2}{r_e^{1/2}}+r_e^{3/2}-3\sqrt{r_e}-
r_e^{3/2}a^2\Lambda^{\prime}\pm 2 a\sqrt{\mathfrak{G}_e}}}\nonumber
\\
&\times\frac{\left(a(-\Lambda^{\prime}r_e^2(r_e^2+a^2)-2r_e+e^2)-r_e^2
\sqrt{\Delta_r^{KN}(r_e)}\right)(\pm [r_e^{3/2}\sqrt{\mathfrak{G}_d}-r_d^{3/2}\sqrt{\mathfrak{G}_e}])}{
(r_d^{3/2}\pm a\sqrt{\mathfrak{G}_d})
[(\Delta_r^{KN}(r_e)-a^2)r_d^{3/2}+(ar_e^2+r_e^2\sqrt{\Delta_r^{KN}(r_e)})(\pm \sqrt{\mathfrak{G}_d})]},
\label{shiftred}
\end{align}

\begin{align}
&z_{\rm blue}=\frac{r_d^{3/4}\sqrt{\frac{2e^2}{r_d^{1/2}}+
r_d^{3/2}-3\sqrt{r_d}-r_d^{3/2}a^2\Lambda^{\prime}\pm
2a\sqrt{\mathfrak{G}_d}}}{
r_e^{3/4}\sqrt{\frac{2e^2}{r_e^{1/2}}+r_e^{3/2}-3\sqrt{r_e}-
r_e^{3/2}a^2\Lambda^{\prime}\pm 2 a\sqrt{\mathfrak{G}_e}}}\nonumber
\\
&\times\frac{\left(a(-\Lambda^{\prime}r_e^2(r_e^2+a^2)-2r_e+e^2)+r_e^2
\sqrt{\Delta_r^{KN}(r_e)}\right)(\pm [r_e^{3/2}\sqrt{\mathfrak{G}_d}-r_d^{3/2}\sqrt{\mathfrak{G}_e}])}{
(r_d^{3/2}\pm a\sqrt{\mathfrak{G}_d})
[(\Delta_r^{KN}(r_e)-a^2)r_d^{3/2}+(ar_e^2-r_e^2\sqrt{\Delta_r^{KN}(r_e)})(\pm \sqrt{\mathfrak{G}_d})]},
\label{shiftblue}
\end{align}
where we define:
\begin{equation}
\Delta_r^{KN}(r_e):=(1-\Lambda^{\prime}r_e^2)(r_e^2+a^2)-2r_e+e^2
\end{equation}
and we have made use of the relation $\Phi_e=\Phi_d$.
The remarkable closed form analytic expressions for the frequency shifts we obtained in eqns.(\ref{shiftred})-(\ref{shiftblue}), constitute a new result in the theory of General Relativity, in which all the physical parameters of the exact theory enter on an equal footing.

In the particular case when the detector is located far away from the source and the detector radius $r_d$ is much larger than the black hole parameters ( its mass,spin and electric charge), the frequency shifts (\ref{shiftred})-(\ref{shiftblue}) become:

\begin{align}
z_{\rm red}=\frac{\sqrt{1-a^2\Lambda^{\prime}}[a(\Lambda^{\prime}r_e^2(r_e^2+a^2)+2r_e-e^2)+r_e^2\sqrt{\Delta_r(r_e)}](\pm\sqrt{1-\Lambda^{\prime}r_e^3-\frac{e^2}{r_e}})}
{r_e^{3/4}\sqrt{\frac{2e^2}{r_e^{1/2}}+r_e^{3/2}-3\sqrt{r_e}-r_e^{3/2}a^2\Lambda^{\prime}\pm 2a\sqrt{-\frac{e^2}{r_e}+r_e^3(\frac{1}{r_e^3}-\Lambda^{\prime})}}[\Delta_r(r_e)-a^2]},
\label{frequencyred}
\end{align}

\begin{align}
z_{\rm blue}=\frac{\sqrt{1-a^2\Lambda^{\prime}}[a(\Lambda^{\prime}r_e^2(r_e^2+a^2)+2r_e-e^2)-r_e^2\sqrt{\Delta_r(r_e)}](\pm\sqrt{1-\Lambda^{\prime}r_e^3-\frac{e^2}{r_e}})}
{r_e^{3/4}\sqrt{\frac{2e^2}{r_e^{1/2}}+r_e^{3/2}-3\sqrt{r_e}-r_e^{3/2}a^2\Lambda^{\prime}\pm 2a\sqrt{-\frac{e^2}{r_e}+r_e^3(\frac{1}{r_e^3}-\Lambda^{\prime})}}[\Delta_r(r_e)-a^2]}.
\label{frequencyblue}
\end{align}

\subsubsection{Redshift/blueshift for circular equatorial orbits in Kerr-de Sitter spacetime}\label{kdsredblue}
For zero electric charge, $e=0$, eqns.(\ref{shiftred})-(\ref{shiftblue}) reduce to:

\begin{align}
&z_{\rm red}=\frac{r_d^{3/4}\sqrt{r_d^{3/2}-3\sqrt{r_d}-r_d^{3/2}a^2\Lambda^{\prime}\pm
2a \sqrt{r_d^3(\frac{1}{r_d^3}-\Lambda^{\prime})}}}{r_e^{3/4}\sqrt{r_e^{3/2}-3\sqrt{r_e}-r_e^{3/2}a^2\Lambda^{\prime}\pm
2a \sqrt{r_e^3(\frac{1}{r_e^3}-\Lambda^{\prime})}}}\nonumber\\
&\times \frac{[a(-\Lambda^{\prime}r_e^2(r_e^2+a^2)-2r_e)-r_e^2\sqrt{\Delta_r(r_e)}](\pm[
r_e^{3/2}\sqrt{1-\Lambda^{\prime}r_d^3}-r_d^{3/2}\sqrt{1-\Lambda^{\prime}r_e^3}])}{
(r_d^{3/2}\pm a\sqrt{1-\Lambda^{\prime}r_d^3})[(\Delta_r(r_e)-a^2)r_d^{3/2}+
(ar_e^2+r_e^2\sqrt{\Delta_r(r_e)})(\pm \sqrt{1-\Lambda^{\prime}r_d^3})]},
\label{kerradsred}
\end{align}

\begin{align}
&z_{\rm blue}=\frac{r_d^{3/4}\sqrt{r_d^{3/2}-3\sqrt{r_d}-r_d^{3/2}a^2\Lambda^{\prime}\pm
2a \sqrt{r_d^3(\frac{1}{r_d^3}-\Lambda^{\prime})}}}{r_e^{3/4}\sqrt{r_e^{3/2}-3\sqrt{r_e}-r_e^{3/2}a^2\Lambda^{\prime}\pm
2a \sqrt{r_e^3(\frac{1}{r_e^3}-\Lambda^{\prime})}}}\nonumber\\
&\times \frac{[a(-\Lambda^{\prime}r_e^2(r_e^2+a^2)-2r_e)+r_e^2\sqrt{\Delta_r(r_e)}](\pm[
r_e^{3/2}\sqrt{1-\Lambda^{\prime}r_d^3}-r_d^{3/2}\sqrt{1-\Lambda^{\prime}r_e^3}])}{
(r_d^{3/2}\pm a\sqrt{1-\Lambda^{\prime}r_d^3})[(\Delta_r(r_e)-a^2)r_d^{3/2}+
(ar_e^2-r_e^2\sqrt{\Delta_r(r_e)})(\pm \sqrt{1-\Lambda^{\prime}r_d^3})]}.
\label{kerradsblue}
\end{align}
In Eqns (\ref{kerradsred}),(\ref{kerradsblue}) we define:
\begin{equation}
\Delta_r(r_e):=(1-\Lambda^{\prime}r_e^2)(r_e^2+a^2)-2r_e.
\end{equation}

In the particular case when the detector is located far away from the source and the detector radius $r_d$ is much larger than the black hole parameters , the frequency shifts (\ref{shiftred})-(\ref{shiftblue}) become:

\begin{align}
z_{\rm red}=\frac{\sqrt{1-a^2\Lambda^{\prime}}[a(\Lambda^{\prime}r_e^2(r_e^2+a^2)+2r_e)+r_e^2\sqrt{\Delta_r(r_e)}](\pm\sqrt{1-\Lambda^{\prime}r_e^3})}
{r_e^{3/4}\sqrt{r_e^{3/2}-3\sqrt{r_e}-r_e^{3/2}a^2\Lambda^{\prime}\pm 2a\sqrt{r_e^3(\frac{1}{r_e^3}-\Lambda^{\prime})}}[\Delta_r(r_e)-a^2]},
\label{frequencyred}
\end{align}

\begin{align}
z_{\rm blue}=\frac{\sqrt{1-a^2\Lambda^{\prime}}[a(\Lambda^{\prime}r_e^2(r_e^2+a^2)+2r_e)-r_e^2\sqrt{\Delta_r(r_e)}](\pm\sqrt{1-\Lambda^{\prime}r_e^3})}
{r_e^{3/4}\sqrt{r_e^{3/2}-3\sqrt{r_e}-r_e^{3/2}a^2\Lambda^{\prime}\pm 2a\sqrt{r_e^3(\frac{1}{r_e^3}-\Lambda^{\prime})}}[\Delta_r(r_e)-a^2]}.
\label{frequencyblue}
\end{align}

The frequency shifts (\ref{frequencyred})-(\ref{frequencyblue}) are plotted in Figs.(\ref{zredblueshifta052Lpos})-(\ref{zredblueshifta09939Lpos}) for different values of the spin of the central black hole and for positive cosmological constant for the corotating case.
As the radius increases, $z_{\rm red}\rightarrow -z_{\rm blue}$.

\begin{figure}
[ptbh]
\psfrag{GruenDeusdieAllee13a}{$a=0.52,\;\Lambda^{\prime}=10^{-33}.$}
\psfrag{re}{$r_e$} \psfrag{Zshift}{$z$}
\psfrag{zred}{$z_{\rm red}$}
\psfrag{zblue}{$z_{\rm blue}$}

\begin{center}
\includegraphics[height=2.4526in, width=3.3797in ]{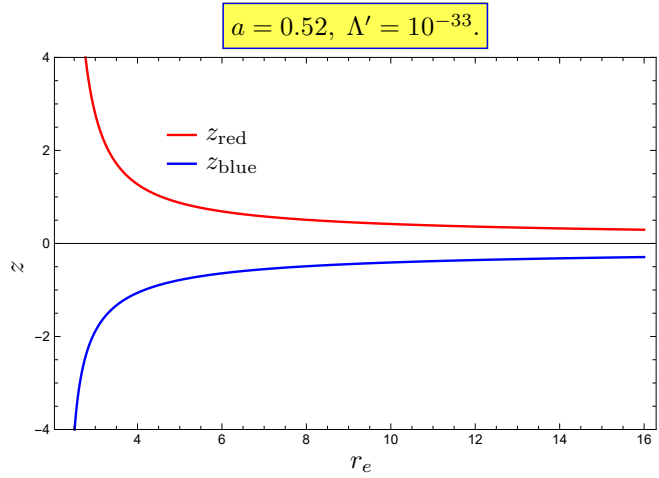}
 \caption{The functions $z_{\rm red},z_{\rm blue}$ as functions of radius $r_e$. The spin of the Kerr-Newman-de Sitter black hole was chosen as $a=0.52$ and the dimensionless cosmological parameter as $\Lambda^{\prime}=10^{-33}$.}%
\label{zredblueshifta052Lpos}%
\end{center}
\end{figure}

\begin{figure}
[ptbh]
\psfrag{GruenDeusdiegrosseAllee13a}{$a=0.9939,\;\Lambda^{\prime}=10^{-33}.$}
\psfrag{re}{$r_e$} \psfrag{Zshift}{$z$}
\psfrag{zred}{$z_{\rm red}$}
\psfrag{zblue}{$z_{\rm blue}$}

\begin{center}
\includegraphics[height=2.4526in, width=3.3797in ]{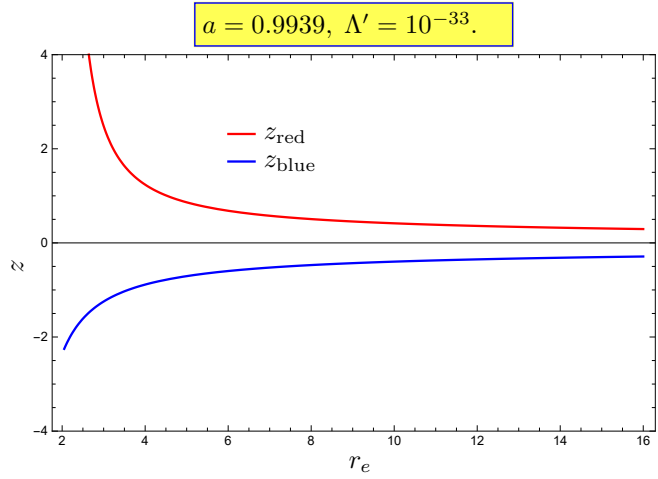}
 \caption{The functions $z_{\rm red},z_{\rm blue}$ as functions of radius $r_e$. The spin of the Kerr-Newman-de Sitter black hole was chosen as $a=0.9939$ and the dimensionless cosmological parameter as $\Lambda^{\prime}=10^{-33}$. }%
\label{zredblueshifta09939Lpos}%
\end{center}
\end{figure}

\section{More general orbits for rotating charged black holes}\label{nonequatorial}

\subsection{Spherical orbits in Kerr-Newman spacetime}

Depending on whether or not the coordinate radius $r$ is constant along
a given timelike geodesic, the corresponding particle orbit is characterised
as spherical or nonspherical, respectively. In this subsection we will focus on spherical non-equatorial orbits.

\subsubsection{Frame-dragging for timelike spherical orbits}
Assuming $\Lambda=0$ we derive from (\ref{azimueq}) and (\ref{polareq}) the following equation:
\begin{equation}
\frac{{\rm d}\phi}{{\rm d}\theta}=-
\frac{\frac{aP}{\Delta^{KN}}-aE+L/\sin^2\theta}{\sqrt{Q-
\frac{L^2\cos^2\theta}{\sin^2\theta}+a^2\cos^2\theta(E^2-1)}}.
\label{lenseFrameKN}
\end{equation}
In (\ref{lenseFrameKN}) $\Delta^{KN}:=r^2+a^2+e^2-2Mr$ \footnote{When we set $M=1$, $\Delta^{KN}=r^2+a^2+e^2-2r$.}. Also $P=E(r^2+a^2)-L a.$
Using the variable $z=\cos^2\theta$, $-\frac{1}{2}\frac{{\rm d}z}{\sqrt{z}}\frac{1}{\sqrt{1-z}}={\rm sgn}(\pi/2-\theta){\rm d}\theta$ we will determine for the \textit{first time} in closed analytic form the amount of frame-dragging for timelike spherical orbits in the Kerr-Newman spacetime \footnote{We should mention at this point  the extreme black hole solutions   for spherical timelike non-polar geodesics in Kerr-Newman spacetime obtained in \cite{johnston} in terms of formal integrals. Latitudinal motion in the Kerr and KN spacetimes has been studied in \cite{felicede} and complemented in \cite{bicakpnc}.}.

Thus for instance, expressed in terms of the new variable :
\begin{align}
\frac{L}{\sin^2\theta}\frac{{\rm d}\theta}{\sqrt{\Theta}}
&=\frac{L}{1-z}\left(-\frac{1}{2}\right)\frac{{\rm d}z}{\sqrt{z}}
\frac{1}{\sqrt{\alpha z^2-(\alpha+\beta)z+Q}}\nonumber \\
&=\frac{L}{1-z}\left(-\frac{1}{2}\right)\frac{{\rm d}z}{\sqrt{z}}
\frac{1}{|a|\sqrt{1-E^2}}\frac{1}{\sqrt{(z-z_{+})(z-z_{-})}},
\end{align}
where $\alpha=a^2(1-E^2),\beta=L^2+Q$. The range of $z$ for which the motion takes place includes the equatorial value, $z=0$:
\begin{equation}
0\leq z\leq z_-.
\end{equation}
The integral of (\ref{lenseFrameKN}) will be split into two parts: Proposition \ref{protasiEins} and Proposition \ref{protasizwei}.
We  prove first the  exact result of Proposition \ref{protasiEins}:
\begin{proposition}\label{protasiEins}
\begin{align}
&\int_0^{z_-}\frac{L}{1-z}\left(-\frac{1}{2}\right)\frac{{\rm d}z}{\sqrt{z}}
\frac{1}{|a|\sqrt{1-E^2}}\frac{1}{\sqrt{(z-z_{+})(z-z_{-})}}\nonumber\\
&=
-\frac{L}{|a|\sqrt{1-E^2}}\frac{\pi}{2}
(1-z_{-})^{-1}(z_+-z_-)^{-1/2}F_1\left(\frac{1}{2},1,\frac{1}{2},1,
\frac{z_-}{z_{-}-1},\frac{z_-}{z_--z_+}\right)\nonumber\\
&=-\frac{L}{|a|\sqrt{1-E^2}}\frac{\pi}{2}z_+^{-1/2}
F_1\left(\frac{1}{2},\frac{1}{2},1,1,\frac{z_-}{z_+},z_-\right),
\end{align}
where $F_1(\alpha,\beta,\beta^{\prime},\gamma,x,y)$ denotes the first Appell's hypergeometric function of two variables ($x$ and $y$), which admits the integral representation:
\begin{equation}
\int_0^1 u^{\alpha-1}(1-u)^{\gamma-\alpha-1}(1-u x)^{-\beta}(1-u
y)^{-\beta^{\prime}}{\rm
d}u=\frac{\Gamma(\alpha)\Gamma(\gamma-\alpha)}{\Gamma(\gamma)}
F_1(\alpha,\beta,\beta^{\prime},\gamma,x,y).
\label{frenchAppellIR}
\end{equation}
Also $\Gamma(a)$ denotes Euler's gamma function.
\end{proposition}
\begin{proof}
We compute first the integral:
\begin{align}
&\int_{z_j}^{z_-}\frac{L}{1-z}\left(-\frac{1}{2}\right)\frac{{\rm d}z}{\sqrt{z}}
\frac{1}{|a|\sqrt{1-E^2}}\frac{1}{\sqrt{(z-z_{+})(z-z_{-})}}\label{intermeintegral}.
\end{align}
Applying the transformation $z=z_-+\xi^2(z_j-z_-)$ in (\ref{intermeintegral}) we obtain
\begin{align}
&\int_{z_j}
^{z_-}\frac{L}{1-z}\left(-\frac{1}{2}\right)\frac{{\rm d}z}{\sqrt{z}}
\frac{1}{|a|\sqrt{1-E^2}}\frac{1}{\sqrt{(z-z_{+})(z-z_{-})}}\nonumber \\
&=\frac{-L}{2|a|\sqrt{1-E^2}}\frac{z_{-}-z_j}{(1-z_-)}\frac{1}{\sqrt{z_-}\sqrt{z_{-}-z_{+}}\sqrt{z_j-z_-}}\int_0^1\frac{{\rm d}x}{[1-x\left(\frac{z_j-z_-}{1-z_-}\right)]}\frac{1}{\sqrt{1-x\frac{z_{-}-z_j}{z_{-}-z_{+}}}}\frac{1}{\sqrt{x}}\frac{1}{\sqrt{1-x\frac{z_{-}-z_j}{z_{-}}}}\nonumber \\
&=\frac{-L}{2|a|\sqrt{1-E^2}}\frac{z_{-}-z_j}{(1-z_-)}\frac{1}{\sqrt{z_-}\sqrt{z_{-}-z_{+}}\sqrt{z_j-z_-}}\nonumber\\
&\times \frac{\Gamma\left(\frac{1}{2}\right)\Gamma(1)}{\Gamma\left(\frac{3}{2}\right)}
F_D\left(\frac{1}{2},1,\frac{1}{2},\frac{1}{2},\frac{3}{2},\frac{z_j-z_-}{1-z_-},\frac{z_{-}-z_j}{z_{-}},\frac{z_{-}-z_j}{z_{-}-z_+
}\right)\label{GuiseppeLauriFd},
\end{align}
where $x\equiv\xi^2$. In eqn.(\ref{GuiseppeLauriFd}) $F_D$ denotes Lauricella's fourth multivariable hypergeometric function (here is a function of three variables). The general multivariable Lauricella's function $F_D$  is defined as follows:
\begin{equation}
\fbox{$\displaystyle
F_D(\alpha,\mbox{\boldmath${\beta}$},\gamma,{\bf z})=
\sum_{n_1,n_2,\dots,n_m=0}^{\infty}\frac{(\alpha)_{n_1+\cdots
n_m}(\beta_1)_{n_1} \cdots (\beta_m)_{n_m}}
{(\gamma)_{n_1+\cdots+n_m}(1)_{n_1}\cdots (1)_{n_m}}
z_1^{n_1}\cdots z_m^{n_m}$} \label{GLauri}
\end{equation}%

where
\begin{eqnarray}
\mathbf{z} &=&(z_{1},\ldots ,z_{m}),  \notag \\
\mbox{\boldmath${\beta}$} &=&(\beta _{1},\ldots ,\beta _{m}).
\label{TUPLES}
\end{eqnarray}

The Pochhammer symbol
\fbox{$\displaystyle (\alpha)_m=(\alpha,m)$}
is defined by
\begin{equation}
(\alpha )_{m}=\frac{\Gamma (\alpha +m)}{\Gamma (\alpha )}=\left\{
\begin{array}{ccc}
1, &
{\rm if}
& m=0 \\
\alpha (\alpha +1)\cdots (\alpha +m-1) & \text{%
{\rm if}%
} & m=1,2,3%
\end{array}%
\right.
\end{equation}
The series (\ref{GLauri}) admits the following integral representation:

\begin{equation}
\fbox{$\displaystyle
F_D(\alpha,\mbox{\boldmath${\beta}$},\gamma,{\bf z})=
\frac{\Gamma(\gamma)}{\Gamma(\alpha)\Gamma(\gamma-\alpha)}
\int_0^1 t^{\alpha-1}(1-t)^{\gamma-\alpha-1}(1-z_1
t)^{-\beta_1}\cdots (1-z_m t)^{-\beta_m} {\rm d}t $}
\label{OloklAnapa}
\end{equation}
which is valid for
\fbox{$\displaystyle {\rm Re}(\alpha)>0,\;{\rm Re}(\gamma-\alpha)>0. $}%
. It
{\em converges\;absolutely}
inside the m-dimensional cuboid:%
\begin{equation}
|z_{j}|<1,(j=1,\ldots ,m).
\end{equation}
For $m=2$ $F_D$ in the notation of Appell becomes the two variable
hypergeometric function
$F_1(\alpha,\beta,\beta^{\prime},\gamma,x,y)$ with integral representation given by eqn.(\ref{frenchAppellIR}).
Setting $z_j=0$ in (\ref{GuiseppeLauriFd}) yields:
\begin{align}
&\int_0^{z_-}\frac{L}{1-z}\left(-\frac{1}{2}\right)\frac{{\rm d}z}{\sqrt{z}}
\frac{1}{|a|\sqrt{1-E^2}}\frac{1}{\sqrt{(z-z_{+})(z-z_{-})}}\nonumber\\
&=\frac{-L}{2|a|\sqrt{1-E^2}}(1-z_-)^{-1}(z_{+}-z_{-})^{-1/2}\frac{\Gamma\left(\frac{1}{2}\right)\Gamma(1)}{\Gamma\left(\frac{3}{2}\right)}
F_D\left(\frac{1}{2},1,\frac{1}{2},\frac{1}{2},\frac{3}{2},\frac{-z_-}{1-z_-},1,\frac{z_{-}}{z_{-}-z_+
}\right)\nonumber \\
&=\frac{-L}{2|a|\sqrt{1-E^2}}(1-z_-)^{-1}(z_{+}-z_{-})^{-1/2}\frac{\Gamma(1/2)^2\Gamma\left(\frac{3}{2}\right)}{\Gamma\left(\frac{3}{2}\right)}
F_1\left(\frac{1}{2},1,\frac{1}{2},1,\frac{-z_-}{1-z_-},\frac{z_{-}}{z_{-}-z_+
}\right)\nonumber \\
&=\frac{-L}{2|a|\sqrt{1-E^2}}z_{+}^{-1/2}\pi F_1\left(\frac{1}{2},\frac{1}{2},1,1,\frac{z_-}{z_+},z_-\right).\label{endiamfdrag}
\end{align}
\end{proof}

For producing the result in the last line of equation (\ref{endiamfdrag}), we used the following transformation property of Appell's hypergeometric function $F_1$:
\begin{lemma}
\begin{align}
&y^{1+\beta-\gamma}(1-y)^{\gamma-\alpha-1}(x-y)^{-\beta}F_1\left(1-\beta^{\prime},\beta,1+\alpha-\gamma,2+\beta-\gamma,
\frac{y}{y-x},\frac{y}{y-1}\right)\nonumber\\
&=y^{1+\beta-\gamma}x^{-\beta}F_1\left(1+\beta+\beta^{\prime}-\gamma,\beta,1+\alpha-\gamma,2+\beta-\gamma,
\frac{y}{x},y\right).
\end{align}
\end{lemma}

On the other hand we compute analytically the second integral that contributes to frame-dragging and we obtain:
\begin{proposition}\label{protasizwei}
\begin{align}
&\int_0^{z_-}\frac{\frac{aP}{\Delta^{KN}}-aE}{\sqrt{\Theta^{\prime}}}{\rm d}\theta\nonumber \\
&=\left(\frac{aP}{\Delta^{KN}}-aE\right)\frac{1}{|a|\sqrt{1-E^2}}\left(
\frac{-1}{2}\right)\frac{\Gamma\left(\frac{1}{2}\right)\Gamma\left(\frac{1}{2}\right)}
{\sqrt{z_+-z_{-}}}F\left(\frac{1}{2},\frac{1}{2},1,-\frac{z_-}{z_+-z_-}\right)\nonumber \\
&=\left(\frac{aP}{\Delta^{KN}}-aE\right)\frac{1}{|a|\sqrt{1-E^2}}\left(-\frac{\pi}{2}
\right)\frac{1}{\sqrt{z_+}}F\left(\frac{1}{2},\frac{1}{2},1,\frac{z_-}{z_+}\right).
\label{Gaussinclude}
\end{align}
\end{proposition}
In eqn.(\ref{Gaussinclude}) $F(a,b,c,z)$ with $a,b,c\in \mathbb{R}$ and $c\not \in \mathbb{Z}_{\leq0}$ denotes the Gau\ss' hypergeometric function which is defined by:
\begin{equation}
F(a,b,c,z)=\sum_{n=0}^{\infty}\frac{(a)_n(b)_n}{(c)_n n!}z^n.
\end{equation}
We thus obtain the following fundamental result, in closed analytic form, for the amount of frame-dragging that a timelike spherical orbit in Kerr-Newman spacetime undergoes.
 \begin{theorem}\label{LenseThirFDRAG}
 As $\theta$ goes through a quarter of a complete oscillation we obtain the change in azimuth $\phi$, $\Delta\phi^{\rm GTR}$:
\begin{align}
\Delta\phi^{\rm GTR}&=\frac{L}{|a|\sqrt{1-E^2}}\frac{\pi}{2}z_+^{-1/2}
F_1\left(\frac{1}{2},\frac{1}{2},1,1,\frac{z_-}{z_+},z_-\right)\nonumber \\
&+\left(\frac{aP}{\Delta^{KN}}-aE\right)\frac{1}{|a|\sqrt{1-E^2}}\left(\frac{\pi}{2}
\right)\frac{1}{\sqrt{z_+}}F\left(\frac{1}{2},\frac{1}{2},1,\frac{z_-}{z_+}\right).
\label{GVKframePrecessionKN}
\end{align}
\end{theorem}
\begin{proof}
Using propositions \ref{protasiEins} and \ref{protasizwei}, we prove eqn (\ref{GVKframePrecessionKN}) and Theorem \ref{LenseThirFDRAG}.
\end{proof}

In eqn.(\ref{GVKframePrecessionKN}) the quantities $z_{\pm}$ are given by:
\begin{equation}
z_{\mp}=\frac{a^2(1-E^2)+L^2+Q\mp\sqrt{(a^2(E^2-1)-L^2-Q)^2-4a^2(1-E^2)Q}}{2a^2(1-E^2)}.
\end{equation}
\subsubsection{Conditions for spherical Kerr-Newman orbits}
In order for a non-equatorial spherical (NES) orbit in Kerr-Newman spacetime to exist at radius $r$, the conditions $R(r)=\frac{{\rm d}R}{{\rm d}r}=0$ must hold at this radius. As in section \ref{FiIntegralsKNdSequatocirc} where we investigated equatorial circular geodesics, we solve these two equations simultaneously, and the solutions now take an elegant compact form when parametrised in terms of $r$ and Carter's constant $Q$. There are four classes of solutions $(E_i,L_i)$ which we label by $i=\alpha,\beta,\gamma,\delta$.  We obtain the following novel relations for the first two:
\begin{align}
E_{\alpha,\beta}(r,Q;a,e,M)&=\frac{e^2 r^2+r^3(r-2M)-a(aQ\mp\sqrt{\Upsilon})}{r^2\sqrt{2e^2r^2+r^3(r-3M)-2a(aQ\mp\sqrt{\Upsilon})}}\label{NESPenergyp}\\
L_{\alpha,\beta}(r,Q;a,e,M)&=-\frac{(r^2+a^2)(\mp\sqrt{\Upsilon}+a Q)+2aMr^3-ae^2r^2}{r^2\sqrt{2e^2r^2+r^3(r-3M)-2a(aQ\mp\sqrt{\Upsilon})}},
\label{AngulMomenNESPher}
\end{align}
where
\begin{equation}
\Upsilon\equiv a^2 Q^2+r^2[-2e^2Q+3QrM-(e^2+Q)r^2+r^3M].
\end{equation}
The third and fourth classes of solutions are related to the first two by:
\begin{equation}
(E_{\gamma,\delta},L_{\gamma,\delta})=-(E_{\alpha,\beta},L_{\alpha,\beta}).
\end{equation}

\tikzstyle{mybox} = [draw=red, fill=white!20, very thick,
    rectangle, rounded corners, inner sep=10pt, inner ysep=20pt]
\tikzstyle{fancytitle} =[fill=red, text=white]
\begin{center}
\begin{tikzpicture}
\node [mybox] (box){%
    \begin{minipage}{1.0\textwidth}
    \begin{theorem}\label{EndiamesoORO}
    Using dimensionless parameters or equivalently setting $M=1$, the constants of motion that solve the conditions for spherical Kerr-Newman orbits are given by the equations:
\begin{align}
E_{\alpha,\beta}(r,Q;a,e)&=\frac{e^2 r^2+r^3(r-2)-a(aQ\mp\sqrt{\Upsilon})}{r^2\sqrt{2e^2r^2+r^3(r-3)-2a(aQ\mp\sqrt{\Upsilon})}}\label{NESPenergyp}\\
L_{\alpha,\beta}(r,Q;a,e)&=-\frac{(r^2+a^2)(\mp\sqrt{\Upsilon}+a Q)+2ar^3-ae^2r^2}{r^2\sqrt{2e^2r^2+r^3(r-3)-2a(aQ\mp\sqrt{\Upsilon})}},
\label{AngulMomenNESPher}
\end{align}
where
\begin{equation}
\Upsilon\equiv a^2 Q^2+r^2[-2e^2Q+3Qr-(e^2+Q)r^2+r^3].
\end{equation}
\end{theorem}
    \end{minipage}

};
\end{tikzpicture}%
\end{center}

Equations (\ref{NESPenergyp}),(\ref{AngulMomenNESPher}), of Theorem \ref{EndiamesoORO}, for zero electric charge ($e=0$) reduce correctly to the corresponding ones in Kerr spacetime \cite{Rana},\cite{TeoB}.

In order that the smaller square root appearing in the solutions (\ref{NESPenergyp}),(\ref{AngulMomenNESPher}) is real, we need to impose the condition
\begin{align}
&\Upsilon \geq 0\\
\Leftrightarrow\; & a^2 Q^2+Q (-2e^2 r^2+3 r^3M-r^4)-e^2r^4+r^5 M\geq 0
\label{trionimoQ}
\end{align}
Since $\Upsilon$ is quadratic in Carter's constant $Q$ its two roots are:
\begin{equation}
Q_{1,2}=\frac{r^2}{2a^2}(r(r-3M)+2e^2\pm\sqrt{\mathfrak{S}}),
\end{equation}
where
\begin{equation}
\mathfrak{S}\equiv r^4-6 M r^3+9M^2 r^2+4e^4-4e^2r(3M-r)+4a^2(e^2-r M).
\end{equation}
Now $\mathfrak{S}$ is a quartic equation in $r$ that is related to the study of circular photon orbits in the equatorial plane around a Kerr-Newman black hole \cite{ChechZS}. The two largest roots of $\mathfrak{S}$ $r_1,r_2$ are the radii of the prograde abd retrograde photon orbits, respectively. They lie in the ranges $M\leq r_1<3M<r_2\leq 4M$. The locations of these two photon orbits will be significant in what follows, as they will demarcate the allowed radii of the timelike spherical orbits. It is useful to note that $\mathfrak{S}$ is negative in the range $r_1<r<r_2$. This means that real solutions for $Q_{1,2}$ only exist outside this range. We have checked that $Q_2\leq Q_1<0$ when $r\leq r_1$, and $0<Q_2\leq Q_1$ when $r\geq r_2$.

Since $\Upsilon$ is quadratic in $Q$  with positive leading coefficient ($a^2>0$), inequality (\ref{trionimoQ}) is satisfied when $Q\leq Q_2$ or $Q\geq Q_1$. Thus when $r<r_1$ or $r>r_2$, the allowed ranges for Carter's constant are $Q\leq Q_2$ and $Q\geq Q_1$. On the other hand, when $r_1\leq r\leq r_2$, there is no restriction on the range of $Q$.

For the solutions (\ref{NESPenergyp}),(\ref{AngulMomenNESPher}) to be valid, the larger square root appearing in them also has to be real. Thus we must impose the condition:
\begin{equation}
\Gamma_{\alpha,\beta}\equiv 2 e^2 r^2+r^3(r-3M)-2a (a Q\mp \sqrt{\Upsilon})\geq 0.
\label{bigsquareroot}
\end{equation}

We first find the values of $r$ and $Q$ for which $\Gamma_{\alpha,\beta}=0$. Initially we obtain:
\begin{equation}
\Gamma_{\alpha}\Gamma_{\beta}=r^4\mathfrak{S},
\end{equation}
thus we see that $\Gamma_{\alpha}$ may vanish only if $r=r_1$ or $r=r_2$, and similarly for $\Gamma_{\beta}$.
For either value of $r$, we have checked that $\Gamma_{\alpha}=0$ if $Q\geq Q_2$ and $\Gamma_{\alpha}(r_{1,2})>0$ if $Q<Q_2(r_{1,2})$. On the other hand,
$\Gamma_{\beta}=0$ if $Q\leq Q_2$, and $\Gamma_{\beta}<0$ if $Q>Q_2$. Since we want to impose $\Gamma_{\alpha,\beta}>0$, we deduce that the first class of solutions ($\alpha$) is allowed only if $Q<Q_2$. The second class of solutions ($\beta$) is not allowed at all.

When $r<r_1$ or $r>r_2$ recall that $Q$ is allowed to take the ranges $Q\leq Q_2$ and $Q\geq Q_1$. Note that $\Gamma_{\alpha,\beta}=\sqrt{r^4\mathfrak{S}}$ when $Q=Q_2$, and $\Gamma_{\alpha,\beta}=-\sqrt{r^4\mathfrak{S}}$ when $Q=Q_1$. Thus we deduce that $\Gamma_{\alpha,\beta}>0$ if $Q\leq Q_2$ and  $\Gamma_{\alpha,\beta}<0$ if $Q\geq Q_1$. Consequently, the range $Q\geq Q_1$ is ruled out for these cases.

When $r_1<r<r_2$, recall that there is no restriction on the range of $Q$. We observe that $\Gamma_{\alpha,\beta}=\pm \sqrt{-r^4\mathfrak{S}}$ when $Q=\frac{r^2}{2a^2}(2e^2+r(r-3M))$. Thus we conclude that $\Gamma_{\alpha}>0$ and $\Gamma_{\beta}<0$ for any value of $Q$. As a result, the second class of solutions ($\beta$) is not allowed in this case.

We also remark that the condition $\Gamma_{\alpha,\beta}>0$ will imply that the numerator of (\ref{NESPenergyp}) is positive. Indeed, from (\ref{bigsquareroot}) we deduce:
\begin{align}
&2 e^2 r^2+r^3(r-3M)-2a (a Q\mp \sqrt{\Upsilon})>0\nonumber \\
\Leftrightarrow\;& -a (a Q\mp \sqrt{\Upsilon})>-\frac{1}{2}r^3(r-3M)-e^2 r^2
\end{align}
so that
\begin{equation}
e^2r^2+r^3(r-2M)-a(aQ\mp\sqrt{\Upsilon})>\frac{1}{2}r^3(r-M)>0.
\end{equation}

The reality conditions from the square roots allow for the Carter constant $Q$ to be negative. In the uncharged Kerr black hole it was argued that negative $Q$ for spherical orbits were not allowed \cite{TeoB}. In fact a necessary condition for negative $Q$ with $E^2>1$  was presented in \cite{TeoB}:
\begin{equation}
a^2(1-E^2)+Q+L^2<0.
\label{qkerrNpos}
\end{equation}
For the Kerr black hole it was argued that this condition cannot be satisfied thus ruling out the case of negative $Q$ for spherical orbits \cite{TeoB}.
Following a similar analysis with \cite{TeoB}, we will discuss now the non-negativity of $Q$ for the more general Kerr-Newman black hole.
Substituting (\ref{NESPenergyp}),(\ref{AngulMomenNESPher}) into the left-hand side of
(\ref{qkerrNpos}) we obtain:
\begin{equation}
a^2(1-E^2_{\alpha,\beta})+Q+L^2_{\alpha,\beta}=\frac{\Psi^{KN}_{\alpha,\beta}}{r^2\Gamma_{\alpha,\beta}},
\label{signKarter}
\end{equation}
where
\begin{align}
\Psi^{KN}_{\alpha,\beta}\equiv & r^3(r-M)(Mr^3+a^2Q)-a^2\Upsilon+(Mr^3+a^2Q\mp 2a\sqrt{\Upsilon})^2 \nonumber \\
&+e^2 r^2(-r^4+r^2a^2\pm 2a\sqrt{\Upsilon}).
\end{align}
The ranges of Carter's constant $Q$ we have found so far ensure that the quantities $\Gamma_{\alpha,\beta}$ are positive. It remains to show that $\Psi^{KN}_{\alpha,\beta}$ are also positive for these ranges.

We have checked by  explicit calculation that $\Psi^{KN}_{\alpha}\Psi^{KN}_{\beta}$ is quadratic in the invariant $Q$. For fixed r, $\Psi^{KN}_{\alpha}$ may vanish only at the roots of this quadratic, and likewise for $\Psi^{KN}_{\beta}$. We have checked that the discriminant of this quadratic is negative, thus both roots are complex numbers in this case. We infer that $\Psi^{KN}_{\alpha,\beta}$ do not vanish for any real value of $Q$. Since $\Psi^{KN}_{\alpha,\beta}$ are continuous functions of $Q$ for the ranges we are interested in, it follows that their signs do not change as $Q$ is varied. Thus we still have to check that $\Psi^{KN}_{\alpha}$ and/or $\Psi^{KN}_{\beta}$ are positive for a specific value of $Q$ in each of the allowed ranges we have found.

When $r<r_1$ or $r>r_2$, recall that $Q$ is allowed to take the range $Q\leq Q_2$. We compute:
\begin{align}
\Psi^{KN}_{\alpha,\beta}(Q_2)&=\Psi^{KN}_{\alpha,\beta}(\frac{r^2}{2a^2}(r(r-3M)+2e^2-\sqrt{\mathfrak{S}}))\nonumber \\
&=\frac{1}{2}r^4(r(r-M)+e^2-\sqrt{\mathfrak{S}})^2+Mr^5\Delta^{KN}\nonumber \\
&+\frac{2e^2r^4(2rM-e^2-2r^2)}{4},
\end{align}
which is positive. Thus $\Psi^{KN}_{\alpha,\beta}>0$ if $Q\leq Q_2$

When $r=r_1$ or $r=r_2$, recall that Carter's constant $Q$  is allowed to take the range $Q<Q_2$, although only the first class of solutions $(\alpha)$ is allowed. The above argument is still valid by continuity, and we deduce that $\Psi^{KN}_{\alpha}>0$ if $Q<Q_2$.

When $r_1<r<r_2$ we found that there is no restriction on the range of $Q$, although only the first class of solutions ($\alpha$) is allowed. We compute for $Q=\frac{r^2}{2a^2}(r(r-3M)+2e^2)$:
\begin{align}
\Psi^{KN}_{\alpha}(\frac{r^2}{2a^2}(r(r-3M)+2e^2))&=\frac{1}{2}r^4(r(r-M)+e^2-\sqrt{\mathfrak{-S}})^2\nonumber \\
&+Mr^5\Delta^{KN}+\frac{2e^2r^4(2rM-e^2-2r^2)}{4},
\end{align}
which is positive. Thus $\Psi^{KN}_{\alpha}>0$ for any value of $Q$.

Thus we have shown that (\ref{signKarter}) is positive for all the ranges of Carter's constant of motion $Q$. The condition (\ref{qkerrNpos}) does not hold, in particular, when $Q$ is negative. This rule out all values of $Q$ which lie in the negative range.

Having obtained exact analytic solutions of the geodesic equations and solved the conditions for spherical orbits, we present examples of particular timelike spherical orbits in Kerr-Newman spacetime in Tables \ref{ProgradeKNnp}, \ref{retrogradeKNnpOr}. In Table \ref{ProgradeKNnp} we display examples of prograde orbits: $\Delta\phi^{GTR}>0$ when $L>0$
while in Table \ref{retrogradeKNnpOr} retrograde orbits: $\Delta\phi^{GTR}<0$ when $L<0$, are exhibited. For a choice of values for the black hole parameters $a,e$ and for Carter's constant $Q$, the constants of motion associated with the Killing vector symmetries of Kerr-Newman spacetime are computed from formulae (\ref{NESPenergyp}),(\ref{AngulMomenNESPher}). In the last column of Tables \ref{ProgradeKNnp}, \ref{retrogradeKNnpOr} we present the change $\Delta\phi^{GTR}$ over one complete oscillation in latitude of the orbit, which means we multiply Eqn.((\ref{GVKframePrecessionKN})) by four. In section \ref{asteraspole} we will study spherical polar orbits ( a special case with zero angular momentum $L=0$ ) and we will determine the value of Carter's constant $Q$ separating prograde and retrograde orbits.

\begin{table}[tbp] \centering
\begin{tabular}
[c]{cccccccc}\hline
Orbit & $r/M$ & $Q/M^2$ & $E$ & $L/M$ & $e/M$  & $a/M$ & $\Delta\phi^{GTR}$ \\ \hline
(a)   & $7$ & $2$ & $0.9329746347090826$ &$2.6203403438437856$ & $0$ & $0.999999$ & $6.85204$\\
(b) & $7$ & $2$ & $0.9329917521582036$ & $2.6193683198162256$ & $0.11$& $0.9939$ & $6.84936$ \\
(c) & $7$ & $2$ & $0.9330249803496844$ & $2.6231975525737035$ &$0$& $0.9939$
& $6.84944$ \\
(d) & $7$ & $2$ & $0.9347100491303167$ & $2.71403958181262$ & $0.11$ & $0.8$
& $6.76005$\\
(e)  & $7$ & $2$ & $0.9346178572238036$ & $2.7048064351873267$ & $0.2$ & $0.8$ & $6.75993$
\\
(f) & $4$ & $8$ & $0.9187670265667813$ & $0.9264577651658418$ & $0$ & $0.9939$& $7.77146$\\
(g) & $4$ & $8$ & $0.9185597216627694$ & $0.9130252832603007$ &$0.11$ & $0.9939$ & $7.77188$\\
(h) & $4$ & $8$ & $0.9263820446220102$ & $1.1924993953748733$ & $0$ &$0.8$ &
$7.49676$\\
(i) & $4$ & $8$ & $ 0.9253520786762164$ & $1.1444194140656208$ & $0.2$ &$0.8$
& $7.49741$\\
\end{tabular}
\caption{Properties including frame-dragging (Lense-Thirring precession) of spherical timelike Kerr-Newman (prograde) orbits, applying the exact analytic formula (\ref{GVKframePrecessionKN}) and the first branch of solutions for the  constants of motion in Eqns. (\ref{NESPenergyp}),(\ref{AngulMomenNESPher}).}\label{ProgradeKNnp}
\end{table}

\begin{table}[tbp] \centering
\begin{tabular}
[c]{cccccccc}\hline
Orbit & $r/M$ & $Q/M^2$ & $E$ & $L/M$ & $e/M$  & $a/M$ & $\Delta\phi^{GTR}$ \\ \hline
(a)   & $7$ & $12$ & $0.9500313121192591$ &$-1.3504453996037784$ & $0$ & $0.999999$ & $-5.58661$\\
(b)  & $7$ & $12$ & $0.9499790519730756$ & $-1.3443270259214553$ & $0$ & $0.9939$ & $-5.59099$\\
(c) & $7$ & $12$ & $0.9497202606484783$ & $-1.313994954227598$ & $0.11$ &$0.9939$ & $-5.59138$ \\
(d)  & $7$ & $12$ & $0.9482593233936484$ & $-1.1250483727157161$ &$0.11$ &$0.8$ & $-5.72968$ \\
(e) &$7$ & $12$ & $0.9477180186254757$ & $-1.046386864788553$ & $0.2$ & $0.8$ & $-5.73048$\\
(f)  & $10$ & $1$ & $0.9626040652152477$ & $-4.1133141793233525$ & $0$ & $0.9939$ & $-5.8465$\\
(g) & $10$ & $1$ & $0.9625680840381455$ & $-4.109699623451054$ & $0.11$ & $0.9939$ & $-5.84647$\\
(h)  & $10$ & $1$ & $0.9611085421961251$ & $-4.008398488785058$ & $0.11$ & $0.8$ & $-5.93674$\\
(i) & $10$ & $1$ & $0.961033868172724$ & $-4.000352717478337$ & $0.2$ & $0.8$
& $-5.93669$\\

\end{tabular}
\caption{Properties including frame-dragging (Lense-Thirring precession) of spherical timelike Kerr-Newman (retrograde) orbits, applying the exact analytic formula (\ref{GVKframePrecessionKN}) and the second branch of solutions for the  constants of motion in Eqns. (\ref{NESPenergyp}),(\ref{AngulMomenNESPher}).}\label{retrogradeKNnpOr}
\end{table}

\subsubsection{Frame-dragging for spherical polar Kerr-Newman geodesics}\label{asteraspole}

In this subsection we shall investigate the physically important class of \textit{polar} orbits of massive particles, i.e. timelike geodesics crossing the symmetry axis of the Kerr-Newman spacetime.
The first integral in Eqn.(\ref{polareq}) for zero cosmological constant reads:
\begin{equation}
\rho^4\dot{\theta}^2=Q+a^2(E^2-1)\cos^2\theta-L^2\cot^2\theta.
\label{expresspolar}
\end{equation}
It follows from (\ref{expresspolar}) that in order for the orbit to reach the polar axis, where $\cos^2\theta=1$, it is necessary that \cite{Tsoubelis}
\begin{equation}
L=0.
\label{sinthikipoliki}
\end{equation}
The polar orbits considered in this subsection, mean that our test particle is supposed not only to cross the symmetry axis but also to sweep the whole range of the angular coordinate $\theta$. Therefore, we demand that $\dot{\theta}$ does not vanish for any $\theta\in [0,\pi]$. According to (\ref{expresspolar}) and (\ref{sinthikipoliki}) , this condition is equivalent to the demand that
\begin{equation}
Q>0\;\;\;\;\text{when}\;\;\;\;E^2\geq 1
\end{equation}
and
\begin{equation}
Q>a^2(1-E^2)\;\;\;\;\text{when}\;\;\;\;E^2<1.
\end{equation}
Spherical polar orbits (i.e polar orbits with constant radius) are given by the conditions $R(r)=0,\frac{{\rm d}R}{{\rm d}r}=0$. Equivalently, polar spherical orbits correspond to local extrema of the following effective potential:
\begin{equation}
V^2_{eff}(r):=\frac{\Delta^{KN}(r^2+K)}{(r^2+a^2)^2},
\label{effectPotPolarKN}
\end{equation}
where $K$ is the hidden constant of motion which for polar geodesics reads $K=Q+a^2E^2$. The local extrema of $V^2_{eff}(r)$ occur at the roots of the equation:
\begin{equation}
e^2r(r^2+a^2)-e^2(2r^3+2rK)+Mr^4-(K-a^2)r^3+3M(K-a^2)r^2-(K-a^2)a^2r-MKa^2=0.
\label{sinthiki}
\end{equation}
The general features of the $r$ motion of a test particle in polar orbit about the source of the Kerr-Newman field can be deduced from graphs of $V^2_{eff}(r)$, such as those in Figs.\ref{EffPotKNpolar1a08e02},\ref{EffPotKNpolar1a052e085}. As can be deduced from the graphs, a test particle follows a bound orbit when its specific energy at infinity $E$ is such that $E^2<1$. It can also be seen that, when $E^2<1$  and a black hole is involved, the particle gets dragged, disappearing behind the event horizon, i.e., the surface $r=r_+$, unless $K$ is such that $V^2_{eff}(r)$ develops a local maximum $E^2_{\max}$ outside the even horizon and $E^2<E^2_{\max}$. In this case, the particle will not be swallowed by the black hole, provided it is initially found in the region $r>r_0$, where $r_0$ is the point on the $r$ axis where $V^2_{eff}(r)=E^2_{\max}$. Moreover, in the case of a rotating charged black hole with vanishing cosmological constant, i.e. when  $e^2\leq M^2-a^2$, $V^2_{eff}(r)$ necessarily vanishes at the radii of the event and Cauchy horizons $r_{\pm}=M\pm\sqrt{M^2-a^2-e^2}$ which are the roots of the equation
$\Delta^{KN}=0$.
Considering the case of spherical orbits, we note that $E^2=V^2_{eff}(r_0)$, where $r_0$ is a root of (\ref{sinthiki}), is a necessary condition for a spherical orbit to obtain. Thus, with the  aid (\ref{sinthiki}) we derive the following relations for the specific energy and Carter's constant:
\begin{align}
E^2&=\frac{r(\Delta^{KN})^2}{(r^2+a^2)Z^{KN}}\label{energiapolars},\\
Q&=\frac{Mr^4+a^2r^3-3Ma^2r^2+a^4r-e^2r^3+e^2ra^2}{Z^{KN}}-a^2E^2\label{consCartpols},
\end{align}
where
\begin{equation}
Z^{KN}:=2e^2r+r^3-3Mr^2+a^2r+Ma^2.
\label{zpolsphekn}
\end{equation}

In addition, using the same techniques as in the proof of Theorem \ref{LenseThirFDRAG}, we derive the following Lense-Thirring precession $\Delta\phi^{GTR}_{Polar}$ per revolution for a polar spherical orbit in Kerr-Newman spacetime:
\begin{equation}
\fbox{$\displaystyle\Delta\phi^{GTR}_{Polar}=4\left(\frac{aP}{\Delta^{KN}}-aE\right)\frac{1}{|a|\sqrt{1-E^2}}\left(\frac{\pi}{2}
\right)\frac{1}{\sqrt{z_+}}F\left(\frac{1}{2},\frac{1}{2},1,\frac{z_-}{z_+}\right),$}
\label{ltfdpolarsphekn}
\end{equation}
where $P=E(r^2+a^2)$ for $L=0$, $z_+=\frac{Q}{a^2(1-E^2)},z_-=1$. Eqn (\ref{ltfdpolarsphekn}) simplifies to:
\begin{equation}
\fbox{$\displaystyle\Delta\phi^{GTR}_{Polar}=\frac{4aE(2Mr-e^2)}{\Delta^{KN}}\frac{\pi}{2}
\frac{1}{\sqrt{Q}}F\left(\frac{1}{2},\frac{1}{2},1,\frac{a^2(1-E^2)}{Q}\right).$}
\label{frameltpreceKNpolas}
\end{equation}
The exact analytic result in Eqn(\ref{frameltpreceKNpolas}) shows that during a revolution of the particle along the polar orbit the lines of nodes advance in a direction which coincides with that of the rotation of the central Kerr-Newman black hole. Thus this is a typical dragging effect, so that if the rotation vanishes, the change in azimuth $\phi$ vanishes too.

Using the convergent series expansion around the origin ($k^2\equiv\frac{a^2(1-E^2)}{Q}<1$ for bound orbits with $E^2<1$) for the Gau\ss's hypergeometric function we obtain:
\begin{equation}
\Delta\phi^{GTR}_{Polar}=\frac{4aE(2Mr-e^2)}{\Delta^{KN}}\frac{\pi}{2}
\frac{1}{\sqrt{Q}}\left(1+\frac{1}{4}\frac{a^2(1-E^2)}{Q}+\frac{9}{64}\left(\frac{a^2(1-E^2)}{Q}\right)^2+\cdots\right)
\end{equation}
Also using Eqns(\ref{energiapolars})-(\ref{zpolsphekn}) we express the variable of the hypergeometric function $k^2$ as follows:
\begin{equation}
k^2=\left(\frac{a}{r}\right)^2\frac{[r^4+2r^2a^2+a^4-4Mr^3-\frac{re^4}{M}+4r^2e^2]}{
r^4+a^4+2a^2r^2-4Mra^2-\frac{2re^2a^2}{M}-\frac{a^4e^2}{rM}-\frac{a^2e^4}{rM}+4a^2e^2-\frac{e^2r^3}{M}}.
\end{equation}

\begin{figure}
[ptbh]

\psfrag{logr}{$\log\left(\frac{r}{M}\right)$} \psfrag{Veff}{$V^2_{eff}(r)$}
\psfrag{a0.8,e0.2,K8}{$a=0.8,e=0.2,K=8$}
\psfrag{a0.8,e0.2,K16}{$a=0.8,e=0.2,K=16$}
\psfrag{a0.8,e0.2,K24}{$a=0.8,e=0.2,K=24$}

\begin{center}
\includegraphics[height=2.4526in, width=3.3797in ]{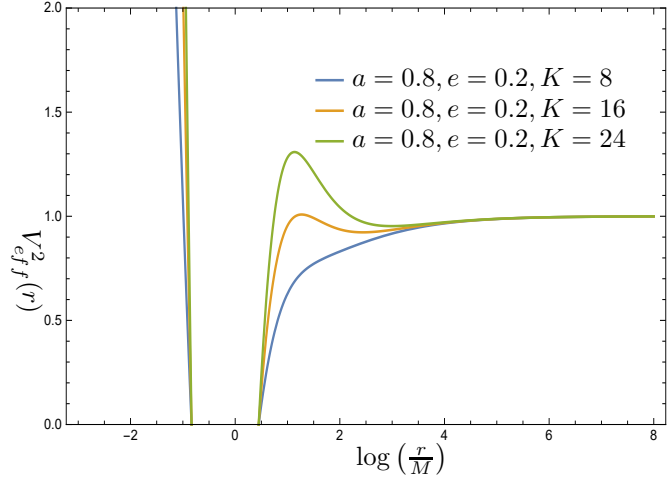}
 \caption{The effective square potential $V^2_{eff}(r)$ for polar spherical orbits in the neighbourhood of a rotating charged black hole with $a=0.8,e=0.2$. The parameter $K$ distinguishing the three curves is Carter's hidden integral of motion in the Kerr-Newman black hole.  }%
\label{EffPotKNpolar1a08e02}%
\end{center}
\end{figure}

\begin{figure}
[ptbh]

\psfrag{logr}{$\log\left(\frac{r}{M}\right)$} \psfrag{Veff}{$V^2_{eff}(r)$}
\psfrag{a0.52,e0.85,K8}{$a=0.52,e=0.85,K=8$}
\psfrag{a0.52,e0.85,K16}{$a=0.52,e=0.85,K=16$}
\psfrag{a0.52,e0.85,K24}{$a=0.52,e=0.85,K=24$}
\psfrag{a0.52,e0,K8}{$a=0.52,e=0,K=8$}
\psfrag{a0.52,e0,K16}{$a=0.52,e=0,K=16$}
\psfrag{a0.52,e0,K24}{$a=0.52,e=0,K=24$}

\begin{center}
\includegraphics[height=2.4526in, width=3.3797in ]{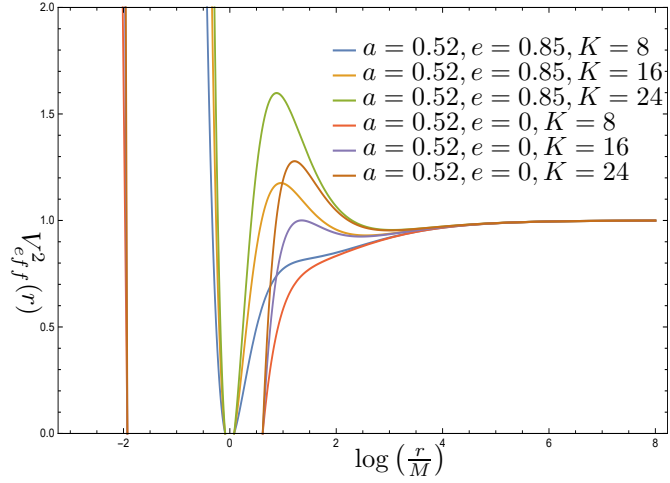}
 \caption{The effective square potential $V^2_{eff}(r)$ for polar spherical orbits in the neighbourhood of a rotating charged black hole with $a=0.52,e=0.85$. For comparison we also plot the case of a Kerr black hole with $a=52,e=0$. The parameter $K$ distinguishing the corresponding curves is Carter's hidden integral of motion in the Kerr-Newman black hole.  }%
\label{EffPotKNpolar1a052e085}%
\end{center}
\end{figure}

\begin{table}[tbp] \centering
\begin{tabular}
[c]{ccccccc}\hline
Orbit & $r/M$ & $Q/M^2$ &$E$ &$e/M$  & $a/M$& $\Delta\phi^{GTR}_{Polar}$ \\ \hline
(a)   & $10$ & $14.198012556704303$ & $0.955967687347955$ &$0$ & $0.8$& $0.316593=18.1394^{\degree}$\\
(b)   & $10$ & $14.175980906902574$ & $0.9559475468735232$ & $0.11$ & $0.8$& $0.316593=18.1395^{\degree}$\\
(c)   & $10$ & $14.125237970658672$ & $0.9559012400826655$ & $0.2$ & $0.8$& $0.316595=18.1396^{\degree}$\\
(d) & $10$ & $14.149643079211955$ & $0.9558495321451093$ & $0$ & $0.999999$ &$0.394826=22.6219^{\degree}$ \\
(e) & $10$ & $14.149642728218016$ & $0.9558495318270006$ & $0.00044$ & $0.999999$
& $0.394826=22.6219^{\degree}$\\
(f) & $10$ & $14.248389653442322$ & $0.9560911507901617$ & $0$ & $0.52$ & $0.206277=11.8188^{\degree}$\\
(g)  & $10$ & $14.175264277554161$ & $0.9560237528884883$ & $0.2$ & $0.52$ &
$0.206278=11.8188^{\degree}$\\
\hline
\end{tabular}
\caption{Properties including frame-dragging (Lense-Thirring precession) of spherical timelike polar Kerr-Newman orbits applying the exact analytic formula (\ref{frameltpreceKNpolas}) and the relations (\ref{energiapolars}), (\ref{consCartpols}), for various values of the black hole's spin and electric charge .}\label{TomiAxonaSymmetrias}
\end{table}

We present examples of spherical polar orbits in Table \ref{TomiAxonaSymmetrias}. The larger the Kerr parameter the larger the frame dragging precession. We note the small contribution of the electric charge on the Lense-Thirring precession for fixed Kerr parameter.

\subsection{Periods}

Squaring the geodesic differential equation for the polar variable (\ref{polareq}) (for $\Lambda=0$), multiplying  by the term $\cos^2\theta\sin^2\theta$, and making the change to the variable $z$, yields  the following differential equation for the proper polar period:
\begin{equation}
{\rm d}\tau_{\theta}=\frac{(r^2+a^2 z){\rm d}z}{2\sqrt{z}\sqrt{a^2(1-E^2)z^2+(-a^2(1-E^2)-L^2-Q)z+Q}}
\label{tautheta}
\end{equation}

Finally, our closed form analytic computation for the proper polar period after integrating (\ref{tautheta}) yields:
\begin{proposition}
\begin{align}
\tau_{\theta}&=4\Biggl[\frac{r^2}{|a|\sqrt{1-E^2}}\frac{\pi}{2}\frac{1}{\sqrt{z_+}}
F\left(\frac{1}{2},\frac{1}{2},1,\frac{z_-}{z_+}\right)\nonumber\\
&-\frac{a^2}{2|a|}\frac{1}{\sqrt{1-E^2}}\sqrt{z_+}\pi F\left(\frac{1}{2},-\frac{1}{2},1,\frac{z_-}{z_+}\right)\nonumber\\
&+\frac{z_+ a^2}{2|a|}\frac{1}{\sqrt{1-E^2}}\frac{\pi}{2}\frac{1}{\sqrt{z_+}}F\left(\frac{1}{2},\frac{1}{2},1,\frac{z_-}{z_+}\right)\Biggr].
\label{polperiodtheta}
\end{align}
\end{proposition}
\begin{proof}
The integral in (\ref{tautheta}) can be split into two parts:
\begin{equation}
\tau_{\theta}=\int\frac{r^2{\rm d}z}{2|a|\sqrt{1-E^2}}\frac{1}{\sqrt{z}\sqrt{(z-z_-)(z-z_+)}}+
\int\frac{a^2 z {\rm d}z}{2\sqrt{z}|a|\sqrt{1-E^2}\sqrt{(z-z_+)(z-z_-)}}.
\end{equation}
We compute first:
\begin{align}
&\int_0^{z_{-}}\frac{r^2{\rm d}z}{2|a|\sqrt{1-E^2}}\frac{1}{\sqrt{z}\sqrt{(z-z_-)(z-z_+)}}\nonumber\\
&=\frac{r^2}{2|a|\sqrt{1-E^2}}F\left(\frac{1}{2},\frac{1}{2},1,
\frac{-z_-}{z_+-z_-}\right)\Gamma\left(\frac{1}{2}\right)\Gamma\left(\frac{1}{2}\right)
\frac{1}{\sqrt{z_+-z_-}}\nonumber\\
&=\frac{r^2}{2|a|\sqrt{1-E^2}}\frac{\pi}{2}\frac{1}{\sqrt{z_+}}
F\left(\frac{1}{2},\frac{1}{2},1,
\frac{z_-}{z_+}\right).
\end{align}
We then write:
\begin{equation}
\int_{z_j}^{z_-}\frac{a^2z{\rm d}z}{2|a|\sqrt{1-E^2}\sqrt{z}\sqrt{(z-z_+)(z-z_-)}}=\int_{z_j}^{z_-}\frac{a^2(z-z_{+}+z_+){\rm d}z}{2|a|\sqrt{1-E^2}\sqrt{z}\sqrt{(z-z_+)(z-z_-)}}
\end{equation}

Applying the change of variables:$z=z_{-}+x(z_j-z_-)$ we compute the term:
\begin{align}
&\int_{z_j}^{z_-}\frac{a^2(z-z_+){\rm d}z}{2|a|\sqrt{1-E^2}\sqrt{z}\sqrt{(z-z_+)(z-z_-)}}
=\frac{a^2}{2|a|\sqrt{1-E^2}}\frac{(z_+-z_-)(z_j-z_-)}{\sqrt{z_-(z_{-}-z_+)(z_j-z_-)}}
\int_0^1\frac{{\rm d}x\left[1-\frac{x(z_{-}-z_j)}{z_{-}-z_+}\right]^{1/2}}{\sqrt{1-\frac{x (z_{-}-z_j)}{z_-}}\sqrt{x}}\nonumber\\
&=\frac{a^2}{2|a|\sqrt{1-E^2}}\frac{(z_+-z_-)(z_j-z_-)}{\sqrt{z_-(z_{-}-z_+)(z_j-z_-)}}
F_1\left(\frac{1}{2},-\frac{1}{2},\frac{1}{2},\frac{3}{2},\frac{z_{-}-z_j}{z_{-}-z_+},
\frac{z_{-}-z_j}{z_-}\right)\frac{\Gamma\left(\frac{1}{2}\right)\Gamma(1)}{\Gamma\left(\frac{3}{2}\right)}
\end{align}
Setting  $z_j=0$ in the expression which involves Appell's hypergeometric function $F_1$ yields:
\begin{align}
&\int_{0}^{z_-}\frac{a^2(z-z_+){\rm d}z}{2|a|\sqrt{1-E^2}\sqrt{z}\sqrt{(z-z_+)(z-z_-)}}\nonumber\\
&=
\frac{-a^2}{2|a|}\frac{1}{\sqrt{1-E^2}}
\frac{(z_+-z_-)}{\sqrt{z_+-z_-}}\frac{\Gamma\left(\frac{1}{2}\right)}{\Gamma\left(\frac{3}{2}\right)}
\frac{\Gamma\left(\frac{3}{2}\right)\Gamma\left(\frac{3}{2}-\frac{1}{2}-\frac{1}{2}\right)}{\Gamma\left(\frac{3}{2}-\frac{1}{2}\right)\Gamma\left(\frac{3}{2}-\frac{1}{2}\right)}
F\left(\frac{1}{2},-\frac{1}{2},1,\frac{-z_-}{z_+-z_-}\right)\nonumber\\
&=\frac{-a^2}{2|a|}\frac{1}{\sqrt{1-E^2}}\sqrt{z_{+}}\Gamma\left(\frac{1}{2}\right)
\Gamma\left(\frac{1}{2}\right)F\left(\frac{1}{2},-\frac{1}{2},1,\frac{z_-}{z_+}\right)
\label{reihekummer}
\end{align}
This yields the second term in Eqn(\ref{polperiodtheta}). In our calculation in (\ref{reihekummer}), we made use of the following transformation property of Gau$\ss$'s hypergeometric function $F$, discovered by Kummer \cite{Kummer}:
\begin{equation}
z^{1-c}F(a-c+1,b-c+1,2-c,z)=z^{1-c}(1-z)^{c-b-1}F(1-a,b-c+1,2-c,\frac{z}{z-1}).
\end{equation}
We also used the special value for Appell's hypergeometric function $F_1$:
\begin{align}
F_1(\alpha,\beta,\beta^{\prime},\gamma,x,1)&=F(a,\beta^{\prime},\gamma,1)F(\alpha,\beta,\gamma-\beta^{\prime},x)\nonumber\\
&=\frac{\Gamma(\gamma)\Gamma(\gamma-\alpha-\beta^{\prime})}{\Gamma(\gamma-\alpha)\Gamma(\gamma-\beta^{\prime})}F(\alpha,\beta,\gamma-\beta^{\prime},x),\;\;\Re(\gamma-\alpha-\beta^{\prime})>0,
\label{summationgauss}
\end{align}
where in the second equality in (\ref{summationgauss}) we applied the Gau$\ss$ summation theorem:
\begin{equation}
F(a,\beta^{\prime},\gamma,1)=\frac{\Gamma(\gamma)\Gamma(\gamma-\alpha-\beta^{\prime})}{\Gamma(\gamma-\alpha)\Gamma(\gamma-\beta^{\prime})},\;\;\Re(\gamma-\alpha-\beta^{\prime})>0.
\end{equation}
\end{proof}

\subsection{Spherical orbits in Kerr-Newman-(anti) de Sitter spacetime}

From (\ref{azimueq}) and (\ref{polareq}) we derive the equation :

\begin{align}
\frac{{\rm d}\phi}{{\rm d}\theta}&=\frac{a\Xi^2}{\Delta_r^{\rm KN}}\frac{[(r^2+a^2)E-aL]}{\sqrt{\Theta^{\prime}}}
-\frac{\Xi^2}{(1+\frac{a^2\Lambda}{3}\cos^2\theta)(\sin^2\theta)}\frac{aE\sin^2\theta-L}{\sqrt{\Theta^{\prime}}}\nonumber\\
&=\frac{a\Xi^2}{\Delta_r^{\rm KN}}\frac{[(r^2+a^2)E-aL]}{\sqrt{\Theta^{\prime}}}
-\frac{\Xi^2}{(1+\frac{a^2\Lambda}{3}z)(1-z)}\frac{aE(1-z)-L}{\sqrt{\Theta^{\prime}}}.
\end{align}

Using the variable $z$ we obtain the following novel exact result in closed analytic form for the amount of frame-dragging that a timelike spherical orbit in Kerr-Newman-(anti)de Sitter spacetime undergoes. As $\theta$ goes through a quarter of a complete oscillation we obtain the change in azimuth $\phi$, $\Delta\phi^{\rm GTR}$ in terms of Lauricella's $F_D$ and Appell's $F_1$  multivariable generalised hypergeometric functions:

\begin{align}
\Delta\phi^{\rm GTR}_{\Lambda}&=\frac{a\Xi^2}{\Delta_r^{\rm KN}}
\frac{[(r^2+a^2)E-aL]}{\sqrt{z_+-z_-}\sqrt{z_{-}-z_{\Lambda}}}
\frac{\Gamma^2\left(\frac{1}{2}\right)}{-2\sqrt{a^4\frac{\Lambda}{4}}}
F_1\left(\frac{1}{2},\frac{1}{2},\frac{1}{2},1,
\frac{z_-}{z_{-}-z_+},\frac{z_-}{z_{-}-z_{\Lambda}}\right)\nonumber\\
&+\frac{H\Xi^2 a E}{2\sqrt{\frac{a^4\Lambda}{3}}}
\Gamma^2 \left(\frac{1}{2}\right)F_D\left(\frac{1}{2},
\frac{1}{2},\frac{1}{2},1,1,\frac{z_-}{z_{-}-z_+},
\frac{z_-}{z_{-}-z_{\Lambda}},-\eta\right)\nonumber \\
&+\frac{-H L\Xi^2}{2\sqrt{\frac{a^4\Lambda}{3}}}\frac{\Gamma^2 \left(\frac{1}{2}\right)}{1-z_-}
F_D\left(\frac{1}{2},
\frac{1}{2},\frac{1}{2},1,1,1,
\frac{z_-}{z_{-}-z_+},\frac{z_-}{z_{-}-z_{\Lambda}},-\eta,\frac{-z_-}{1-z_-}\right),
\label{LTlambdaKNdS}
\end{align}
where we define:
\begin{align}
H&\equiv\frac{1}{\sqrt{z_+-z_-}\sqrt{z_--z_{\Lambda}}}\frac{1}{(
1+\frac{a^2\Lambda z_-}{3})}\nonumber \\
\eta&\equiv \frac{-a^2\Lambda z_-}{1+\frac{a^2\Lambda}{3}z_-}.
\end{align}
The variables $z_+,z_-,z_{\Lambda}$ appearing in the hypergeometric functions in (\ref{LTlambdaKNdS}) are the roots of polynomial equation (\ref{latitudinalroots}).

In our calculations we used the following property for the values of Lauricella's multivariate function $F_D$:
\begin{align}
F_D^{(n)}(\alpha,\beta_1,&\ldots,\beta_n,\gamma,1,x_2,\ldots,x_n)\nonumber\\
&=\frac{\Gamma(\gamma)\Gamma(\gamma-\alpha-\beta_1)}{\Gamma(\gamma-\alpha)
\Gamma(\gamma-\beta_1)}F_D^{(n-1)}(\alpha,\beta_2,\ldots,\beta_n,\gamma-\beta_1,
x_2,\ldots,x_n),\nonumber \\
&\max\{|x_2|,\ldots,|x_n|\}<1,\;\;\Re(\gamma-\alpha-\beta_1)>0.
\label{Fdspecialvalue1}
\end{align}
\subsubsection{Conditions for spherical orbits in Kerr-Newman spacetime with a cosmological constant}

We now proceed to solve the conditions for timelike spherical orbits in the Kerr-Newman spacetime in the presence of the cosmological constant. As before, for a spherical orbit to exist at radius $r$, the conditions $R^{\prime}(r)={\rm d}R^{\prime}/{\rm d}r=0$ must hold at this radius.  The procedure of solving these conditions will lead us to a generalisation of equations (\ref{NESPenergyp}),(\ref{AngulMomenNESPher}) and (\ref{marvelOne}),(\ref{marvelTwo}). We solved these conditions simultaneously and the novel solutions take the following elegant and compact form when parametrised in terms of $r$ and Carter's constant $Q$:

\tikzstyle{mybox} = [draw=blue, fill=white!20, very thick,
    rectangle, rounded corners, inner sep=10pt, inner ysep=20pt]
\tikzstyle{fancytitle} =[fill=red, text=white]
\begin{center}
\begin{tikzpicture}
\node [mybox] (box){%
    \begin{minipage}{1.0\textwidth}
    \begin{theorem}\label{OroTheoLambda}
\begin{align}
E_{\alpha,\beta}(r,Q;\Lambda^{\prime},a,e)&=
\frac{r^2 e^2+r^3(r-2)-r^2(r^2+a^2)\Lambda^{\prime}r^2-a(aQ\mp\sqrt{\Upsilon_{\Lambda}})}
{r^2\sqrt{2e^2 r^2+r^3(r-3)-a^2 r^4\Lambda^{\prime}-2a(aQ\mp\sqrt{\Upsilon_{\Lambda}})}}\label{GENERALESPHE},\\
L_{\alpha,\beta}(r,Q;\Lambda^{\prime},a,e)&=
-\frac{(r^2+a^2)(\mp\sqrt{\Upsilon_{\Lambda}}+aQ)+2 a r^3+ar^4\Lambda^{\prime}(r^2+a^2)-ae^2r^2}{r^2\sqrt{2e^2 r^2+r^3(r-3)-a^2 r^4\Lambda^{\prime}-2a(aQ\mp\sqrt{\Upsilon_{\Lambda}})}}\label{LGENERALSPHERE},
\end{align}
where
\begin{equation}
\Upsilon_{\Lambda}\equiv a^2Q(Q+r^4\Lambda^{\prime})+r^2(-2e^2Q-(e^2+Q)r^2+3Qr+r^3-r^6 \Lambda^{\prime}).
\label{ipsilonL}
\end{equation}
\end{theorem}
    \end{minipage}

};
\end{tikzpicture}%
\end{center}

For vanishing cosmological constant Equations (\ref{GENERALESPHE}),(\ref{LGENERALSPHERE}) reduce to the first integrals Eqns. (\ref{NESPenergyp}) and (\ref{AngulMomenNESPher}) for spherical timelike orbits in Kerr-Newman spacetime.  For vanishing Carter's constant $Q$, Eqns (\ref{GENERALESPHE}),(\ref{LGENERALSPHERE}) reduce to the constants of motion, eqns. (\ref{marvelOne}),(\ref{marvelTwo}), that describe  prograde and retrograde  timelike circular orbits in the equatorial plane of the Kerr-Newman-(anti) de Sitter black hole.

Thus, in Theorem \ref{OroTheoLambda} we have derived the most general solutions for the constants of motion $E,L$ that are associated with the two Killing vectors of the Kerr-Newman-(anti) de Sitter black hole spacetime, which satisfy the conditions for timelike spherical orbits.

To ensure that the smaller square root appearing in the solutions (\ref{GENERALESPHE}),(\ref{LGENERALSPHERE}) is real, we need to impose the condition:
\begin{equation}
\Upsilon_{\Lambda}\geq 0.
\end{equation}

Note from (\ref{ipsilonL}) that $\Upsilon_{\Lambda}$ is quadratic in Carter's constant Q, and its two roots are:
\begin{equation}
Q_{1,2}=\frac{r^2}{2a^2}(-a^2r^2\Lambda^{\prime}+r(r-3)+2e^2\pm \sqrt{\mathfrak{G}_{\Lambda}}),
\end{equation}
where
\begin{equation}
\mathfrak{G}_{\Lambda}\equiv a^4r^4\Lambda^{\prime 2}+2a^2\Lambda^{\prime}r^2(r^2+3r-2e^2)
+r^4+9r^2-6r^3+4e^4-4e^2r(3-r)+4a^2e^2-4a^2r.
\end{equation}
We mention at this point, that $\mathfrak{G}_{\Lambda}$ is a polynomial equation in $r$ that is familiar from the study of circular photon orbits in the equatorial plane around a Kerr-Newman-(anti) de Sitter black hole \cite{ZST}.

Our Theorem \ref{OroTheoLambda} also generalises in a non-trivial way our results for the gravitational frequency shifts in section \ref{grredblue}. Using our solutions for the invariant parameters in (\ref{GENERALESPHE}),(\ref{LGENERALSPHERE}) we can investigate the redshift and blueshift of light emitted by timelike geodesic particles in spherical orbits that are not necessarily confined to the equatorial plane, around a Kerr-Newman-(anti) de Sitter black hole. We can work either with eqn (\ref{gravrbshiftdoppler}) or the corresponding of eqn (\ref{kineticredshift}). Indeed in eqn(\ref{gravrbshiftdoppler}) with $U^r=0$:
\begin{align}
1+z&=\frac{(k^tE-k^{\phi}L+g_{\theta\theta}k^{\theta}U^{\theta})|_e}{(Ek^t-Lk^{\phi}+g_{\theta\theta}k^{\theta}U^{\theta})|_d}\nonumber \\
&=\frac{(E_{\gamma}U^t-L_{\gamma}U^{\phi}+g_{\theta\theta}K^{\theta}U^{\theta})|_e}{(
E_{\gamma}U^t-L_{\gamma}U^{\phi}+g_{\theta\theta}K^{\theta}U^{\theta})|_d}
\label{gravrbshiftdopplerSAORO}
\end{align}
we substitute $E,L$ through relations (\ref{GENERALESPHE}) and (\ref{LGENERALSPHERE}) to obtain:
\begin{equation}
1+z=\frac{(k^tE_{\alpha,\beta}(r,Q;\Lambda^{\prime},a,e)-k^{\phi}L_{\alpha,\beta}(r,Q;\Lambda^{\prime},a,e)+g_{\theta\theta}k^{\theta}U^{\theta})|_e}{(E_{\alpha,\beta}(r,Q;\Lambda^{\prime},a,e)k^t-L_{\alpha,\beta}(r,Q;\Lambda^{\prime},a,e)k^{\phi}+g_{\theta\theta}k^{\theta}U^{\theta})|_d}
\end{equation}
Also the velocity component $U^{\theta}$ of the test massive particle is given by:
\begin{align}
(U^{\theta}(r,\theta,Q;\Lambda^{\prime},a,e)^2&=\Biggl[
(Q+(L_{\alpha,\beta}(r,Q;\Lambda^{\prime},a,e)-aE_{\alpha,\beta}(r,Q;\Lambda^{\prime},a,e))^2-a^2\cos^2\theta)\Delta_{\theta}\nonumber \\
&-\frac{(aE_{\alpha,\beta}(r,Q;\Lambda^{\prime},a,e)\sin^2\theta-L_{\alpha,\beta}(r,Q;\Lambda^{\prime},a,e))^2}{\sin^2\theta}\Biggr]\Big/\rho^4
\end{align}
The impact factor for such more general  orbits (non-equatorial) is computed in section \ref{prosptosiparagon} see equation (\ref{apparentimpPKNadS}).

A thorough investigation will appear in a separate publication.

\subsubsection{Frame-dragging for spherical polar geodesics in Kerr-Newman-de Sitter spacetime}\label{morecosmos}

The generalisation of the effective potential eqn (\ref{effectPotPolarKN}) in the presence of the cosmological constant is:
\begin{equation}
V_{eff.\Lambda}^2=\frac{\Delta_r^{KN}(r^2+K)}{(r^2+a^2)^2}.
\label{effpLambdaPolarSphera}
\end{equation}
From the local extrema of the effective potential (\ref{effpLambdaPolarSphera}) we determine the following first integrals of motion for spherical polar geodesics in Kerr-Newman-de Sitter spacetime:
\begin{align}
&K=\nonumber\\
&\frac{-\Lambda^{\prime}r^3(r^2+a^2)^2-\Lambda^{\prime}r^3(2r^2+a^2)(r^2+a^2)+
2\Lambda^{\prime}r^5(r^2+a^2)+Mr^4+(a^2-e^2)r^3-3Ma^2r^2+ra^2(a^2+e^2)}{Z_{\Lambda}},\\
&E=\frac{r(\Delta_r^{KN})^2}{(r^2+a^2)Z_{\Lambda}},
\end{align}
where
\begin{equation}
Z_{\Lambda}\equiv 2e^2r+r^3-3Mr^2+a^2r+Ma^2+\Lambda^{\prime}r(2r^2+a^2)(r^2+a^2)-2\Lambda^{\prime}r^3(r^2+a^2).
\end{equation}

\tikzstyle{mybox} = [draw=blue, fill=white!20, very thick,
    rectangle, rounded corners, inner sep=10pt, inner ysep=20pt]
\tikzstyle{fancytitle} =[fill=red, text=white]
\begin{center}
\begin{tikzpicture}
\node [mybox] (box){%
    \begin{minipage}{1.0\textwidth}
    \begin{theorem}\label{LambdaUniverseDomin}
 The Lense-Thirring precession $\Delta\phi^{GTR}_{SpherPolar\Lambda}$ per revolution for a polar spherical orbit in Kerr-Newman-de Sitter spacetime is expressed in closed analytic form in terms of Appell's $F_1$ and Lauricella's $F_D$ multivariable hypergeometric functions as follows:
        \begin{align}
&\Delta\phi^{GTR}_{SpherPolar\Lambda}=\frac{4a\Xi^2(r^2+a^2)E}{\Delta_r^{KN}\sqrt{a^4\Lambda/3}}
 \frac{\pi}{2}\frac{1}{\sqrt{z_+}}\frac{1}{\sqrt{-z_{\Lambda}}}
 F_1\left(\frac{1}{2},\frac{1}{2},\frac{1}{2},1,\frac{z_-}{z_+},\frac{z_-}{z_{\Lambda}}\right)\nonumber\\
 &-\frac{4\Xi^2 a E\pi}{2}\frac{1}{1+\frac{a^2\Lambda}{3}z_-}\frac{1}{\sqrt{a^4\Lambda/3}}\frac{1}{\sqrt{z_+}}\frac{1}{\sqrt{-z_{\Lambda}}}
 \frac{1}{\sqrt{1+\eta}}
 F_D\left(\frac{1}{2},\frac{1}{2},\frac{1}{2},1,1,\frac{z_-}{z_+},\frac{z_-}{z_{\Lambda}},
 \frac{-\eta}{-\eta-1}\right).\label{CosmoLenseThirringMeta}
 \end{align}
  \end{theorem}
    \end{minipage}

};
\end{tikzpicture}%
\end{center}

\begin{proof}
The relevant differential equation is:
\begin{equation}
\frac{{\rm d}\phi}{{\rm d}\theta}=\frac{a\Xi^2}{\Delta_r^{KN}}\left[\frac{(r^2+a^2)E}{\sqrt{\Theta^{\prime}}}\right]-
\frac{\Xi^2 a E}{(1+\frac{a^2\Lambda}{3}z)\sqrt{\Theta^{\prime}}}.
\label{spLambdageo}
\end{equation}
The integration of Eqn(\ref{spLambdageo}) splits into two parts. Integrating the first part yields:
\begin{align}
&\frac{a\Xi^2(r^2+a^2)E}{\Delta_r^{KN}}\int\frac{{\rm d}\theta}{\sqrt{\Theta^{\prime}}}\nonumber \\
&=\frac{\textcolor{red}{-}a\Xi^2(r^2+a^2)E}{\Delta_r^{KN}}\frac{1}{\sqrt{a^4\Lambda/3}}\frac{1}{\sqrt{z_+-z_-}}
\frac{1}{\sqrt{z_{-}-z_{\Lambda}}}\frac{\pi}{2}F_1\left(\frac{1}{2},\frac{1}{2},\frac{1}{2},1,
\frac{z_-}{z_{-}-z_+},\frac{z_-}{z_{-}-z_{\Lambda}}\right)\nonumber\\
&=\frac{\textcolor{red}{-}a\Xi^2(r^2+a^2)E}{\Delta_r^{KN}}\frac{1}{\sqrt{a^4\Lambda/3}}\frac{\pi}{2}
\frac{1}{\sqrt{z_+}}\frac{1}{\sqrt{-z_{\Lambda}}}F_1\left(\frac{1}{2},\frac{1}{2},\frac{1}{2},1,
\frac{z_-}{z_+},\frac{z_-}{z_{\Lambda}}\right).
\label{merosproto}
\end{align}

In deriving the last line in Eqn(\ref{merosproto}) we used the following Lemma for Appell's hypergeometric function $F_1$:
\begin{lemma}
\begin{equation}
\fbox{$\displaystyle
F_1(\alpha,\beta,\beta^{\prime},\gamma,x,y)=(1-x)^{-\beta}(1-y)^{-\beta^{\prime}}
F_1(\gamma-\alpha,\beta,\beta^{\prime},\gamma,\frac{x}{x-1},\frac{y}{y-1}).$}
\end{equation}
\end{lemma}

Likewise, using initially the transformation as before: $z=z_-+\xi^2(z_j-z_-)$ and then setting $z_j=0$, integration of the second term gives:
\begin{align}
&\int\frac{-aE\Xi^2}{(1+\frac{a^2\Lambda}{3}z}\frac{{\rm d}\theta}{\sqrt{\Theta^{\prime}}}=\frac{H_1}{\sqrt{a^4\Lambda/3}}\frac{\sqrt{\pi}}{\sqrt{\pi}/2}
F_D\left(\frac{1}{2},\frac{1}{2},\frac{1}{2},\frac{1}{2},1,\frac{3}{2},
\frac{z_{-}-z_j}{z_{-}},\frac{z_{-}-z_j}{z_{-}-z_+},\frac{z_{-}-z_j}{z_{-}-z_{\Lambda}},-\eta\right)\nonumber \\
&\overset{z_j=0}{=}\frac{1}{\sqrt{z_{+}-z_{-}}}\frac{1}{\sqrt{z_{-}-z_{\Lambda}}}
\frac{1}{(1+a^2\frac{\Lambda}{3}z_{-})}\frac{1}{\sqrt{a^4\Lambda/3}}
\frac{\Xi^2 aE}{2}\frac{\sqrt{\pi}}{\sqrt{\pi}/2}
F_D\left(\frac{1}{2},\frac{1}{2},\frac{1}{2},\frac{1}{2},1,\frac{3}{2},1,\frac{z_-}
{z_{-}-z_+},\frac{z_{-}}{z_{-}-z_{\Lambda}},-\eta\right)\nonumber \\
&\overset{(\ref{Fdspecialvalue1})}{=}\frac{1}{\sqrt{z_{+}-z_{-}}}\frac{1}{\sqrt{z_{-}-z_{\Lambda}}}
\frac{1}{(1+a^2\frac{\Lambda}{3}z_{-})}\frac{1}{\sqrt{a^4\Lambda/3}}
\frac{\Xi^2a E\pi}{2}F_D\left(\frac{1}{2},\frac{1}{2},\frac{1}{2},1,1,
\frac{z_-}{z_{-}-z_{+}},\frac{z_{-}}{z_{-}-z_{\Lambda}},-\eta\right)\nonumber\\
&=\frac{1}{(1+a^2\frac{\Lambda}{3}z_{-})}\frac{1}{\sqrt{a^4\Lambda/3}}\frac{\Xi^2aE\pi}{2}
\frac{1}{\sqrt{z_+}}\frac{1}{\sqrt{-z_{\Lambda}}}\frac{1}{\sqrt{1+\eta}} F_D\left(\frac{1}{2},\frac{1}{2},\frac{1}{2},1,1,\frac{z_-}{z_+},
\frac{z_-}{z_{\Lambda}},\frac{-\eta}{-\eta-1}\right),
\label{diomerosfd}
\end{align}
where $H_1=\frac{-(z_j-z_-)(\Xi^2a E/2)}{\sqrt{z_-}\sqrt{z_j-z_-}\sqrt{z_{-}-z_+}\sqrt{z_{-}-z_{\Lambda}}(1+a^2\frac{\Lambda}{3}z_{-})}$. In deriving the last line of (\ref{diomerosfd}) we used the following Lemma for Lauricella's multivariate hypergeometric function $F_D$:
\begin{lemma}
\begin{equation}
\fbox{$\displaystyle
F_D(\alpha,\beta,\beta^{\prime},\beta^{\prime\prime},\gamma,x,y,z)=(1-x)^{-\beta}(1-y)^{-\beta^{\prime}}(1-z)^{-\beta^{\prime\prime}}
F_D(\gamma-\alpha,\beta,\beta^{\prime},\beta^{\prime\prime},\gamma,\frac{x}{x-1},\frac{y}{y-1},\frac{z}{z-1}).$}
\end{equation}
\end{lemma}
\end{proof}

The quantities $z_+,z_-,z_{\Lambda},$ of the hypergeometric functions in eqn.(\ref{CosmoLenseThirringMeta}) are roots of the cubic equation:
\begin{align}
&z^3 (\frac{a^4\Lambda}{3})-z^2(\frac{a^2\Lambda}{3}(Q+a^2E^2\Xi^2)-a^2\Xi(1-E^2\Xi))\nonumber \\
&-z((Q+a^2 E^2 \Xi^2)\Xi+a^2+2aE\Xi^2(-aE))+Q=0.
\end{align}

\subsection{Frame-dragging effect for polar non-spherical bound orbits in Kerr-Newman spacetime}\label{nonspherepLTprece}
In this section we will derive novel closed-form expressions for frame-dragging (Lense-Thirring precession) for polar \textit{non-spherical} timelike geodesics in Kerr-Newman spacetime.
Thus we  assume in this section that $\Lambda=L=0$. The relevant differential equation for the calculation of frame-dragging is:
\begin{equation}
\frac{{\rm d}\phi}{{\rm d}r}=\frac{(2ar-ae^2)E}{\Delta^{KN}\sqrt{R}}.
\end{equation}
The quartic radial polynomial $R$ is obtained from $R^{\prime}$ in (\ref{sextic}) for $\Lambda=L=0$.
Using the partial fractions technique  we integrate from
the periastron
distance $r_{P}$ to the apoastron distance $r_{A}$:

We apply the transformation:
\begin{equation}
z=\frac{1}{\omega}\frac{r-\alpha_{\mu+1}}{r-\alpha_{\mu+2}%
}=\frac{\alpha-\gamma}{\alpha-\beta}\frac{r-\beta}{r-\gamma}%
\label{fundtransform}
\end{equation}
and denote the real roots of the radial polynomial $R$ by $\alpha,\beta,\gamma,\delta$, $\alpha>\beta>\gamma>\delta$.
We organise all the roots  in the ascending order of magnitude:
\begin{equation}
\alpha_{\rho}>\alpha_{\sigma}>\alpha_{\nu}>\alpha_{i},
\end{equation}
with the correspondence $\alpha_{\rho}=\alpha_{\mu}=\alpha,\alpha_{\sigma
}=\alpha_{\mu+1}=\beta,\alpha_{\nu}=\alpha_{\mu+2}=\gamma,\alpha_{i}%
=\alpha_{\mu-i},i=1,2,3,\alpha_{\mu-1}=a_{\mu-2}=r_{\pm},\alpha
_{\mu-3}=\delta$, where $r_+,r_-$ denote the radii of the event horizon
and the inner or Cauchy horizon respectively.

We thus compute the following new  exact analytic result for Lense-Thirring precession that a test particle in a non-spherical polar orbit undergoes, in terms of Appell's
hypergeometric function $F_{1}$:

\begin{align}
\Delta\phi_{tpKN}^{GTR}  & =2
\Biggl[
-\frac{\omega^{3/2}}{H_{+}}A_{tpKN}^{+}F_{1}\left(  \frac{3}{2},1,\frac{1}%
{2},2,\kappa_{+}^{t2},\kappa^{\prime2}\right)  \frac{\pi}{2}\nonumber\\
& +\frac{\sqrt{\omega}}{H_{+}}A_{tpKN}^{+}F_{1}\left(  \frac{1}{2}%
,1,\frac{1}{2},1,\kappa_{+}^{t2},\kappa^{\prime2}\right)  \pi\nonumber\\
& -\frac{\omega^{3/2}}{H_{-}}A_{tpKN}^{-}F_{1}\left(  \frac{3}{2},1,\frac
{1}{2},2,\kappa_{-}^{t2},\kappa^{\prime2}\right)  \frac{\pi}{2}\nonumber\\
& +\frac{\sqrt{\omega}}{H_{-}}A_{tpKN}^{-}F_{1}\left(  \frac{1}{2}%
,1,\frac{1}{2},1,\kappa_{-}^{t2},\kappa^{\prime2}\right) \pi
\Biggr]
 \label{LTprecessionKNpolarNonsphere}
\end{align}
where the partial fraction expansion parameters are given by:
\begin{equation}
A_{tpKN}^+=\frac{-r_+2aE+ae^2E}{r_--r_+},\;\;A_{tpKN}^-=\frac{+r_{-}2aE-ae^2E}{r_--r_+}.
\end{equation}
The variables of the hypergeometric functions are given in terms
of the roots
of the quartic and the radii of the horizons by the expressions:
\begin{equation}
\kappa_{\pm}^{t2}:=\frac{\alpha-\beta}{\alpha-\gamma}\frac{r_{\pm}-\gamma
}{r_{\pm}-\beta},\text{ \ }\kappa^{\prime2}:=\frac{\alpha-\beta}{\alpha
-\gamma}\frac{\delta-\gamma}{\delta-\beta},
\end{equation}
while
\begin{align}
H_{\pm}&\equiv\sqrt{(1-E^2)}(\alpha_{\mu+1}-\alpha_{\mu-1})\sqrt{\alpha_{\mu}-\alpha_{\mu+1}}
\sqrt{\alpha_{\mu+1}-\alpha_{\mu-3}}\nonumber \\
&=\sqrt{(1-E^2)}(\beta-r_{\pm})\sqrt{\alpha-\beta}\sqrt{\beta-\delta}.
\end{align}

\subsubsection{Exact calculation of the orbital period in non-spherical polar Kerr-Newman geodesics}\label{troxiakiPeriodPol}
In this section we will compute a novel exact formula for the orbital period for a test particle in a non-spherical polar Kerr-Newman geodesic.
The relevant differential equation is:
\begin{equation}
\frac{cdt}{dr}=\frac{r^{2}+a^{2}}{\Delta^{KN} \sqrt{R}}E(r^{2}+a^{2})-\frac{%
a^{2}E\sin ^{2}\theta }{\sqrt{R}},
\label{diaforikiperiod}
\end{equation}
and we integrate from periapsis to apoapsis and back to periapsis. Indeed, our analytic computation yields:

\begin{eqnarray}
ct &\equiv &cP_{KN}=\frac{E\beta ^{2}2\frac{GM}{c^{2}}}{\sqrt[2]{1-E^{2}}\sqrt[2]{%
\alpha -\gamma }\sqrt[2]{\beta -\delta }}
\Biggl[
\pi F_{1}\left( \frac{1}{2},2,\frac{1}{2}%
,1,\omega ,\kappa ^{2}\right)  \notag \\
&&-\frac{2\omega \gamma }{\beta }\frac{\pi}{2}
F_{1}\left( \frac{3}{2},2,\frac{1}{2},2,\omega ,\kappa ^{2}\right)  +\frac{\gamma ^{2}\omega ^{2}}{\beta ^{2}}\frac{3\pi}{8}F_{1}\left( \frac{5}{2},2,\frac{1}{2},3,\omega ,\kappa
^{2}\right)
\Biggr]
\notag \\
&&+2E(a^{2}-e^2)\frac{GM}{c^{2}}\sqrt{\frac{\omega }{1-E^{2}}}\frac{1}{\sqrt[2]{%
(\alpha -\beta )(\beta -\delta )}}F(1/2,1/2,1,\kappa ^{2})\pi  \notag \\
&&+\frac{4EGM}{c^{2}}\sqrt{\frac{\omega }{1-E^{2}}}\frac{\beta }{\sqrt[2]{%
(\alpha -\beta )(\beta -\delta )}}%
\Biggl[
\pi F_{1}\left( \frac{1}{2},1,\frac{1}{2}%
,1,\omega ,\kappa ^{2}\right)  \notag \\
&&-\frac{\omega \gamma }{\beta }F_{1}\left( \frac{3}{2},1,\frac{1}{2}%
,2,\omega ,\kappa ^{2}\right) \frac{\pi}{2}%
\Biggr]
\notag \\
&&+\frac{4EG2M}{c^{2}}\sqrt{\frac{\omega }{1-E^{2}}}\frac{1}{\sqrt[2]{%
(\alpha -\beta )(\beta -\delta )}}F(1/2,1/2,1,\kappa ^{2})\pi  \notag \\
&&-\frac{4EG2M}{c^{2}}%
\Biggl[
-\frac{\omega ^{3/2}A_{+}^{KN}}{H_{+}}F_{1}\left( \frac{3}{2},1,\frac{1}{2}%
,2,\kappa _{+}^{2},\mu ^{2}\right) \frac{\pi}{2}  \notag \\
&&+\frac{\omega ^{1/2}A_{+}^{KN}}{H_{+}}F_{1}\left( \frac{1}{2},1,\frac{1}{2}%
,1,\kappa _{+}^{2},\mu ^{2}\right) \pi
-\frac{\omega ^{3/2}A_{-}^{KN}}{H_{-}}F_{1}\left( \frac{3}{2},1,\frac{1}{2}%
,2,\kappa _{-}^{2},\mu ^{2}\right) \frac{\pi}{2%
}  \notag \\
&&+\frac{\omega ^{1/2}A_{-}^{KN}}{H_{-}}F_{1}\left( \frac{1}{2},1,\frac{1}{2}%
,1,\kappa _{-}^{2},\mu ^{2}\right)\pi
\Biggr]
\notag \\
&&+2E\frac{GM}{c^2}\Biggl[
-\frac{\omega ^{3/2}(4e^2r^{\prime}_{+}-e^4)}{(-2\sqrt{1-a^2-e^2})H_{+}}F_{1}\left( \frac{3}{2},1,\frac{1}{2}%
,2,\kappa _{+}^{2},\mu ^{2}\right) \frac{\pi}{2%
}  \notag \\
&&+\frac{\omega ^{1/2}(4e^2r^{\prime}_{+}-e^4)}{(-2\sqrt{1-a^2-e^2})H_{+}}F_{1}\left( \frac{1}{2},1,\frac{1}{2}%
,1,\kappa _{+}^{2},\mu ^{2}\right) \pi
\notag \\
&&-\frac{\omega ^{3/2}(e^4-4e^2r^{\prime}_-)}{(-2\sqrt{1-a^2-e^2})H_{-}}F_{1}\left( \frac{3}{2},1,\frac{1}{2}%
,2,\kappa _{-}^{2},\mu ^{2}\right) \frac{\pi}{2%
}  \notag \\
&&+\frac{\omega ^{1/2}(e^4-4e^2r^{\prime}_-)}{(-2\sqrt{1-a^2-e^2})H_{-}}F_{1}\left( \frac{1}{2},1,\frac{1}{2}%
,1,\kappa _{-}^{2},\mu ^{2}\right) \pi
\Biggr]\notag \\
&&+\frac{-a^{2}E\frac{GM}{c^{2}}}{\sqrt[2]{Q}}
\Biggl[%
\sin (\varphi )F_{1}\left( \frac{1}{2},\frac{1}{2},\frac{1}{2},\frac{3}{2}%
,\sin ^{2}\varphi ,\kappa ^{2\prime }\sin ^{2}\varphi \right) +\Biggl\{
-\sin (\varphi )F_{1}\left( \frac{1}{2},\frac{1}{2},\frac{1}{2},\frac{3}{2}%
,\sin ^{2}\varphi ,\kappa ^{2\prime }\sin ^{2}\varphi \right) \notag \\
&&+\sin (\varphi )F_{1}\left( \frac{1}{2},\frac{1}{2},-\frac{1}{2},\frac{3}{2%
},\sin ^{2}\varphi ,\kappa ^{2\prime }\sin ^{2}\varphi \right)
\Biggr\}
\times \frac{1}{\kappa ^{2\prime }}
\Biggr]
\label{PolarPeriodOrbitKN}
\end{eqnarray}%
where
\begin{equation}
\varphi =am\left( \sqrt[2]{Q}\frac{4}{\sqrt[2]{1-E^{2}}}\frac{1}{\sqrt[2]{%
\alpha -\gamma }}\frac{1}{\sqrt[2]{\beta -\delta }}\frac{\pi }{2}F\left(
\frac{1}{2},\frac{1}{2},1,\kappa ^{2}\right) ,\frac{a^{2}(1-E^{2})}{Q}\right)
\label{gustavCJACOBI}
\end{equation}

\begin{equation}
A_{+}^{KN}:=-\frac{a^{2}+e^2-2r_{+}^{\prime }}{r_{-}^{\prime }-r_{+}^{\prime }}%
,\qquad A_{-}^{KN}:=-\frac{-a^{2}-e^2+2r_{-}^{\prime }}{r_{-}^{\prime }-r_{+}^{\prime
}},
\end{equation}
and the moduli (variables) of the hypergeometric function of Appell are
given by:
\begin{eqnarray}
\mu ^{2} &=&\kappa ^{2}=\frac{\alpha -\beta }{\alpha -\gamma }\frac{\delta
-\gamma }{\delta -\beta }  \notag \\
\kappa _{\pm }^{2} &=&\frac{\alpha -\beta }{\alpha -\gamma }\frac{r_{\pm
}^{\prime }-\gamma }{r_{\pm }^{\prime }-\beta },\;\kappa^{2\prime}=\frac{a^2(1-E^2)}{Q}.
\end{eqnarray}
Also $\omega=\frac{\alpha-\beta}{\alpha-\gamma},\;$ $H_{\pm}\equiv\sqrt{(1-E^2)}(\beta-r^{\prime}_{\pm})\sqrt{\alpha-\beta}\sqrt{\beta-\delta}$ and $r^{\prime}_{\pm}=\frac{r_{\pm}}{\frac{GM}{c^2}}$ are the dimensionless horizon radii.
In Eqn.(\ref{gustavCJACOBI}) $am(u,k^{2\prime})$ denotes the Jacobi amplitude function \cite{CGJampli}.

The last three terms in (\ref{PolarPeriodOrbitKN}) arise after integrating the second term on the right hand side of (\ref{diaforikiperiod}). Details of this particular angular  integration are given in appendix A.1, pages 1801-1802 in \cite{GVKperiapsisadvLTprecession}.

For zero electric charge, $e=0$, Eqn.(\ref{PolarPeriodOrbitKN}) reduces correctly to Eqn.(33) in \cite{GVKperiapsisadvLTprecession} for the case of a Kerr black hole.
The Lense-Thirring period for a non-spherical polar timelike geodesic in Kerr-Newman BH geometry, is defined in terms of the Lense-Thirring precession Eq. (\ref{LTprecessionKNpolarNonsphere}) and its orbital period  Eq. (\ref{PolarPeriodOrbitKN}) as follows:
\begin{equation}
{\rm LTP}:=\frac{2\pi P_{KN}}{\Delta \phi ^{GTR}_{tpKN}}.
\end{equation}

\begin{table}[tbp] \centering
\begin{tabular}
[c]{cccccccc}\hline
Star &   $Q/M^2$ & $E$  & $e/M$  & $a/M$ & $\Delta\phi_{tpKN}^{GTR}$ & $P_{KN}({\rm yr})$ & $LTP({\rm yr})$\\ \hline
S2   & $5693.30424$ & $0.999979485$ & $0.33$ &$0.52$ & $3.14284\frac
{\operatorname{arcsec}}{revol.}$ &$15.15$ & $6.25\times 10^6$\\
S2  & $5693.30424$ & $0.999979485$ & $0.11$ & $0.52$ & $3.14295\frac
{\operatorname{arcsec}}{revol.}$ & $15.15$ &$6.25\times 10^6$\\
S2  & $5693.30424$ & $0.999979485$ & $0$ & $0.52$ & $3.14297\frac
{\operatorname{arcsec}}{revol.}$ & $15.15$ &$6.25\times 10^6$\\
S2 & $5273.53220$ & $0.999979145$ & $0.11$ & $0.52$& $3.5261\frac
{\operatorname{arcsec}}{revol.}$ & $14.78$ & $5.43\times 10^6$\\
S2 & $5273.53220$ & $0.999979145$ & $0.33$ & $0.52$& $3.52596\frac
{\operatorname{arcsec}}{revol.}$ & $14.78$ & $5.43\times 10^6$\\
\end{tabular}
\caption{Lense-Thirring precession for the star $S2$ in the central arcsecond
of the galactic centre, using the exact formula
$(\ref{LTprecessionKNpolarNonsphere})$
. We assume a central galactic Kerr-Newman black hole with mass
$M_{\rm BH}=4.06\times 10^6M_{\odot}$ and that the orbit of  $S2$  star is a timelike non-spherical polar Kerr-Newman geodesic. The computation of the orbital period of the star S2 was performed using the exact result in eqn.(\ref{PolarPeriodOrbitKN}). The periapsis and apoapsis distances for the star $S2$ in the first three lines are respectively $1.82\times 10^{15}{\rm cm}$  and $2.74\times 10^{16}{\rm cm}$. In the last two lines these distances are: $r_P=1.68\times 10^{15}{\rm cm}$, $r_A=2.71\times 10^{16}{\rm cm}$. }\label{AstroS2FDeapolar1SgraA*}
\end{table}

\begin{table}[tbp] \centering
\begin{tabular}
[c]{cccccccc}\hline
Star &   $Q/M^2$ & $E$  & $e/M$  & $a/M$ & $\Delta\phi_{tpKN}^{GTR}$ & $P_{KN}({\rm yr})$ & $LTP({\rm yr})$\\ \hline
S14   & $5321.06355$ & $0.999988863$ & $0.11$ &$0.9939$ & $6.64977\frac
{\operatorname{arcsec}}{revol.}$ &$37.88$ & $7.38\times 10^6$\\
S14  & $5321.06355$ & $0.999988863$ & $0$ & $0.9939$ & $6.64981\frac
{\operatorname{arcsec}}{revol.}$ & $37.88$ &$7.38\times 10^6$\\
S14  & $4204.76359$ & $0.999987653$ & $0.11$ & $0.9939$ & $9.47145\frac
{\operatorname{arcsec}}{revol.}$ & $32.45$ & $4.44\times 10^6$\\
S14  & $4204.76359$ & $0.999987653$ & $0$ & $0.9939$ & $9.47151\frac
{\operatorname{arcsec}}{revol.}$ & $32.45$ & $4.44\times 10^6$\
\end{tabular}
\caption{Lense-Thirring precession for the star $S14$ in the central arcsecond
of the galactic centre, using the exact formula
$(\ref{LTprecessionKNpolarNonsphere})$
. We assume a central galactic Kerr-Newman black hole with mass
$M_{\rm BH}=4.06\times 10^6M_{\odot}$ and that the orbit of  $S14$  star is a timelike non-spherical polar Kerr-Newman geodesic. The computation of the orbital period of the star S14 was performed using the exact result in eqn.(\ref{PolarPeriodOrbitKN}). The periapsis and apoapsis distances for the star $S14$ in the first two lines are respectively $1.64\times 10^{15}{\rm cm}$  and $5.22\times 10^{16}{\rm cm}$, whereas in the last two lines are predicted to be $1.29\times 10^{15}{\rm cm}$ and $4.73\times 10^{16}{\rm cm}$ respectively.}\label{AstroS14FDLTeapolar1SgraA*}
\end{table}

We now proceed to calculate using our exact analytic solutions and assuming a central galactic Kerr-Newman black hole, the Lense-Thirring effect and the corresponding Lense-Thirring period for the observed stars S2,S14 for various values of the Kerr parameter and the electric charge of the central black hole-see Tables \ref{AstroS2FDeapolar1SgraA*}-\ref{AstroS14FDLTeapolar1SgraA*}. The choice of values for the invariant parameters $Q,E$ is restricted by requiring that the predictions of the theory for the periapsis, apoapsis distances and the period $P_{KN}$ are in agreement with the orbital data from observations in \cite{Eisenhauer}. We note here the orbital data from observations that we use to constrain our theory. For S2 the eccentricity measured is  $\textbf{e}_{S2}=0.8760\pm 0.0072$, its orbital period $P_{S2}(yr)=15.24\pm 0.36$, and the semimajor axis $\textbf{a}_{S2}(\rm arcsec)=0.1226\pm 0.0025$ \cite{Eisenhauer}. The corresponding orbital measurements for S14,  reported in  \cite{Eisenhauer} $\textbf{e}_{S14}=0.9389\pm 0.0078$, $P_{S14}(yr)=38.0\pm 5.7$, $\textbf{a}_{S14}(\rm arcsec)=0.225\pm 0.022$.
The theoretical prediction for the eccentricity of the S2 orbit for the choice  of values for the invariant parameters $Q,E$ in the last two lines in Table \ref{AstroS2FDeapolar1SgraA*} is $\textbf{e}_{S2}^{theory}=0.883$ . It is consistent with the upper allowed value by experiment. On the other hand the prediction of the exact theory, for the choice of values for the constants of motion  $Q,E$ in the first three lines in Table \ref{AstroS2FDeapolar1SgraA*}, for the orbital eccentricity is $\textbf{e}_{S2}^{theory}=0.8755$ in precise agreement with experiment.
We observe that the contribution of the electric charge on the frame-dragging precession is small.

\subsection{Periapsis advance for non-spherical polar timelike Kerr-Newman orbits}\label{periapsisKN}

In this section we shall investigate the pericentre advance for a non-spherical bound Kerr-Newman polar orbit, assuming a vanishing cosmological constant. We shall first obtain the exact solution for the orbit and then derive an exact closed form formula for the periastron precession. The relevant differential equation is:
\begin{equation}
\int^r\frac{{\rm d}r}{\sqrt{R}}=\pm\int^{\theta}\frac{{\rm d}\theta}{\sqrt{\Theta}}
\label{abelastronomy}
\end{equation}
Now applying the transformation (\ref{fundtransform}) on the left hand side of (\ref{abelastronomy}) yields:
\begin{eqnarray}
\int\frac{{\rm d}r}{\sqrt{R}}&=\frac{1}{\frac{GM}{c^2}}\int\frac{{\rm d}r}
{\sqrt{(1-E^2)(-)(r-\alpha)(r-\beta)(r-\gamma)(r-\delta)}}\nonumber\\
&=\frac{1}{\frac{GM}{c^2}}\int\frac{{\rm d}z\sqrt{\omega}}{\sqrt{(1-E^2)}\sqrt{(\alpha-\beta)(\beta-\delta)}\sqrt{z(1-z)(1-k^2z)}},
\end{eqnarray}
where
\begin{equation}
k^2=\frac{\alpha-\beta}{\alpha-\gamma}\frac{\delta-\gamma}{\delta-\beta}\equiv\omega\frac{\delta-\gamma}{\delta-\beta}.
\end{equation}
The roots $\alpha,\beta,\gamma,\delta$ of the quartic polynomial equation:
\begin{equation}
R=((r^2+a^2)E)^2-(r^2+a^2+e^2-2r)(r^2+Q+a^2 E^2)=0,
\end{equation}
are organised as $\alpha_{\mu}>\alpha_{\mu+1}>\alpha_{\mu+2}>\alpha_{\mu-3}$, and we have the correspondence $\alpha=\alpha_{\mu},\beta=\alpha_{\mu+1},\gamma=\alpha_{\mu+2},\delta=\alpha_{\mu-3}$.

By setting, $z=x^2$, we obtain the equation
\begin{equation}
\int\frac{{\rm d}x}{\sqrt{(1-x^2)(1-k^2x^2)}}=\frac{\sqrt{1-E^2}\sqrt{\alpha-\beta}\sqrt{\beta-\delta}}{2\sqrt{\omega}}\int\frac{{\rm d}\theta}{\sqrt{\Theta}}.
\end{equation}
Using the idea of \textit{inversion} for the orbital elliptic integral on the left-hand side we obtain
\begin{equation}
x={\rm sn}\left(\frac{\sqrt{1-E^2}\sqrt{\alpha-\gamma}\sqrt{\beta-\delta}}{2}\int\frac{{\rm d}\theta}{\sqrt{\Theta}},k^2\right).
\end{equation}
In terms of the original variables we derive the equation
\begin{equation}
r=\frac{\beta-\gamma\frac{\alpha-\beta}{\alpha-\gamma}{\rm sn}^2\left(\frac{\sqrt{1-E^2}\sqrt{(\alpha-\gamma)(\beta-\delta)}}{2}\int\frac{{\rm d}\theta}{\sqrt{\Theta}},k^2\right)}{1-\frac{\alpha-\beta}{\alpha-\gamma}{\rm sn}^2\left(\frac{\sqrt{1-E^2}\sqrt{(\alpha-\gamma)(\beta-\delta)}}{2}\int\frac{{\rm d}\theta}{\sqrt{\Theta}},k^2\right)}\frac{GM}{c^2}.
\label{orbitalsolKN}
\end{equation}
Equation (\ref{orbitalsolKN}) represents the first exact solution that describes the motion of a test particle in a polar non-spherical bound orbit in the Kerr-Newman field in terms of Jacobi's sinus amplitudinus elliptic function ${\rm sn}$. The function ${\rm sn}^2(y,k^2)$ has period $2K(k^2)=\pi F\left(\frac{1}{2},\frac{1}{2},1,k^2\right)$ which is also a period of $r$. This means that after one complete revolution the angular integration has to satisfy the equation
\begin{equation}
-\int\frac{{\rm d}\theta}{\sqrt{\Theta}}=\frac{1}{\sqrt{Q}}\int\frac{\partial \Psi}{\sqrt{1-k^{\prime 2}\sin^2\Psi}}=\frac{4}{\sqrt{1-E^2}}\frac{1}{\sqrt{\alpha-\gamma}}\frac{1}{\sqrt{\beta-\delta}}
\frac{\pi}{2}F\left(\frac{1}{2},\frac{1}{2},1,k^2\right),
\label{periaadvapolos}
\end{equation}
where the latitude variable $\Psi:=\pi/2-\theta$ has been introduced. Equation (\ref{periaadvapolos}) can be rewritten as
\begin{equation}
\int\frac{{\rm d}x^{\prime}}{\sqrt{(1-x^{\prime 2})(1-k^{\prime 2}x^{\prime 2})}}=\sqrt{Q}\frac{4}{\sqrt{1-E^2}}\frac{1}{\sqrt{\alpha-\gamma}}\frac{1}{\sqrt{\beta-\delta}}
\frac{\pi}{2}F\left(\frac{1}{2},\frac{1}{2},1,k^2\right)
\label{exactanguastronomy}
\end{equation}
and
\begin{equation}
k^{\prime 2}:=\frac{a^2(1-E^2)}{Q}.
\end{equation}
Also $x^{\prime2}=\sin^2\Psi=\cos^2\theta.$ Equation (\ref{exactanguastronomy}) determines the exact amount that the angular integration satisfies after a complete radial oscillation for a bound non-spherical polar orbit in KN spacetime.
Now using the latitude variable $\Psi$ the change in latitude after a complete radial oscillation leads to the following exact expression for the periastron advance for a test particle in a non-spherical polar Kerr-Newman orbit, assuming a vanishing cosmological constant, in terms of Jacobi's amplitude function and Gau\ss's hypergeometric function

\begin{align}
\Delta\Psi^{\rm GTR}_{{\rm KN}}&=\Delta\Psi-2\pi\nonumber\\
&={\rm am}\left(\sqrt{Q}\frac{4}{\sqrt{1-E^2}}\frac{1}{\sqrt{\alpha-\gamma}}
\frac{1}{\sqrt{\beta-\delta}}\frac{\pi}{2}F\left(\frac{1}{2},\frac{1}{2},1,
\kappa^2\right),\frac{a^2(1-E^2)}{Q}\right)-2\pi.
\label{PolarKNperiapsis}
\end{align}
The Abel-Jacobi's amplitude ${\rm am}(m,k^{\prime 2})$ is the function that inverts the elliptic integral
\begin{equation}
\int_0^{\Psi}\frac{\partial \Psi}{\sqrt{1-k^{\prime 2}\sin^2\Psi}}=m
\end{equation}

We compute with the aid of (\ref{PolarKNperiapsis}) the periapsis advance for the stars S2 and S14 assuming that they orbit in a timelike non-spherical polar Kerr-Newman geodesic. Our results are displayed in Tables \ref{AstroS2eapolar1SgraA*} and \ref{AstroS14eapolar2SgraA*}. We observe that the effect of $e$ is small.

We also compute with the aid of the following exact formula for the periapsis advance for an equatorial non-circular timelike geodesic in the Kerr-Newman spacetime, first derived in \cite{GRGKRANIOTIS}:
\begin{equation}
\delta_P^{teKN}:=\Delta\phi_{teKN}^{GTR}-2\pi,
\label{deltaequatopericentre}
\end{equation}
where
\begin{align}
\Delta\phi_{teKN}^{GTR}  & =2
\Biggl[
-\frac{\omega^{3/2}}{H_{+}}A_{teKN}^{+}F_{1}\left(  \frac{3}{2},1,\frac{1}%
{2},2,\kappa_{+}^{t2},\kappa^{\prime2}\right)  \frac{\pi}{2}\nonumber\\
& +\frac{\sqrt{\omega}}{H_{+}}A_{teKN}^{+}F_{1}\left(  \frac{1}{2}%
,1,\frac{1}{2},1,\kappa_{+}^{t2},\kappa^{\prime2}\right)  \pi\nonumber\\
& -\frac{\omega^{3/2}}{H_{-}}A_{teKN}^{-}F_{1}\left(  \frac{3}{2},1,\frac
{1}{2},2,\kappa_{-}^{t2},\kappa^{\prime2}\right)  \frac{\pi}{2}\nonumber\\
& +\frac{\sqrt{\omega}}{H_{-}}A_{teKN}^{-}F_{1}\left(  \frac{1}{2}%
,1,\frac{1}{2},1,\kappa_{-}^{t2},\kappa^{\prime2}\right)  \pi
\Biggr]
\nonumber\\
& +\frac{2\sqrt{\omega}L}{\sqrt{(1-E^{2})(\alpha-\beta)(\beta-\delta)}%
}F\left(  \frac{1}{2},\frac{1}{2},1,\kappa^{\prime2}\right)  \pi\nonumber\\
& \label{RelativiPeriastronPrec}%
\end{align}
the pericentre-shift for the stars S2 and S14  for various values for the spin and charge of the central black hole. The derivation of the formula (\ref{RelativiPeriastronPrec}) as well as the definition of the coefficients $A_{teKN}^{\pm},H_{\pm}$ can be found in \cite{GRGKRANIOTIS},pp 24-26.

By performing this calculation, we gain a more complete appreciation of the effect of the electric charge of the rotating
galactic black hole (we assume that the KN solution describes the
curved spacetime geometry around SgrA*) on this observable. We
also assume that the angular momentum axis of the orbit is
co-aligned with the spin axis of the black hole and that the
$S-$stars can be treated as neutral test particles i.e. their orbits are  timelike non-circular equatorial Kerr-Newman geodesics. Our results are displayed in Tables \ref{AstroS2equato1aSgraA*}-\ref{AstroS14equa4SgraA*}.
The choice of values for the constants of motion $L,E$ is restricted by requiring that the predictions of the theory for the periapsis apoapsis distances and the orbital period are in agreement with published orbital data from observations in \cite{Eisenhauer}.
It is evident in this case that the value of electric charge plays a significant role in the value of the pericentre-shift as opposed to the results in Tables  \ref{AstroS2eapolar1SgraA*} and \ref{AstroS14eapolar2SgraA*} especially for moderate values of the spin of the black hole.

We note at this point that a more precise analysis would involve the calculation of relativistic periapsis advance for more general timelike non-spherical orbits, inclined non-equatorial and non-polar in the KN(a)dS spacetime, which is beyond the scope of the current publication. Such a generalisation is important because the observed high eccentricity orbits of S-stars are such that:  neither of the S stars orbits is  equatorial nor polar. Such a general analysis will be a subject of a future publication \footnote{For non-spherical non-polar orbits with orbital inclination $0^{\degree}<i<90^{\degree}$ one expects that the resulting periapsis advance has a value between those calculated using formula (\ref{PolarKNperiapsis}) and (\ref{RelativiPeriastronPrec})(assuming the three orbital configurations, polar,non-polar and equatorial have the same eccentricity and semimajor axis and fixed values of the black hole parameters)}.

\begin{table}[tbp] \centering
\begin{tabular}
[c]{cccccccc}\hline
Star &   $Q/M^2$ & $E$  & $e/M$  & $a/M$ & $\Delta\Psi^{{\rm GTR}}_{{\rm KN}}$ & Periapsis $r_P$ & Apoapsis $r_A$\\ \hline
S2   & $5693.30424$ & $0.999979485$ & $0.11$ &$0.9939$ & $682.512\frac
{\operatorname{arcsec}}{revol.}$ & $1.82\times 10^{15}{\rm cm}$ & $2.74\times 10^{16}{\rm cm}$\\
S2  & $5693.30424$ & $0.999979485$ & $0.11$ & $0.52$ & $682.533\frac
{\operatorname{arcsec}}{revol.}$ & $1.82\times 10^{15}{\rm cm}$ & $2.74\times 10^{16}{\rm cm}$\\
\end{tabular}
\caption{Periastron precession for the star $S2$ in the central arcsecond
of the galactic centre, using the exact formula
$(\ref{PolarKNperiapsis})$
. We assume a central galactic Kerr-Newman black hole with mass
$M_{\rm BH}=4.06\times 10^6M_{\odot}$ and that the orbit of  $S2$  star is a timelike non-spherical polar Kerr-Newman geodesic.}\label{AstroS2eapolar1SgraA*}
\end{table}

\begin{table}[tbp] \centering
\begin{tabular}
[c]{cccccccc}\hline
Star &   $Q/M^2$ & $E$  & $e/M$  & $a/M$ & $\Delta\Psi^{{\rm GTR}}_{{\rm KN}}$ & Periapsis $r_P$ & Apoapsis $r_A$\\ \hline
S14   & $5321.06355$ & $0.999988863$ & $0.11$ &$0.9939$ & $730.351\frac
{\operatorname{arcsec}}{revol.}$ & $1.64\times 10^{15}{\rm cm}$ & $5.22\times 10^{16}{\rm cm}$\\
S14  & $5321.06355$ & $0.999988863$ & $0.11$ & $0.52$ & $730.376\frac
{\operatorname{arcsec}}{revol.}$ & $1.64\times 10^{15}{\rm cm}$ & $5.22\times 10^{16}{\rm cm}$\\
\end{tabular}
\caption{Periastron precession for the star $S14$ in the central arcsecond
of the galactic centre, using the exact formula
$(\ref{PolarKNperiapsis})$
. We assume a central galactic Kerr-Newman black hole with mass
$M_{\rm BH}=4.06\times 10^6M_{\odot}$ and that the orbit of  $S14$  star is a timelike non-spherical polar Kerr-Newman geodesic.}\label{AstroS14eapolar2SgraA*}
\end{table}

A few further comments are in order. The values of the hypothetical electric charge of the central Kerr-Newman black hole have been chosen so that the surrounding spacetime represents a black hole, i.e.
the singularity surrounded by the horizon, the electric charge and
angular
momentum $J$ must be restricted by the relation:
\begin{equation}
\fbox{$\dfrac{GM}{c^{2}}\geq\left[  \left(  \dfrac{J}{Mc}\right)  ^{2}%
+\dfrac{Ge^{2}}{c^{4}}\right]  ^{1/2}$}\Leftrightarrow
\end{equation}%
\begin{align}
\dfrac{GM}{c^{2}}  & \geq\left[  a^{2}+\dfrac{Ge^{2}}{c^{4}}\right]
^{1/2}\Rightarrow\\
e^{2}  & \leq GM^{2}(1-a^{\prime2})\label{FundEcGMa}%
\end{align}
where in the last inequality $a^{\prime}=\frac{a}{GM/c^2}$ denotes a dimensionless Kerr
parameter.

The values of the electric charge used for instance in Tables \ref{AstroS14eapolar2SgraA*},\ref{AstroS2equato1aSgraA*} , for the SgrA*
galactic black hole correspond to magnitudes:
\begin{align}
e  &
=0.85\sqrt{6.6743\times10^{-8}}4.06\times10^{6}\times1.9884\times
10^{33}\mathrm{esu=}1.77\times10^{36}\mathrm{esu}\Leftrightarrow
5.94\times10^{26}%
\operatorname{C}%
,\nonumber\\
e  & =0.11\sqrt{6.6743\times10^{-8}}4.06\times10^{6}\times1.9884\times
10^{33}\mathrm{esu=}2.29\times10^{35}\mathrm{esu}\Leftrightarrow
7.65\times10^{25}%
\operatorname{C}%
.\nonumber\\
&
\end{align}
Concerning these tentative values for the electric charge $e$ we
used in applying our exact solutions for the case of SgrA* black
hole we note that their likelihood is debatable: There is an expectation that the
electric charge trapped in the galactic nucleous will not likely
reach so high values as the ones close to the extremal values
predicted in (\ref{FundEcGMa}) that allow the avoidance of a naked
singularity. However, more precise statements on the electric
charge's magnitude of the galactic black hole or its upper bound
will only be reached once the relativistic effects predicted in
this work are measured and a comparison of the theory we developed
with experimental data will take place \footnote{In this regard,
we also mention that the author in \cite{LIORIO}, under the
assumption that the curved geometry surrounding the massive object
in the Galactic Centre is a Reissner-Nordstr\"{o}m (RN) spacetime,
obtained an upper bound of $e\lesssim 3.6\times10^{27}$C. This
upper bound does not distinguish yet between a RN black hole
scenario and a RN naked singularity scenario.}.
\begin{table}[tbp] \centering
\begin{tabular}
[c]{cccccccc}\hline
Star &   $L/M$ & $E$  & $e/M$  & $a/M$ & $\delta^{teKN}_{\rm p}$ & Periapsis $r_P$ & Apoapsis $r_A$\\ \hline
S2   & $75.4539876$ & $0.999979485$ & $0.11$ &$0.9939$ & $670.565\frac
{\operatorname{arcsec}}{revol.}$ & $1.82\times 10^{15}{\rm cm}$ & $2.74\times 10^{16}{\rm cm}$\\
S2  & $75.4539876$ & $0.999979485$ & $0.025$ & $0.9939$ & $671.876\frac
{\operatorname{arcsec}}{revol.}$ & $1.82\times 10^{15}{\rm cm}$ & $2.74\times 10^{16}{\rm cm}$\\
S2  & $75.4539876$ & $0.999979485$ & $0.025$ & $0.52$ & $677.571\frac
{\operatorname{arcsec}}{revol.}$ & $1.82\times 10^{15}{\rm cm}$ & $2.74\times 10^{16}{\rm cm}$\\
S2  & $75.4539876$ & $0.999979485$ & $0.1$ & $0.52$ & $676.5\frac
{\operatorname{arcsec}}{revol.}$  & $1.82\times 10^{15}{\rm cm}$ & $2.74\times 10^{16}{\rm cm}$\\
S2  & $75.4539876$ & $0.999979485$ & $0.85$ & $0.52$ & $595.1\frac
{\operatorname{arcsec}}{revol.}$& $1.82\times 10^{15}{\rm cm}$ & $2.74\times 10^{16}{\rm cm}$\\
S2  & $72.6190898$ & $0.999979145$ & $0$ & $0.9939$ & $725.018\frac
{\operatorname{arcsec}}{revol.}$& $1.68\times 10^{15}{\rm cm}$ & $2.71\times 10^{16}{\rm cm}$\\
S2  & $72.6190898$ & $0.999979145$ & $0.11$ & $0.9939$ & $723.525\frac
{\operatorname{arcsec}}{revol.}$& $1.68\times 10^{15}{\rm cm}$ & $2.71\times 10^{16}{\rm cm}$\\
S2  & $72.6190898$ & $0.999979145$ & $0$ & $0.52$ & $731.407\frac
{\operatorname{arcsec}}{revol.}$& $1.68\times 10^{15}{\rm cm}$ & $2.71\times 10^{16}{\rm cm}$\\
S2  & $72.6190898$ & $0.999979145$ & $0.11$ & $0.52$ & $728.176\frac
{\operatorname{arcsec}}{revol.}$& $1.68\times 10^{15}{\rm cm}$ & $2.71\times 10^{16}{\rm cm}$\\
\end{tabular}
\caption{Periastron precession for the star $S2$ in the central arcsecond
of the galactic centre, using the exact formula
Eqn. (\ref{RelativiPeriastronPrec}). We assume a central galactic Kerr-Newman black hole with mass
$M_{\rm BH}=4.06\times 10^6M_{\odot}$ and that the orbit of the  star $S2$ is a timelike non-circular equatorial Kerr-Newman geodesic. }\label{AstroS2equato1aSgraA*}
\end{table}

\begin{table}[tbp] \centering
\begin{tabular}
[c]{cccccccc}\hline
Star &   $L/M$ & $E$  & $e/M$  & $a/M$ & $\delta^{teKN}_{\rm p}$ & Periapsis $r_P$ & Apoapsis $r_A$\\ \hline
S14   & $72.9456205$ & $0.999988863$ & $0.11$ &$0.9939$ & $717.128\frac
{\operatorname{arcsec}}{revol.}$ & $1.64\times 10^{15}{\rm cm}$ & $5.22\times 10^{16}{\rm cm}$\\
S14  & $72.9456205$ & $0.999988863$ & $0.025$ & $0.9939$ & $718.531\frac
{\operatorname{arcsec}}{revol.}$ & $1.64\times 10^{15}{\rm cm}$ & $5.22\times 10^{16}{\rm cm}$\\
S14  & $72.9456205$ & $0.999988863$ & $0.11$ & $0.52$ & $723.432\frac
{\operatorname{arcsec}}{revol.}$ & $1.64\times 10^{15}{\rm cm}$ & $5.22\times 10^{16}{\rm cm}$\\
S14  & $72.9456205$ & $0.999988863$ & $0.33$ & $0.52$ & $711.595\frac
{\operatorname{arcsec}}{revol.}$  & $1.64\times 10^{15}{\rm cm}$ & $5.22\times 10^{16}{\rm cm}$\\
S14  & $72.9456205$ & $0.999988863$ & $0.85$ & $0.52$ & $636.568\frac
{\operatorname{arcsec}}{revol.}$& $1.64\times 10^{15}{\rm cm}$ & $5.22\times 10^{16}{\rm cm}$\\
S14 & $64.8441485$ & $0.999987653$ & $0$ & $0.9939$ & $907.686\frac
{\operatorname{arcsec}}{revol.}$& $1.29\times 10^{15}{\rm cm}$ & $4.73\times 10^{16}{\rm cm}$\\
S14 & $64.8441485$ & $0.999987653$ & $0.11$ & $0.9939$ & $905.812\frac
{\operatorname{arcsec}}{revol.}$& $1.29\times 10^{15}{\rm cm}$ & $4.73\times 10^{16}{\rm cm}$\\
S14 & $64.8441485$ & $0.999987653$ & $0$ & $0.52$ & $916.661\frac
{\operatorname{arcsec}}{revol.}$& $1.29\times 10^{15}{\rm cm}$ & $4.73\times 10^{16}{\rm cm}$\\
S14 & $64.8441485$ & $0.999987653$ & $0.11$ & $0.52$ & $914.786\frac
{\operatorname{arcsec}}{revol.}$& $1.29\times 10^{15}{\rm cm}$ & $4.73\times 10^{16}{\rm cm}$\\
\end{tabular}
\caption{Periastron precession for the star $S14$ in the central arcsecond
of the galactic centre, using the exact analytic formula
(\ref{deltaequatopericentre}),(\ref{RelativiPeriastronPrec}) for  different values of the electric
charge and the spin of the galactic black hole. We assume a central black hole mass
$M_{\rm BH}=4.06\times 10^6M_{\odot}$ and that the orbit of the  star $S14$ is a timelike non-circular equatorial Kerr-Newman geodesic.}\label{AstroS14equa4SgraA*}
\end{table}

\subsection{Periapsis advance for non-spherical polar timelike Kerr-Newman -de Sitter orbits }\label{periapsisLambdaKNpolar}

In this section we are going to derive a new closed form expression for the pericentre-shift of a test particle in a timelike non-spherical polar Kerr-Newman-de Sitter geodesic.

After one complete revolution the angular integration has to satisfy the equation:
\begin{align}
\int\frac{{\rm d}\theta}{\sqrt{\Theta^{\prime}}}&=2\int_{r_A}^{r_P}\frac{{\rm d}r}{\sqrt{R^{\prime}}}=2\frac{\sqrt{\omega}}{\sqrt{\frac{\Lambda}{3}}H}
\Biggl[-F_D\left(\frac{1}{2},\frac{1}{2},\frac{1}{2},\frac{1}{2},1,\kappa^2,\lambda^2,\mu^2\right)
\pi\nonumber\\
&+\omega F_D\left(\frac{3}{2},\frac{1}{2},\frac{1}{2},\frac{1}{2},2,\kappa^2,\lambda^2,\mu^2\right)
\frac{\pi}{2}\Biggr]
\label{balancePolar}
\end{align}
For computing the radial hyperelliptic integral in (\ref{balancePolar})\footnote{The sextic polynomial $R^{\prime}$ is obtained by setting $\mu=1$ and $L=0$ in (\ref{sextic}).} in closed analytic form in terms of Lauricella's multivariable hypergeometric function $F_D$, we apply the transformation \cite{GVKperiapsisadvLTprecession}:
\begin{equation}
z^{\prime}=\frac{\alpha_{\mu-1}-\alpha_{\mu+1}}{\alpha_{\mu}-\alpha_{\mu+1}}\frac{r-\alpha_{\mu}}{r-\alpha_{\mu-1}}
\end{equation}
The roots of the sextic radial polynomial are organised as follows:
\begin{equation}
\alpha_{\nu}>\alpha_{\mu}>\alpha_{\rho}>\alpha_i,
\end{equation}
where $\alpha_{\nu}=\alpha_{\mu-1},\alpha_{\rho}=\alpha_{\mu+1}=r_P,\alpha_{\mu}=r_A,
\alpha_i=\alpha_{\mu+i+1},i=\overline{1,3}$.
Also we define:
\begin{equation}
H\equiv\sqrt{(\alpha_{\mu}-\alpha_{\mu+1})(\alpha_{\mu}-\alpha_{\mu+2})(\alpha_{\mu}-
\alpha_{\mu+3})(\alpha_{\mu}-\alpha_{\mu+4})},
\end{equation}
\begin{equation}
\omega:=\frac{\alpha_{\mu}-\alpha_{\mu+1}}{\alpha_{\mu-1}-\alpha_{\mu+1}}.
\end{equation}

The integral on the left of Eqn.(\ref{balancePolar}) is an elliptic integral of the form:
\begin{equation}
\int\frac{{\rm d}\theta}{\sqrt{\Theta^{\prime}}}=
-\frac{1}{2}\int\frac{{\rm d}z}{\sqrt{z}\sqrt{a^4\frac{\Lambda}{3}(z-z_{\Lambda})(z-z_+)(z-z_-)}},
\label{AbelUmkehr}
\end{equation}

Inverting the elliptic integral for $z$ we obtain:
\begin{align}
z=\frac{-\beta_1}{\omega_1 {\rm sn}^2\left(2\frac{\sqrt{\omega}}{\sqrt{\frac{\Lambda}{3}}H}
\Biggl[-F_D\left(\frac{1}{2},\mbox{\boldmath${\beta}$},1,\mathbf{x}\right)
\pi+\omega F_D\left(\frac{3}{2},\mbox{\boldmath${\beta}$},2,\mathbf{x}\right)
\frac{\pi}{2}\Biggr]
\frac{\sqrt{\omega_1}\sqrt{\delta_1-\beta_1}\sqrt{\alpha_1-\beta_1}2a^2}{2\omega_1}\sqrt{\frac{\Lambda}{3}},\;\varkappa^2\right)-1}.
\end{align}
Equivalently the change in latitude after a complete radial oscillation leads to the following exact novel expression for the periastron advance for a test particle in an non-spherical polar Kerr-Newman-de Sitter orbit:
\begin{align}
&\theta=\arccos{\pm\sqrt{z}}=\nonumber\\
&\cos^{-1}\left(\pm \sqrt{\frac{-\beta_1}{\omega_1 {\rm sn}^2\left(2\frac{\sqrt{\omega}}{\sqrt{\frac{\Lambda}{3}}H}
\Biggl[-F_D\left(\frac{1}{2},\mbox{\boldmath${\beta}$},1,\mathbf{x}\right)
\pi+\omega F_D\left(\frac{3}{2},\mbox{\boldmath${\beta}$},2,\mathbf{x}\right)
\frac{\pi}{2}\Biggr]
\frac{\sqrt{\delta_1-\beta_1}\sqrt{\alpha_1-\beta_1}a^2}{\sqrt{\omega_1}}\sqrt{\frac{\Lambda}{3}},\varkappa^2\right)-1}}\right),
\label{remarkablePeriapsisKNdS}
\end{align}
where:
\begin{equation}
\mbox{\boldmath${\beta}$}\equiv\left(\frac{1}{2},\frac{1}{2},\frac{1}{2}\right),
\;\;\mathbf{x}\equiv\left(\kappa^2,\lambda^2,\mu^2\right).
\end{equation}
Also the Jacobi modulus $\varkappa$ of the Jacobi's sinus amplitudinous elliptic function in formula (\ref{remarkablePeriapsisKNdS}), for the periapsis advance that a non-spherical polar orbit undergoes in the Kerr-Newman-de Sitter spacetime, is given in terms of the roots of the angular elliptic integral by:
\begin{equation}
\varkappa^2=\frac{\omega_1 \delta_1}{\delta_1-\beta_1}=\frac{\alpha_1-\beta_1}{\alpha_1}\frac{\delta_1}{\delta_1-\beta_1}.
\end{equation}
The roots $z_{\Lambda},z_+,z_-$ appearing in (\ref{AbelUmkehr}) are roots of the polynomial equation:
\begin{align}
&z^3(\frac{a^4\Lambda}{3})-z^2\frac{a^2\Lambda}{3}[Q+(L-aE)^2\Xi^2]+a^2\Xi z^2(1-E^2\Xi)\nonumber\\
&-z\{[Q+(L-aE)^2\Xi^2]\Xi+a^2+2aE\Xi^2(L-aE)\}+Q=0,
\label{latitudinalroots}
\end{align}
after setting \footnote{We have the correspondence $\alpha_1=z_+,\beta_1=z_-,\delta_1=z_{\Lambda}$.} $L=0$.
Also the variables of the hypergeometric function $F_D$ are:

\begin{align}
\kappa^2&=\omega\frac{\alpha_{\mu-1}-\alpha_{\mu+2}}{\alpha_{\mu}-\alpha_{\mu+2}}=
\frac{\alpha-\beta}{r_{\Lambda}^1-\beta}\frac{r_{\Lambda}^1-\gamma}{\alpha-\gamma}\nonumber \\
\lambda^2&=\omega\frac{\alpha_{\mu-1}-\alpha_{\mu+3}}{\alpha_{\mu}-\alpha_{\mu+3}}=
\frac{\alpha-\beta}{r_{\Lambda}^1-\beta}\frac{r_{\Lambda}^1-\delta}{\alpha-\delta}\\
\mu^2&=\omega\frac{\alpha_{\mu-1}-\alpha_{\mu+4}}{\alpha_{\mu}-\alpha_{\mu+4}}=
\frac{\alpha-\beta}{r_{\Lambda}^1-\beta}\frac{r_{\Lambda}^1-r_{\Lambda}^2}{\alpha-r_{\Lambda}^2}\nonumber
\end{align}

\subsection{Computation of first integrals for spherical timelike geodesics in Kerr-Newman spacetime}

The equations determining the timelike orbits of constant radius in KN spacetime are:
\begin{align}
R&=r^4+(2Mr-e^2)(\eta_Q+(\xi-a)^2)+r^2(a^2-\xi^2-\eta_Q)-a^2\eta_Q-r^2\frac{\Delta^{KN}}{E^2}=0,
\label{crad1}\\
\frac{\partial R}{\partial r}&=4r^3+2M(\eta_Q+(\xi-a)^2)+2r(a^2-\xi^2-\eta_Q)-2r\frac{\Delta^{KN}}{E^2}-\frac{r^2}{E^2}(2r-2M)=0,
\label{consrad2}
\end{align}
where we define:
\begin{equation}
\xi\equiv L/E,\;\;\eta_Q\equiv \frac{Q}{E^2}.
\end{equation}
Equations (\ref{crad1})-(\ref{consrad2}) can be combined to give:
\begin{align}
&3r^4+a^2(r^2-\frac{re^2}{M})-\eta_Q(r^2-a^2-\frac{re^2}{M})-\frac{r^2}{E^2}(3r^2-4rM+a^2+e^2)
\nonumber\\
&+\frac{re^2}{ME^2}(2r^2-3Mr+a^2+e^2)-\frac{2e^2r^3}{M}=(r^2-\frac{e^2r}{M})\xi^2,\\
&r^4+\eta_Q(a^2-Mr+\frac{e^2r}{M})-\frac{r^2}{E^2}r(r-M)-Mra^2\nonumber\\
&+\frac{re^2}{ME^2}(2r-3Mr+a^2+e^2)-\frac{2e^2r^3}{M}=Mr(\xi^2-2a\xi)-\frac{e^2r}{M}\xi^2.
\end{align}

These equations can be solved for $\xi$ and $\eta_Q$. Thus, eliminating $\eta_Q$ between them, we obtain:
\begin{align}
&a^2(r-M)\xi^2-2aM(r^2-a^2-\frac{e^2 r}{M})\xi\nonumber \\
&-\left\{(r^2+a^2)[r(r^2+a^2)-M(3r^2-a^2)+2e^2 r-\frac{a^2e^2}{M}]+
\frac{2a^2e^2r}{M}(r-\frac{e^2}{M})+\frac{a^2e^4}{M^2}-\frac{(\Delta^{KN})^2}{E^2}\right\}
=0.
\end{align}
We find that the solution of this quadratic equation is given by:
\begin{equation}
\xi=\frac{M(r^2-a^2-\frac{e^2r}{M})\pm r\Delta^{KN}\sqrt{1-\left(1-\frac{M}{r}\right)\frac{1}{E^2}-
\frac{(e^2-2rM)e^2(r-M)a^2}{rM^2(\Delta^{KN})^2}-
\frac{(r-M)(r^2+a^2)a^2e^2}{M (\Delta^{KN})^2r^2}}}{a(r-M)}.
\label{MotionsphericalXI}
\end{equation}
The parameter $\eta_Q$ is then determined from equation:
\begin{align}
&-\eta_Q\left(r^2-a^2-\frac{re^2}{M}\right)=-3r^4-a^2\left(r^2-\frac{re^2}{M}\right)
+\frac{r^2}{E^2}(3r^2-4rM+a^2+e^2)\nonumber\\
&-\frac{re^2}{ME^2}(2r^2-3Mr+a^2+e^2)+\frac{2e^2r^3}{M}+
\left(r^2-\frac{e^2 r}{M}\right)\xi^2.
\label{MotionZweiCarterSpherical}
\end{align}

For zero electric charge $e=0$, Equations (\ref{MotionsphericalXI}) and (\ref{MotionZweiCarterSpherical}) reduce correctly to the corresponding equations for the first integrals of motion in Kerr spacetime \cite{Chandrasekhar}:
\begin{align}
\xi&=\frac{M(r^2-a^2)\pm r\Delta\sqrt{(1-\frac{1}{E^2}(1-\frac{M}{r}))}}{a(r-M)},\\
\eta_Q a^2(r-M)&=\frac{r^3}{M}[4a^2M-r(r-3M)^2]-\frac{2r^3M}{r-M}
\Delta[1\pm\sqrt{1-\frac{1}{E^2}(1-\frac{M}{r})}]\nonumber\\
&+\frac{r^2}{E^2}[r(r-2M)^2-a^2M].
\end{align}

\subsection{The apparent impact factor for more general orbits}\label{prosptosiparagon}
The apparent impact parameter $\Phi$ for the Kerr-Newman-(anti) de Sitter black hole can also be computed in the case in which the considered orbits depart from the equatorial plane and therefore $\theta\not=\pi/2$. Again, we compute this quantity from the $k^{\mu}k_{\mu}=0$ relation just taking into account its maximum character, i.e., that $k^r=0$. Our calculation yields:
\begin{equation}
\Phi_{\gamma}=\frac{-[a\Xi^2(r^2+a^2)-a\Xi^2\Delta_r^{KN}]\pm\sqrt{\Xi^2\Delta_r^{KN}[\Xi^2 r^4+\mathcal{Q}_{\gamma}(a^2-\Delta_r^{KN})]}}{-a^2\Xi^2+\Xi^2\Delta_r^{KN}}.
\label{apparentimpPKNadS}
\end{equation}
Our exact expression  (\ref{apparentimpPKNadS}) for the apparent impact parameter in the KN(a)dS spacetime, for zero cosmological constant ($\Lambda=0$) and zero electric charge ($e=0$), reduces to eqn.(59) (the apparent impact parameter for the Kerr black hole) in \cite{HERRERA}.
Also for zero value for Carter's constant $\mathcal{Q}_{\gamma}$ equation (\ref{apparentimpPKNadS}) reduces to eqn.(\ref{phiparagontas}).

\section{Conclusions}

In this work using the Killing-vector formalism and the associated first integrals we computed the redshift and blueshift of photons that are emitted by geodesic massive particles and travel along null geodesics towards a distant observer-located at a finite distance from the KN(a)dS black hole.
As a concrete example we calculated analytically the redshift and blueshift experienced by photons emitted by massive objects orbiting the Kerr-Newman-(anti) de Sitter black hole in equatorial and circular orbits, and following null geodesics towards a distant observer.

In addition and extending previous results in the literature we calculated in closed analytic form firstly, the frame-dragging that experience test particles in non-equatorial spherical timelike orbits in KN and KNdS spacetimes in terms of generalised hypergeometric functions of Appell and Lauricella.  We also derived new exact results for the frame-dragging, pericentre-shift and orbital period for timelike non-spherical polar geodesics in Kerr-Newman spacetime and applied them for the computation of the corresponding relativistic effects for the orbits of stars S2 and S14 in the central arcsecond of SgrA*, assuming the Galactic centre supermassive black hole is a Kerr-Newman black hole for various values of the Kerr parameter and electric charge. Secondly, we computed in closed analytic the periapsis advance for timelike non-spherical polar orbits in Kerr-Newman and Kerr-Newman de Sitter spacetimes. In the Kerr-Newman case, the pericentre-shift is expressed in terms of Jacobi's amplitude function and Gau$\ss$ hypergeometric function, while in the Kerr-Newman-de Sitter the periapsis-shift is expressed in an elegant way in terms of Jacobi's sinus amplitudinus elliptic function $\rm {sn}$ and Lauricella's hypergeometric function $F_D$ with three-variables.

We also computed the first integrals of motion for non-equatorial Kerr-Newman and Kerr-Newman-de Sitter geodesics of constant radius. We achieved that by solving the conditions for timelike spherical orbits in KNdS and KN spacetimes. We derived new elegant compact forms for the parameters (constants of motion) of these orbits. Our results are culminated in Theorems \ref{EndiamesoORO} and \ref{OroTheoLambda}. These expressions together with the analytic equation for the apparent impact factor we derived in this work-eqn (\ref{apparentimpPKNadS}), can be used to derive closed form expressions for the redshift/blueshift of the emitted photons from test particles in such non-equatorial constant radius orbits in Kerr-Newman and Kerr-Newman (anti) de Sitter spacetimes. A thorough analysis  will be a task for the future \footnote{Another interesting application of the theory of the frequency shift of the photons developed in this paper, is the possibility to focus on the measurements of the special class of the principal null congruences (PNC) photons as was shown for PNC photons emitted from sources on equatorial circular orbits around a Kerr naked singularity in \cite{KerrNakedZS}. The significance of the PNC photon trajectories lies in the fact that they mold themselves to the spacetime curvature in such a way that, if $C_{\alpha\beta\gamma\delta}$ is the Weyl conformal tensor and $^*C_{\alpha\beta\gamma\delta}$ is its dual, then $C_{\alpha\beta\gamma[\delta}k_{\epsilon]}k^{\beta}k^{\gamma}=0$ and $^*C_{\alpha\beta\gamma[\delta}k_{\epsilon]}k^{\beta}k^{\gamma}=0$ \cite{MTWschwarzBuch}.}. Such a future endeavour will also involve the computation of the redshift/blueshift of the emitted photons for realistic values of the first integrals of motion associated with the observed orbits of S-stars (the emitters) such as those of Section \ref{periapsisKN}, especially when the first measurements of the pericentre-shift of S2 will take place.
The ultimate aim of course is to determine in a consistent way the parameters of the supermassive black hole that resides at the Galactic centre region SgrA*.

It will also be interesting to investigate the effect of a massive scalar field on the orbit of $S2$   star and in particular on its redshift and periapsis advance by combining  the results of this work and the exact solutions of the Klein-Gordon-Fock (KGF) equation on the KN(a)dS and KN black hole backgrounds in terms of  Heun functions produced in \cite{Kraniotis1} (see also \cite{KraniotisDirac}) . This research will be the theme of a future publication \footnote{An initial study of such hypothetical scalar effects  has been performed  by Gravity Collaboration for the Kerr background and in solving approximately the KGF equation for the case in which the Compton wavelength of the scalar field is much larger than the gravitational radius of the black hole \cite{scalar}.}.

The fruitful synergy of theory and experiment in this fascinating research field will lead to the identification of the resident of the Milky Way's Galactic centre region and will provide an important test of General Relativity at the strong field regime.

\begin{acknowledgement}
I thank Dr. G. Kakarantzas for discussions. During the last stages of this work, the research was co-financed by Greece and the
European Union (European Social Fund - ESF) through
the Operational Programme "Human Resources Development,
Education and Lifelong Learning 2014-2020"
in the context of the project "Scalar fields in Curved
Spacetimes: Soliton Solutions, Observational Results and
Gravitational Waves" (MIS 5047648). I thank Prof. Z.Stuchl\'{\i}k  and Dr. P. Slan\'{y} for useful correspondence. I am grateful to the referees for their very constructive comments.
\end{acknowledgement}

\appendix

\section{Linking the equatorial circular radii for emitter \& detector}\label{app}

In this Appendix and for the case of a Kerr-Newman black hole we will show how the radii of the emitter and detector can be linked for circular equatorial orbits through the constants of motion $L_{\gamma},E_{\gamma}$.  The procedure was followed in \cite{HERRERA} for the Kerr black hole. Indeed, from the fact that the first integrals of motion $L_{\gamma},E_{\gamma}$ and hence the apparent impact parameter $\Phi$ are preserved along the whole trajectory followed by the photons, the latter quantity is the same when evaluated either at the emitter or detector positions, rendering the following relation $\Phi_e=\Phi_d$. For circular and equatorial orbits the maximised impact parameter for a Kerr-Newman black hole is given by the expression:
\begin{equation}
\Phi_{\gamma}=\frac{-a(2Mr-e^2)\pm r^2\sqrt{r^2+a^2+e^2-2Mr}}{r^2+e^2-2Mr}.
\end{equation}
Its preservation provides the following equation that must be satisfied by the radius of a circular equatorial observer:
\begin{align}
&(r^2-2Mr+e^2)\nonumber \\
&\times(r^4+r^2(a^2-\Phi_e)+2Mr(\Phi_e-a)^2-e^2(\Phi_e-a)^2)=0
\label{ferrarikraniotis}
\end{align}
The roots of the quartic equation in (\ref{ferrarikraniotis}) can be obtained either by the Ferrari algorithm or in a more elegant way by the use of the Weierstra$\ss$ functions making use of the addition theorem for points on the elliptic curve \cite{GRGKRANIOTIS}:
\begin{align}
\alpha & =\frac{1}{2}\frac{\wp^{\prime}(-x_{1}/2+\omega)-\wp^{\prime}(x_{1}%
)}{\wp(-x_{1}/2+\omega)-\wp(x_{1})},\label{maxweierstrass}\\
\beta & =\frac{1}{2}\frac{\wp^{\prime}(-x_{1}/2+\omega+\omega^{\prime}%
)-\wp^{\prime}(x_{1})}{\wp(-x_{1}/2+\omega+\omega^{\prime})-\wp(x_{1})},\label{weirg3}\\
\gamma & =\frac{1}{2}\frac{\wp^{\prime}(-x_{1}/2+\omega^{\prime})-\wp^{\prime
}(x_{1})}{\wp(-x_{1}/2+\omega^{\prime})-\wp(x_{1})}\label{weirstathec},\\
\delta & =\frac{1}{2}\frac{\wp^{\prime}(-x_{1}/2)-\wp^{\prime}(x_{1})}%
{\wp(-x_{1}/2)-\wp(x_{1})}\label{fourth},
\end{align}
where the point $x_{1}$ is defined by the equation:
\begin{equation}
a^{2}-\Phi^{2}=-6\wp(x_{1}),
\end{equation}
and $\omega,\omega^{\prime}$ denotes the half-periods of the
elliptic function $\wp$. The equations
\begin{equation}
2(a-\Phi)^{2}=4\wp^{\prime}(x_1),-3\wp^2(x_1)+g_2=-e^2(\Phi-a)^2
\end{equation}
determine the Weierstra$\ss$  invariants $(g_2,g_3)$ with the
result:
\begin{align}
g_{2}  & =\frac{1}{12}(a^{2}-\Phi^{2})^{2}-e^{2}(\Phi-a)^{2},\\
g_{3}  & =-\frac{1}{216}(a^{2}-\Phi^{2})^{3}-\frac{1}{4}(a-\Phi)^{4}%
-e^{2}(\Phi-a)^{2}\left(  \frac{a^{2}-\Phi^{2}}{6}\right)  .
\end{align}
The maximum root provides the circular radius of the detector.
For zero electric charge, $e=0$ (i.e. Kerr black hole) the detector radius reduces to the result obtained in \cite{HERRERA}:
\begin{align}
r_d&=\sqrt{\frac{\Phi_e-a}{3}}\Biggl[\Biggl(-\sqrt{27M^2(\Phi_e-a)}\nonumber\\
&+\sqrt{27M^2(\Phi_e-a)+(\Phi_e+a)^3}\Biggr)^{\frac{1}{3}}\nonumber\\
&+(\Phi_e+a)\Biggl(-\sqrt{27M^2(\Phi_e-a)}\nonumber\\
&+\sqrt{27M^2(\Phi_e-a)+(\Phi_e+a)^3}\Biggr)^{-\frac{1}{3}}\Biggr].
\end{align}
%
%




\end{document}